\documentclass[10pt, draftclsnofoot, journal, onecolumn]{IEEEtran}

\usepackage[export]{adjustbox}  %Tight frame around image
\usepackage{algorithm} %Für algorithm Umgebung
\usepackage[noend]{algpseudocode} %Für algorithmicx Befehle
\usepackage{amssymb, amsmath, amsthm} %Für Mathe Symbole und Theoreme

\usepackage[english]{babel}
\usepackage{bbm} %For double-stroke math symbols (indicator fct and so on)
\usepackage[url=false, doi=false, eprint=false, isbn=false,style=ieee, backend=biber]{biblatex}
\usepackage{booktabs} %For nicer tables

\usepackage{caption}
\usepackage[aboveskip=-3pt]{subcaption}
\usepackage[babel]{csquotes}

\usepackage{glossaries} %Für Abkürzungen

\usepackage{hyperref} %Für Querverweise

\usepackage{mathtools} %Für \coloneqq
\usepackage{multirow} % For merging rows in latex tables

\usepackage{siunitx} %For SI units
\usepackage{subcaption} %For subfigures
\usepackage{stackengine}

\usepackage{tikz}

\usepackage{upgreek}

\definecolor{1b}{HTML}{005AA9} %The default SPG color
\definecolor{9c}{HTML}{B90F22} %Some dark red
\definecolor{7b}{HTML}{F5A300} %Orange
\definecolor{4d}{HTML}{6A8B22} %Dark Green

\newcommand{\dashfrac}[2]{%
	{\sbox0{$\genfrac{}{}{0pt}{0}{#1}{#2}$}%
		\vphantom{\copy0}%
		\ooalign{%
			\hidewidth
			$\vcenter{\moveright\nulldelimiterspace
				\hbox to\wd0{%
					\xleaders\hbox{\kern.5pt\vrule height 0.4pt width 1.5pt\kern.5pt}\hfill
					\kern-1.5pt
				}%
			}$
			\hidewidth\cr
			\box0\cr}}%
}

\setstackEOL{\\}
\newtheorem{ass}{Assumption}%[section]

\newtheorem{lemma}{Lemma}
\newtheorem{theo}{Theorem}

\newtheorem{prop}{Proposition}

\usetikzlibrary{dsp, positioning}

\tikzstyle{block} = [rectangle, anchor = center, fill=TUDa-5b, 
text width=2cm, text centered, minimum height=1cm, minimum width = 2.1cm, text = black]
\tikzstyle{output} = [rectangle, anchor = center, fill=TUDa-5b, 
text width=2cm, text centered, minimum height=1cm, minimum width = 2.1cm, text = black]
\tikzstyle{input} = [rectangle, anchor = center, fill=TUDa-5b, 
text width=2cm, text centered, minimum height=1cm, minimum width = 2.1cm, text = black]

\tikzstyle{line} = [draw, -latex', color=black]

\tikzstyle{fcbox} = [rectangle, draw, minimum height = 2.7cm, minimum width = 4.68cm, dashed, color = black]
\tikzstyle{sensorbox} = [rectangle, draw, minimum height =2.7cm, minimum width = 2.4cm, dashed, color=black]

\DeclareSourcemap{
	\maps[datatype=bibtex, overwrite]{
		\map{
			\step[fieldset=editor, null]
		}
	}
}

\expandafter\def\expandafter\normalsize\expandafter{%
	\normalsize
	\setlength\abovedisplayskip{4.7pt}
	\setlength\belowdisplayskip{4.7pt}
	\setlength\abovedisplayshortskip{4.7pt}
	\setlength\belowdisplayshortskip{4.7pt}
}

\newacronym{awgn}{AWGN}{additive white Gaussian noise}
\newacronym{aic}{AIC}{Akaike information criterion}	
\newacronym{acf}{ACF}{autocorrelation function}

\newacronym{bfdr}{BFDR}{Bayesian false discovery rate}
\newacronym{bic}{BIC}{Bayesian information criterion}
\newacronym{bc}{BC}{Bonferroni correction}
\newacronym{bh}{BH}{Benjamini-Hochberg procedure}
\newacronym{blue}{BLUE}{best linear unbiased estimator}
\newacronym{bm}{BM}{beta distribution mixture}
\newacronym{bs}{BS}{base station}
\newacronym{bum}{BUM}{beta-and-uniform mixture}

\newacronym{cdf}{CDF}{cumulative distribution function}
\newacronym{clfdr}{clfdr}{contextual \gls{lfdr}}
\newacronym{cr}{CR}{cognitive radio}

\newacronym{doa}{DoA}{direction of arrival}	
\newacronym{dbh}{dBH}{distributed Benjamini-Hochberg}
\newacronym{dp}{DP}{detection proportion}
\newacronym{dr}{DR}{detection}
\newacronym{dt}{DT}{Delaunay triangulation}
\newacronym{dvbt}{DVB-T}{digital video broadcasting - terrestrial}

\newacronym{edf}{EDF}{empirical distribution function}
\newacronym{em}{EM}{expectation-maximization}

\newacronym{fc}{FC}{fusion center}
\newacronym{fcr}{FCR}{false cluster rate}
\newacronym{ffnn}{FFNN}{feedforward neural network}
\newacronym{fdp}{FDP}{false discovery proportion}
\newacronym{fdr}{FDR}{false discovery rate}
\newacronym{fdrs}{FDRS}{false discovery rate smoothing}
\newacronym{fmri}{fMRI}{functional magnetic resonance imaging}
\newacronym{fwer}{FWER}{family-wise error rate}

\newacronym{glr}{GLR}{generalized likelihood-ratio}
\newacronym{glrt}{GLRT}{generalized likelihood-ratio test}
\newacronym{gmm}{GMM}{Gaussian mixture model}

\newacronym{hmrf}{HMRF}{hidden Markov random field}

\newacronym{idw}{IDW}{inverse distance weighting}
\newacronym{iot}{IoT}{Internet of Things}
\newacronym{itc}{ITC}{Information Theoretic Criteria}

\newacronym{jsd}{JS}{Jensen–Shannon}

\newacronym{kkf}{KKF}{kriged Kalman filtering}
\newacronym{kld}{KL}{Kullback-Leibler}
\newacronym{ksd}{KS}{Kolmogorov-Smirnov}

\newacronym{lasso}{LASSO}{least absolute shrinkage and selection operator}
\newacronym{lr}{LR}{likelihood ratio}
\newacronym{llr}{LLR}{log-likelihood ratio}
\newacronym{lfdr}{lfdr}{local false discovery rate}
\newacronym{lm}{LM}{Lindsey's method}
\newacronym{los}{LOS}{line-of-sight}

\newacronym{mc}{MC}{Monte Carlo}
\newacronym{mcmc}{MCMC}{Markov chain Monte Carlo}
\newacronym{mfdr}{mFDR}{marginal false discovery rate}
\newacronym{mht}{MHT}{multiple hypothesis testing}
\newacronym{mle}{MLE}{maximum likelihood estimator}
\newacronym{mom}{MoM}{method of moments}
\newacronym{mse}{MSE}{mean squared error}

\newacronym{np}{NP}{Neyman-Pearson}
\newacronym{nusm}{NUSM}{nonuniform sampling method}

\newacronym{ofdr}{OFDR}{online false discovery rate}
\newacronym{ok}{OK}{ordinary Kriging}
\newacronym{ols}{OLS}{ordinary least squares}

\newacronym{pdf}{PDF}{probability density function}
\newacronym{pr}{PR}{predictive recursion}
\newacronym{psd}{PSD}{power spectral density}
\newacronym{pu}{PU}{primary user}

\newacronym{rbf}{RBF}{radial basis function}
\newacronym{rem}{REM}{radio environmental map}
\newacronym{rf}{RF}{radio field}
\newacronym{rmse}{RMSE}{root-mean-square error}
\newacronym{rv}{RV}{random variable}

\newacronym{sfdr}{sFDR}{smoothed false discovery rate}
\newacronym{sc}{SC}{spectrum cartography}
\newacronym{sk}{SK}{simple Kriging}
\newacronym{smom}{sMoM}{spectral method of moments}
\newacronym{snr}{SNR}{signal-to-noise ratio}
\newacronym{sos}{SOS}{sum of sinusoids}
\newacronym{su}{SU}{secondary user}
\newacronym{svd}{SVD}{singular value decomposition}

\newacronym{tps}{TPS}{thin-plate splines}

\newacronym{ue}{UE}{user equipment}
\newacronym{uk}{UK}{universal Kriging}
\newacronym{umts}{UMTS}{universal mobile telecommunications system}

\newacronym{wfdr}{wFDR}{weighted false discovery rate}
\newacronym{wlog}{w.l.o.g}{without loss of generality}
\newacronym{wls}{WLS}{weighted least-squares}
\newacronym{wsd}{WS}{Wasserstein}
\newacronym{wsn}{WSN}{wireless sensor network}
\newacronym{wss}{WSS}{wide-sense stationary}

\newcommand{\altLocStPdf}[1][]{f_{\locStRV|\HAlt}\!(\locSt[#1])}
\newcommand{\altPPdf}[1][]{f_{\pRV|\HAlt}\!(\p[#1])}

\newcommand{\altPPdfEst}[1][]{\hat{f}_{\pRV|\HAlt}\!(\p[#1])}
\newcommand{\altZPdf}[1][]{f_{\zRV|\HAlt}\!(\z[#1])}
\newcommand{\altLocStPdfHat}[1][]{\hat{f}_{\locStRV|\HAlt}\!(\locSt[#1])}
\newcommand{\altZPdfHat}[1][]{\hat{f}_{\zRV|\HAlt}\!(\z[#1])}

\newcommand{\altLocStPdfLoc}[2][\locStPxRV]{f_{#1|\HAlt}(#2)}
\newcommand{\altReg}{\mathcal{H}_1}
\newcommand{\altRegEst}{\hat{\mathcal{H}}_1}

\newcommand{\avVar}{\overline{\cum{2}}}

\newcommand{\avVarEstSet}[1]{\hat{\overline{\kappa}}_{2_{#1}}}
\newcommand{\avThrdOrdCenMom}{\overline{\boldsymbol{\kappa}_3}}
\newcommand{\avElCmpThrdCenMom}[1][\vecElIdx]{\overline{\cum{3_{#1}}}}

\newcommand{\B}{\mathbf{B}}
\newcommand{\baseVec}[1][\vecElIdx]{\mathbf{e}_{#1}}
\newcommand{\betaA}[1][\cmpIdx]{{a}^{(#1)}}

\newcommand{\betaMixA}[1][\vecElIdx]{{a}^{(\cmpIdx)}_{#1}}
\newcommand{\betaMixAEst}[1][\vecElIdx]{\hat{a}^{(\cmpIdx)}_{#1}}

\newcommand{\bfdr}[1][\rejRegLocSt]{\mathrm{BFDR}\big(#1\big)}

\newcommand{\cmpIdx}{k}
\newcommand{\cmpIdxAlt}{l}
\newcommand{\cmpVar}[1][\cmpIdx]{\cum{2}^{(#1)}}

\newcommand{\cmpWgt}[1][\cmpIdx]{{w}^{(#1)}}
\newcommand{\cmpWgtEst}[1][\cmpIdx]{\hat{w}^{(#1)}}
\newcommand{\cmpWgtVec}{\mathbf{w}}
\newcommand{\cmpWgtVecEst}{\hat{\mathbf{w}}}
\newcommand{\crd}[1][]{\mathbf{c}_{#1}}

\newcommand{\crdPx}[1][\pxIdx]{\crd[#1]}
\newcommand{\crdSen}[1][\senIdx]{\crd[#1]}
\newcommand{\cum}[1]{\kappa_{#1}}

\newcommand{\disMsr}{\ell}
\newcommand{\disFct}[2]{\mathrm{d}(#1, #2)}

\newcommand{\domLocSt}{\mathcal{F}}

\DeclareMathOperator{\e}{E}
\newcommand{\E}[2][]{\e_{#1}\left[#2\right]}

\newcommand{\elCmpExp}[2][\vecElIdx]{\cum{1_{#1}}^{(#2)}}
\newcommand{\elExp}[1][\vecElIdx]{\overline{\cum{1_{#1}}}}
\newcommand{\elCmpThrdCenMom}[2][\vecElIdx]{\cum{3_{#1}}^{(#2)}}

\newcommand{\ev}[1][\evIdx]{A_{#1}}

\newcommand{\evIdx}{k}

\newcommand{\expNam}[5]{#1.#2(#3)$\,-\,$\texttt{#4}.#5}
\newcommand{\expRef}[5]{\hyperref[exp:#1#2#3-#4#5]{Exp.~\expNam{#1}{#2}{#3}{#4}{#5}}}

\DeclareMathOperator{\fdr}{FDR}

\newcommand{\fdrThr}{\thr[\fdr]}
\newcommand{\firstPopSMoM}{\tilde{\mathbf{m}}_1}
\newcommand{\firstPopSMoMEst}[1]{\hat{\tilde{\mathbf{m}}}_{1_{#1}}}
\newcommand{\firstPopSMoMTheo}{\mathbf{m}_1}
\newcommand{\firstPopSMoMErr}{\mathbf{m}_1^\Delta}

\newcommand{\freBnd}{\mathcal{B}}

\newcommand{\Hyp}[1][]{H_{#1}}
\newcommand{\HypEst}[1][]{\hat{H}_{#1}}
\newcommand{\HAlt}{\mathrm{\Hyp}_{1}}
\newcommand{\HNul}{\mathrm{\Hyp}_0}

\newcommand{\HPx}{\Hyp[\pxIdx]}
\newcommand{\HPxEst}{\HypEst[\pxIdx]}
\newcommand{\HPxTime}{H_{\pxIdx, \timeIdx}}
\newcommand{\HSen}{\Hyp[\senIdx]}

\newcommand{\idc}[1]{\mathbbm{1}\!\left\{#1\right\}}

\newcommand{\lfdr}[1][\locSt]{\mathrm{lfdr}(#1)}
\newcommand{\lfdrHat}[1][\locSt]{\hat{\mathrm{lfdr}}({#1})}

\newcommand{\lfdrLoc}[1][\pxIdx]{\mathrm{lfdr}_{#1}}
\newcommand{\lfdrLocHat}[1][\pxIdx]{\hat{\mathrm{lfdr}}_{#1}}
\newcommand{\locStSetPx}{\mathcal{S}^\numPx}

\newcommand{\meanMvCmp}[2][\cmpIdx]{\cum{1_{#2}}^{(#1)}}

\newcommand{\mixLocStPdf}[1][]{f_{\locStRV}(\locSt[#1])}
\newcommand{\mixPPdf}[1][\p]{f_{\pRV}(#1)}
\newcommand{\mixLocStPdfEst}[1][]{\hat{f}_{\locStRV}(\locSt[#1])}
\newcommand{\mixPPdfBM}[1][]{{f}^\text{BM}_{\pRV}(\p[#1])}
\newcommand{\mixPPdfBMEst}[1][]{\hat{f}^\text{BM}_{\pRV}(\p[#1])}
\newcommand{\mixPPdfBUM}[1][]{{f}^\text{BUM}_{\pRV}(\p[#1])}
\newcommand{\mixPPdfBUMEst}[1][]{\hat{f}^\text{BUM}_{\pRV}(\p[#1])}
\newcommand{\mixPPdfEst}[1][]{\hat{f}_{\pRV}^{#1}(\p)}
\newcommand{\mixPPdfPM}[1][\p]{f_{\pRV}^{\boldsymbol{\theta}}(#1)}
\newcommand{\mixPPdfPMid}[1][\p]{f_{\pRV}^{\hat{\boldsymbol{\theta}}}(#1)}

\newcommand{\mixZPdf}[1][]{f_{\zRV}(\z[#1])}
\newcommand{\mixZPdfEst}[1][]{\hat{f}_{\zRV}(\z[#1])}

\newcommand{\mulPPdfMix}[1][]{f_{\pvecRV}(\pvec_{#1})}

\newcommand{\newMajor}[1]{\textcolor{black}{#1}}
\newcommand{\newNewMajor}[1]{\textcolor{black}{#1}}
\newcommand{\newMinor}[1]{\textcolor{black}{#1}}
\newcommand{\newNewMinor}[1]{\textcolor{black}{#1}}

\newcommand{\noiseEigVec}{\mathbf{v}}
\newcommand{\noiseEigVecEst}[1]{\hat{\mathbf{v}}_{#1}}
\newcommand{\noiseEigVecEl}[1][\vecElIdx]{v_{#1}}

\newcommand{\nulLocStPdf}[1][]{f_{\locStRV|\HNul}\!(\locSt[#1])}
\newcommand{\nulPPdf}[1][]{f_{\pRV|\HNul}\!(\p[#1])}
\newcommand{\nulZPdf}[1][]{f_{\zRV|\HNul}\!(\z[#1])}

\newcommand{\nulFrc}{\pi_0}
\newcommand{\nulFrcHat}{\hat{\pi}_0}

\newcommand{\nulReg}{\mathcal{H}_0}
\newcommand{\nulRegEst}{\hat{\mathcal{H}}_0}
\newcommand{\numCmp}{K}
\newcommand{\numCmpAlt}{L}
\newcommand{\numVecEl}{d}

\newcommand{\numPer}{G}
\newcommand{\numPx}{Q}

\newcommand{\numSam}{T}%Number of observation samples at the individual locations

\newcommand{\numSen}{N}

\newcommand{\numTls}{M}

\newcommand{\old}[1]{}

\newcommand{\p}[1][]{p_{#1}}
\newcommand{\pRV}[1][]{P_{#1}}
\newcommand{\perIdx}{g}

\newcommand{\pHist}{\mathrm{hist}(\pSetPx)}
\newcommand{\pPx}[1][\pxIdx]{\p[#1]}
\newcommand{\pPxRV}[1][\pxIdx]{\pRV[#1]}
\newcommand{\pSen}[1][\senIdx]{\p[#1]}

\newcommand{\pSetPx}{\mathcal{P}^\numPx}
\newcommand{\pSetSen}{\mathcal{P}^\numSen}
\newcommand{\pSubSetPx}[1][\tlIdx]{\mathcal{P}^{\numVecEl}_{#1}}

\newcommand{\pval}{$p$-value}
\newcommand{\pvec}{\mathbf{\p}}
\newcommand{\pvecRV}{\boldsymbol{p}}
\newcommand{\pvecSet}{\mathcal{P}}
\newcommand{\pvecSetOne}{\mathcal{Q}}
\newcommand{\pvecSetTwo}{\mathcal{R}}
\newcommand{\pxIdx}{q}

\newcommand{\q}[1][\vecElIdx]{q_{#1}}
\newcommand{\qvec}{\mathbf{q}}
\newcommand{\qvecRV}{\boldsymbol{q}}

\newcommand{\radius}[2]{r_{{#1},{#2}}}

\newcommand{\resSym}[1][]{r} %residuals

\newcommand{\secondPopSMoMTheoEst}{\hat{{\mathbf{M}}}_2}
\newcommand{\secondPopSMoMTheo}{\mathbf{M}_2}

\newcommand{\senIdx}{n}

\newcommand{\sigTrm}[1][\ev]{\sigSym_{\evIdx}(\tim, \freBnd)}
\newcommand{\sigSym}{x}
\newcommand{\spaSig}{x_{\old{\pxIdx}\newMinor{\senIdx}}}

\newcommand{\spaTmpNoi}{\old{n}\newMinor{\delta}_{\old{\pxIdx}\newMinor{\senIdx}}(\timeIdx)}

\newcommand{\tauEnDet}[1][\crd]{\tau_{\text{ED}}\left(\crd\right)}
\newcommand{\tensorprod}{\circ}
\newcommand{\locSt}[1][]{s_{#1}}
\newcommand{\locStRV}[1][]{S_{#1}}
\newcommand{\locStPx}[1][\pxIdx]{\locSt[#1]}
\newcommand{\locStPxRV}[1][\pxIdx]{\locStRV[#1]}
\newcommand{\thirdCenMoMMvCmp}[2][\cmpIdx]{\cum{3_{#2}}^{(#1)}}
\newcommand{\thirdPopSMoM}{\tilde{\underline{\mathbf{M}}}_3}
\newcommand{\thirdPopSMoMVec}[1]{\tilde{{\mathbf{M}}}_3(#1)}
\newcommand{\thirdPopSMoMEst}{\hat{\tilde{\underline{\mathbf{M}}}}_3}
\newcommand{\thirdPopSMoMEstVec}[1]{\hat{\tilde{{\mathbf{M}}}}_3({#1})}
\newcommand{\thirdPopSMoMTheo}{\underline{\mathbf{M}}_3}
\newcommand{\thirdPopSMoMTheoVec}[1]{{\mathbf{M}}_3({#1})}
\newcommand{\thirdPopSMoMErr}{\underline{\mathbf{M}}_3^\Delta}
\newcommand{\thirdPopSMoMErrFirst}{\underline{\mathbf{M}}_3^{\Delta_1}}
\newcommand{\thirdPopSMoMErrThird}{\underline{\mathbf{M}}_3^{\Delta_3}}
\newcommand{\thr}[1][]{\alpha_{#1}}
\newcommand{\timeIdx}{t}
\newcommand{\tlIdx}{m}

\newcommand{\U}{{\mathbf{U}}}

\newcommand{\varMvCmp}[2][\cmpIdx]{\cum{2_{#2}}^{(#1)}}
\newcommand{\vecElIdx}{i}
\newcommand{\vecElIdxAltI}{j}
\newcommand{\vecElIdxAltII}{h}
\newcommand{\vecCmpExp}[1][\cmpIdx]{\boldsymbol{\kappa}_1^{(#1)}}
\newcommand{\vecCmpExpEst}[1][\cmpIdx]{\hat{\boldsymbol{\kappa}}_1^{(#1)}}
\newcommand{\vecCmpThrdOrdCenMom}[1][\cmpIdx]{\boldsymbol{\kappa}_3^{(#1)}}
\newcommand{\vecExp}{\overline{\boldsymbol{\kappa}_1}}

\newcommand{\vecExpEstSet}[1]{{\hat{\overline{\boldsymbol{\kappa}}}_{1_{#1}}}}
\newcommand{\vecCmpCov}[1][\cmpIdx]{\mathbf{\Sigma}^{(#1)}}
\newcommand{\vecCov}{\mathbf{\Sigma}}

\newcommand{\W}{\mathbf{W}}

\newcommand{\x}[1][]{c_{\text{x}, #1}}

\newcommand{\xPx}[1][\pxIdx]{\x[#1]}
\newcommand{\xSen}[1][\senIdx]{\x[#1]}

\newcommand{\y}[1][]{c_{\text{y},#1}}

\newcommand{\yPx}[1][\pxIdx]{\y[#1]}
\newcommand{\ySen}[1][\senIdx]{\y[#1]}

\newcommand{\z}[1][]{z_{#1}}
\newcommand{\zRV}[1][]{Z_{#1}}

\newcommand{\zPxRV}[1][\pxIdx]{\zRV[#1]}
\newcommand{\zval}{$z$-score}

\definecolor{TUDa-1a}{HTML}{5D85C3}
\definecolor{TUDa-1b}{HTML}{005AA9}

\definecolor{TUDa-4b}{HTML}{99C000}
\definecolor{TUDa-5b}{HTML}{C9D400}
\definecolor{TUDa-6b}{HTML}{FDCA00}

\definecolor{TUDa-9b}{HTML}{E6001A}

\graphicspath{{./images/},{./figures/}}

\renewcommand{\newMajor}[1]{\textcolor{black}{#1}}
\renewcommand{\newMinor}[1]{\textcolor{black}{#1}}
\begin{document}
	\title{Multiple Hypothesis Testing Framework for Spatial Signals}

	\author{Martin~Gölz,~\IEEEmembership{Student Member,~IEEE,}
	        Abdelhak~M.~Zoubir,~\IEEEmembership{Fellow,~IEEE}~and~Visa~Koivunen,~\IEEEmembership{Fellow,~IEEE}%
	\thanks{This work has been submitted to the IEEE for possible publication. Copyright may be transferred without notice, after which this version may no longer be accessible. The work of M. Gölz is supported by the German Research Foundation (DFG) under grant ZO 215/17-2.}% <-this % stops a space
	\thanks{M. Gölz and A. M. Zoubir are with the Signal Processing Group, TU Darmstadt, Germany. V. Koivunen is with the Department of Signal Processing and Acoustics, Aalto University, Espoo, Finland. E-mail: \{goelz, zoubir\}@spg.tu-darmstadt.de, visa.koivunen@aalto.fi.}}

	% The paper headers
	\markboth{SUBMITTED TO IEEE~TRANSACTIONS~ON~SIGNAL~AND~INFORMATION~PROCESSING~OVER~NETWORKS}%
	{Gölz \MakeLowercase{\textit{et al.}}: TBD}
	
	% make the title area
	\maketitle

	% Making the glossaries. Split into acronyms and glossary
%	\glsaddall
%	\printglossary[type=\acronymtype]
%	\printglossary

	% As a general rule, do not put math, special symbols or citations
	% in the abstract or keywords.
	\begin{abstract}
		The problem of identifying regions of spatially interesting, different or adversarial behavior is inherent to many practical applications involving distributed multisensor systems. In this work, we develop a general framework stemming from multiple hypothesis testing to identify such regions. A discrete spatial grid is assumed for the monitored environment. The spatial grid points associated with different hypotheses are identified while controlling the false discovery rate at a pre-specified level. Measurements are acquired using a large-scale sensor network. We propose a novel, data-driven method to estimate local false discovery rates based on the spectral method of moments. Our method is agnostic to specific spatial propagation models of the underlying physical phenomenon. It relies on a broadly applicable density model for local summary statistics. In between sensors, locations are assigned  to regions associated with different hypotheses based on interpolated local false discovery rates. The benefits of our method are illustrated by applications to spatially propagating radio waves.
	\end{abstract}

	% Note that keywords are not normally used for peerreview papers.

	% The sections
	\newMajor{\section{Introduction}
\label{sec:int}
	The rapid development of ever cheaper and smaller sensors, as well as the rise of faster, lower latency and more reliable wireless connectivity standards have facilitated the deployment of large-scale \glspl{wsn}. \glspl{wsn} are a key technology in the \gls{iot} and 5G wireless systems to gather information on spatial phenomena. The terms \textit{spatial phenomenon} or \textit{spatial signal} refer to the general concept of a physical quantity of interest that is a smooth function of location \cite{Arias2018}. These occur in a large number of applications, for example in electromagnetic spectrum awareness, wireless communications, environmental monitoring, agriculture, smart buildings and acoustics. A \gls{wsn} may be composed of heterogeneous devices with different sensing capabilities. The individual nodes are commonly battery powered and operate in a congested wireless spectrum. Thus, \glspl{wsn} must communicate their local information from each sensor efficiently in terms of spectrum use and energy consumption \cite{Wang2018, Nitzan2020}.
	
	In this work, we develop a method for identifying the spatial regions of \textit{interesting}, \textit{different} or \textit{anomalous} behavior of a spatial signal using \glspl{wsn} while strictly controlling the error levels. Practical examples for such regions include areas in a city where emission levels are intolerably high, radio frequency bands are densely used/underutilized, regions where moisture is too low in agricultural fields or rooms with \newNewMinor{unusual} oxygen \newNewMinor{or temperature} levels inside a building. Due to the spatial smoothness assumption, these areas form locally continuous subregions. Fig.~\ref{fig:spa-inf} displays \newNewMinor{the  problem in a simplified manner}. % a simple overall problem description.
	\newNewMinor{We identify these interesting regions} by developing methods stemming from detection theory, in particular \gls{mht}. This allows for providing rigorous statistical performance guarantees independently of domain-specific user knowledge. Minimal assumptions on the underlying physical phenomenon are needed. \newNewMajor{	The sensors are placed sparsely in distinct locations. No particular geometry or configuration is assumed for the distributed sensors.} Hence, the proposed method is suitable for a large variety of practical applications.
	
	%While the literature on the monitoring of spatial fields via \glspl{wsn} is vast and mature,
	\newNewMinor{
		\vspace{-25pt}
		\subsection{\textbf{Limiting the consumption of resources}}
		\vspace{-3pt}
	}
	To the best of our knowledge, none of the existing approaches in the literature on monitoring of spatial fields using \glspl{wsn} solves the aforementioned problem. However, two areas of research are \newNewMinor{closely} related: \textit{field estimation} and \textit{hypothesis testing for spatial signals}. The objective in field estimation is to determine the numerical value of the observed spatial signal (or a transformation of the spatial signal) as a function of location. Field estimation methods usually assume that the sensors transmit all raw local measurements to the \gls{fc} or the other nodes \cite{Arias2018, Bosman2017}. This results in significant communication overhead and increased power consumption.
	To deal with the communication and power constraints for \glspl{wsn}, many methods for hypothesis testing for spatial signals in \glspl{wsn} limit the amount of information transmitted from each node to other nodes/the \gls{fc}. In distributed detection \cite{Viswanathan1997}, a hard (binary) decision on the local state is communicated \newNewMinor{to the \gls{fc}}. The local decisions are fused at \gls{fc} to make a decision on the overall state of the field.
	Recent works transmit more than a single bit \cite{Ray2008, Ermis2005, Ermis2006, Vempaty2014, Nurellari2016, Panigrahi2019, Ciuonzo2020}. In line with \cite{Wang2018, Fonseca2019a, Fonseca2019b}, we allow each sensor to communicate a sufficient statistic or soft decisions quantized using only few  bits. \newNewMinor{This enables the application of advanced \gls{mht} methods while requiring (at most) only one communication cycle per sensor.} %This enables the application of advanced \gls{mht} methods. %communications use the energy and spectrum efficiently.
	%The transmission of a single local soft decision statistic per sensor requires (at most) one communication cycle. This differs from a centralized scheme, where each sensor forwards all its local raw measurements to the \gls{fc}. This would require a potentially large number communication cycles per sensor.

	\begin{figure}
		\centering
		\includegraphics[scale=0.58, frame]{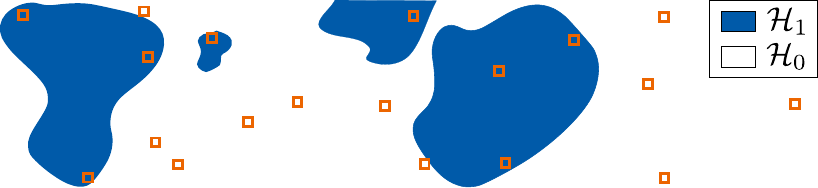}
		\caption{An exemplary spatial inference problem. Squares indicate sensors. The unknown region of interest $\altReg$ consists of \old{five}\newMinor{four} spatially continuous subregions (blue). %We discretize the entire observation area $\nulReg\cup\altReg$ by a regular grid of $\numPx$ points. The unknown binary local hypotheses $\HPx\in\{\HNul, \HAlt\}$ represent the state of the phenomenon at each grid point $\pxIdx\in[\numPx]$. $\altReg$ is the set of all grid points where $\HPx = \HAlt$ holds.
		}
		\label{fig:spa-inf}
		\vspace{-17pt}
	\end{figure}

	\newNewMajor{
		\vspace{-12pt}
		\subsection{\textbf{Enabling localized inference}}
		\vspace{-3pt}
	}
	The existing works on hypothesis testing methods for spatial signals with \glspl{wsn} focus on the detection of the presence of interesting, different or anomalous behavior of the spatial signal \textit{somewhere} within entire observation area under guarantees on the error probabilities. The only exception is \cite{Ermis2010}, where the authors also identify the sensors that observe the anomaly. In this work, we consider the more demanding problem of identifying the \textit{areas} where the spatial signal exhibits different behavior than in nominal conditions such that statistical performance guarantees in terms of Type~I error control are provided. To this end, we model the spatial area of interest as a regular spatial grid. We make a decision on the state of the observed phenomenon at each grid point. We discriminate between the nominal state of the phenomenon, represented by the null hypothesis $\HNul$, and any state that deviates from the nominal, represented by alternate hypotheses. While in many problems, one could distinguish various classes of {anomalous}, {interesting} or {different} behavior, we summarize everything that is not conform to $\HNul$ under the alternative $\HAlt$. %An \gls{mht} approach is employed in spatial inference.

	\newNewMinor{
		\vspace{-12pt}
		\subsection{\textbf{The necessity of multiple hypothesis testing}}
		\vspace{-3pt}
	}
	Depending on the size of the monitored area and the desired spatial resolution, the number of grid points and hence decisions might easily reach the order of tens of thousands. To prevent a flood of false positives resulting from testing a large number of binary hypotheses \cite{Tukey1991}, we follow the principles of \gls{mht}, where choosing the alternative $\HAlt$ is called a \textit{discovery} \cite{Soric1989}. Performance guarantees are commonly provided using the \gls{fdr} criterion. The \gls{fdr} is the expected proportion of false discoveries among all discoveries, \cite{Benjamini1995}. The past work on \gls{fdr} control in the context of spatial data has mostly focused on testing a priori \cite{Benjamini2007, Barber2015a} or a posteriori \cite{Chouldechova2014, Chumbley2009, Chumbley2010, Schwartzman2019} formed groups of data. While these procedures typically rely on assumptions that may not be realistic \cite{Eklund2016}, they do also not provide guarantees w.r.t. to the localization accuracy of the identified alternative area.
	
	\newNewMinor{
		\vspace{-12pt}
		\subsection{\textbf{Overview on the proposed inference methodology}}
		\vspace{-3pt}
	}
	For decision making with false positive control, information on the state of the field is needed at each grid point. This can be the value of a local decision statistic such as a $p$-value, $z$-score or likelihood ratio in combination with the probability model for this statistic under $\HNul$, or the local posterior probability of the null hypothesis. In this work, sensors are \newNewMinor{placed} in distinct locations at a sparse subset of grid points. Each sensor records noisy observations of the field and condenses the information into a local decision statistic. Based on the local decision statistics from each sensor and on the probability model of the local statistic under $\HNul$, we compute the probability of the null hypothesis at each grid point where a sensor is located. In particular, we propose to estimate the \gls{lfdr} \cite{Efron2001, Efron2005, Efron2008, Efron2010}, which is the empirical Bayes posterior probability of $\HNul$ at each sensor. We then exploit the spatial smoothness assumption and interpolate the \gls{lfdr}'s to make decisions at grid points in between sensors. \newNewMinor{To the best of our knowledge, this problem has not been addressed in existing works.} The main stages of the proposed spatial inference algorithm are shown in Fig.~\ref{fig:flowchart}. The decisions are made with strict control of the \gls{fdr}. This provides quantitative justification for the identified areas of different behavior \newNewMinor{and distinguishes the proposed approach from} %This makes the proposed approach clearly different from
	methods that reconstruct the field and threshold it to get a visually pleasing segmentation result without any statistical error performance guarantees \cite{Kuruvilla2016}.
	
	\begin{figure}
		\centering
%		\input{figures/flowgraph}
		% For single column!
%		\includegraphics[scale=.3]{flowgraph}
		% For double column!
		\includegraphics[scale=.25]{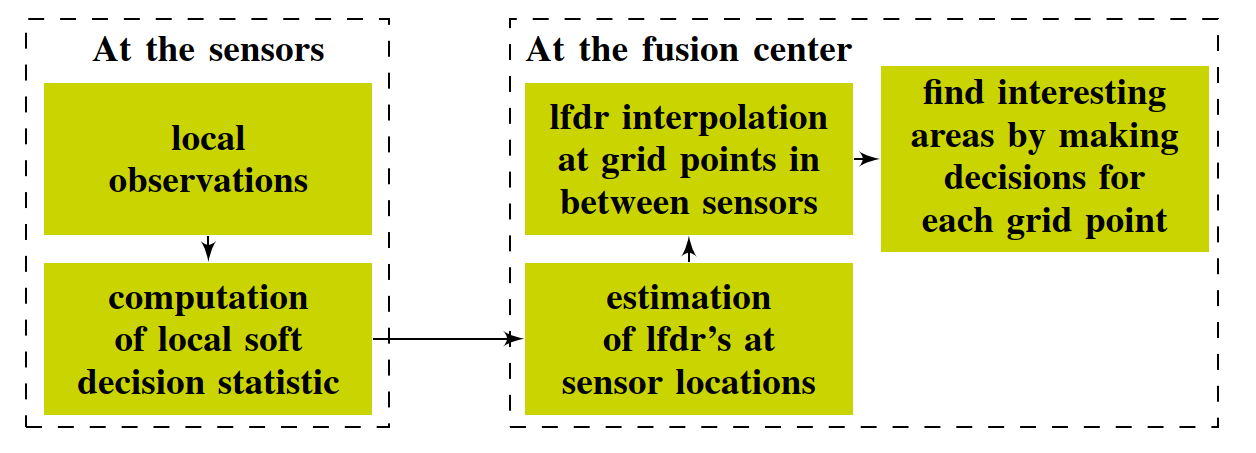}
		\caption{The proposed lfdr-based spatial inference method.}
		\label{fig:flowchart}
		\vspace{-17pt}
	\end{figure}

%	We propose a\newNewMinor{n} lfdr-based spatial inference method for identifying regions where alternative hypotheses are in place.
	\newNewMinor{
		\vspace{-12pt}
		\subsection{\textbf{Advantages of the proposed spatial inference approach}}
		\vspace{-3pt}
	}
	The local signal model under $\HNul$ may be different in each sensor and it can be learned or its parameters estimated when there is only noise present. This makes the method robust against modeling errors and suitable to a wide range of applications and operational environments. The local signal and noise models may differ, which makes the method suitable for heterogeneous \glspl{wsn}. The local information is communicated efficiently to the \gls{fc}. Our results demonstrate that the local soft decision statistics can be quantized using few bits with only a negligible performance loss. Communication overhead and power consumption can be further reduced by censoring the transmission of uninteresting local soft decision statistics. Finally, the detection power of the proposed \gls{lfdr}-based inference method can be further improved by incorporating available side information similar to \cite{Efron2010, Tansey2018, Halme2019, Goelz2019, Scott2015, Chen2019}.
	
	\newNewMinor{As an alternative approach to obtaining the local null probabilities via \gls{lfdr}'s, one could imagine a fully integrated method in which the }raw measurements or soft decision statistics would feed a model that then provided the local null probabilities at the grid points. \newNewMinor{However, such a }holistic model would depend on a variety of unknown parameters and thus be highly complex. For a radio frequency field, for example, the model would require the appropriate incorporation of the local signal model and position of each sensor, the unknown location and number of active transmitters as well as the propagation environment. In addition, such a model would also be different for each application, thereby limit\newNewMinor{ing its} general applicability. To the best of our knowledge, the two existing approaches \cite{Sun2014, Shu2015} are not applicable to identify the areas of interesting signal behavior with \glspl{wsn} while controlling decision making error levels. \cite{Sun2014} assumes a simple Gaussian random field model and the authors highlight that the method is very sensitive to model deviations. The method from \cite{Shu2015} has been shown in \cite{Tansey2018} to violate the nominal \gls{fdr} level. Also, \cite{Sun2014, Shu2015} are highly computationally complex even for moderately sized grids.
	
%	The \gls{fc} computes the \gls{lfdr}'s for all grid points. For the sensor locations, the \gls{lfdr}'s are obtained by comparing a sensor's local soft decision statistic to those from other sensors. In between sensors, the sensor \gls{lfdr}'s are interpolated. This yields the desired local probabilities of $\HNul$ much more efficiently than with a fully integrated method. In such an approach, the raw measurements or soft decision statistics would feed a model that then provided the local null probabilities at the grid points. The required overall holistic model would depend on a variety of unknown parameters and thus be highly complex. %, i.e., taking days to solve inference tasks that \gls{lfdr}-based methods solve in a few minutes \cite{Tansey2018}.

	\newNewMinor{
		\vspace{-12pt}
		\subsection{\textbf{Challenges}}
		\vspace{-3pt}
	}
	The \newNewMinor{computation of the} \gls{lfdr}'s \newNewMinor{relies }on the joint \gls{pdf} of the local soft decision statistics and the overall proportion of alternatives that are unknown in practice. These quantities need to be learned or estimated accurately from the data. This is referred to as \textit{\gls{lfdr} estimation} in the \gls{mht} literature \cite{Efron2010}. A variety of estimators exist, often assuming that the joint \gls{pdf} belongs to the exponential family \cite{Efron2010, Muralidharan2010, Martin2011, Scott2015}. As our simulation results show, the existing \gls{lfdr} estimators are not suitable for spatial inference with \glspl{wsn}. They either yield inaccurate \gls{lfdr} estimates due to too simplistic assumed data models, scale poorly with network size and can thus not be applied to large-scale \gls{wsn}s and/or yield insufficient results when the local soft decision statistics are quantized with few bits.

	\newNewMinor{
		\vspace{-12pt}
		\subsection{\textbf{The original contributions of this paper}}
		\vspace{-3pt}
	}
%	The original contributions of this paper are:
	\begin{itemize}
		\item We propose an lfdr-based spatial inference method for \glspl{wsn} as a flexible data-driven approach to determine the areas of interest, difference or anomaly of a physical phenomenon. Our proposed method scales well to spatial inference with large-scale sensor networks under strict statistical performance guarantees.
		\item We propose a novel, highly computationally efficient method for computing \gls{lfdr}'s. It bases upon an innovative mixture distribution model and the method of moments. It can deal with quantized local soft decision statistics.
	\end{itemize}}
	\newNewMinor{
	\vspace{-40pt}
	\subsection{\textbf{Notation}}
	\vspace{-7pt}
}
%\textbf{Notation:}
Throughout the paper, regular lowercase letters $x$ denote scalars, whereas bold lowercase letters $\mathbf{x}$, bold uppercase letters $\mathbf{X}$ and underlined bold uppercase letters $\underline{\mathbf{X}}$ denote vectors, matrices and third order tensors, respectively. $\mathbf{X}^\frac{1}{2}$ and $\mathbf{X}^\dagger$ denote matrix square root and Moore-Penrose inverse. Calligraphic letters $\mathcal{X}$ denote sets and $[X] \coloneqq 1, 2, \dots, X$ sets containing all positive integers $\leq X$. $|\mathcal{X}|$ is the cardinality of set $\mathcal{X}$ and $\idc{\cdot}$ denotes the indicator function. The Hadamard product operator is $\odot$, while $\tensorprod$ represents the outer product. $f_X(x)$ denotes the \gls{pdf} of \gls{rv} $X$ and $f_{X|\mathrm{A}}(x)$ its \gls{pdf} conditioned on event $\mathrm{{A}}$. \old{$\fdr$ stands for the false discovery rate, $\mathrm{lfdr}$ for the local false discovery rate.}
	%\vspace{-10pt}
\section{The spatial inference problem}
\label{sec:spa-inf}
	The spatial inference problem is illustrated in Fig.~\ref{fig:spa-inf}. $\nulReg$ and $\altReg$ denote the regions of nominal and anomalous, different or interesting behavior, respectively. The continuous observation area is discretized by a regular grid of $\numPx$ elements, to each of which we refer by its index $\pxIdx\in[\numPx]$. Their position on the grid is denoted by $\crdPx = [\xPx, \yPx]^\top$. \newMinor{Sensors are placed at $\numSen \leq\numPx$ grid points. To keep the notation simple, we use the	same ordering in the indices of sensors and grid points, i.e., the sensor $\senIdx\,\forall\,[\numSen]$ is located at grid point $\pxIdx = \senIdx$.}
	%\subsection{The local summary statistics}
	The state of the observed phenomenon at grid point $\pxIdx \in[\numPx]$ and time instant $\timeIdx\in[\numSam]$ is described by the unknown true binary local hypothesis $\HPxTime \in \{\HNul, \HAlt\}$, with the \textit{null hypothesis} $\HNul$ and the \textit{alternative} $\HAlt$. If $\HPxTime = \HNul$, we say that the observed phenomenon is in its nominal state at $\pxIdx\in[\numPx]$, $\timeIdx\in[\numSam]$. If $\HPxTime = \HAlt$, an interesting phenomenon or \textit{anomaly} is present. $\HAlt$ holds under any deviation from $\HNul$ and is thus in general \textit{composite}. In environmental monitoring, $\HNul$ could represent clean air, whereas $\HAlt$ could indicate contamination above a tolerable level.
	We consider phenomena that vary smoothly in space and slowly in time. Due to the latter, we assume that the true local hypotheses are constant over the observation period and write $\HPx = \HPxTime\,\forall\,\timeIdx\in[\numSam]$.% For now, we assume that a sensor is placed at each grid point.

	$\altReg$ and $\nulReg$ are mutually exclusive sets that comprise the grid points at which $\HNul$ or $\HAlt$ hold. Formally, the objective of spatial inference is to identify the set $\nulReg = \{\pxIdx\in[\numPx]\,|\,\HPx = \HNul\}$ of all grid points where $\HNul$ is in place, and the set $\altReg = \{\pxIdx\in[\numPx]\,|\,\HPx = \HAlt\}$ containing all locations where $\HAlt$ holds. %We cannot directly observe $\HPx,\pxIdx \in[\numPx]$.
	\newMinor{At the sensors, the hypothesis-dependent }\old{The }models for the measured field levels $y_{\old{\pxIdx}\newMinor{\senIdx}}(\timeIdx), \newMinor{\senIdx}\in[\newMinor{\numSen}],\timeIdx\in[\numSam],$ \newMinor{are}\old{differ under $\HNul$ and $\HAlt$, i.e.,} % under $\HNul$ and $\HAlt$, i.e.,
%	\begin{equation}
%		\label{eq:spa-tmp-sig-mod}
%		\spaTmpObs[\HPx] = \begin{cases}
%			\spaTmpNoi, &\quad \HPx = \HNul, \\
%			\spaSig + \spaTmpNoi, &\quad \HPx = \HAlt,
%			\
%		\end{cases}
%	\end{equation}
	\begin{align}
		\label{eq:spa-tmp-sig-mod}
			\begin{split}
			\old{\HPx}\newMinor{\HSen} = \HNul: \qquad\qquad& y_{\old{\pxIdx}\newMinor{\senIdx}}(\timeIdx) = \spaTmpNoi,\\
			\old{\HPx}\newMinor{\HSen} = \HAlt: \qquad\qquad& y_{\old{\pxIdx}\newMinor{\senIdx}}(\timeIdx) = \spaSig + \spaTmpNoi,
		\end{split}
	\end{align}
	where $\spaSig\neq 0$ is the non-zero level of the phenomenon at location \old{$\pxIdx$}\newMinor{$\senIdx$} and $\spaTmpNoi$ the temporally i.i.d. measurement noise. This noise process is spatially independent but not necessarily identically distributed in different nodes. %$\HNul$ is a simple hypothesis, whereas $\HAlt$ is composite. %, whose statistical properties in general depend on $\pxIdx\in[\numPx]$, for example, because of the deployment of different sensor types at different locations.	
	 $\spaSig$ varies with $\old{\pxIdx}\newMinor{\senIdx}\in[\old{\numPx}\newMinor{\numSen}]$, but takes on similar values at close-by locations due to the spatial smoothness assumption. This is well-justified by the underlying mechanisms of many physical phenomena. Radio waves, for example, are subject to path-loss and shadow fading \cite{Molisch2012} that vary slowly.
	 %	  Thus, $\forall\, \timeIdx\in[\numSam]$$ \HPx = \HPxTime$ and $\spaTmpSig=\spaSig$ , i.e., the true hypothesis and the field levels are constants over time. $\spaTmpNoi$, $\spaTmpSig$ and $\spaTmpObs[\HPx]$ are i.i.d. \gls{rv}s. The dependence structure in the spatial domain depends on the value of $\HPx,\pxIdx\in[\numPx]$. In general, the statistical properties of the measurement noise $\spaTmpNoi$ might vary for different $\pxIdx\in[\numPx]$. Regardless, the $\spaTmpNoi$ are spatially independently distributed. This is a reasonable assumption in practice, where $\spaTmpNoi$ represents, for example, thermal sensor noise. Consequently, $\spaTmpObs[\HPx]$ are i.i.d. $\forall\,\,\{\pxIdx\in[\numPx]|\HPx= \HNul\}, \timeIdx\in[\numSam]$. In contrast, the $\spaTmpObs[\HPx] \,\forall\,\{\pxIdx\in[\numPx]|\HPx = \HAlt\}, \timeIdx\in[\numSam]$ are spatially correlated, because $\spaTmpSig$ is a spatially correlated signal. The correlation originates in the physics of the observed phenomenon. For example, in a radio field, the signal $\spaTmpSig$ is subject to path loss and shadow fading \cite{Molisch2012}.
	 %The individual sensors condense their raw observations in local summary statistics $\locStPx = \locSt(\spaObs, \dots, \spaObs[\numSam])$ with $\locSt: \mathbb{R}^\numSam \rightarrow \mathbb{R}$, i.e., $\locStPx$ is a function of the observations at $\pxIdx\in[\numPx]$ over period $\numSam$.
	 We cannot directly observe $\HPx,\,\pxIdx \in[\numPx]$. \newMinor{However, we can exploit the model differences in Eq.~\eqref{eq:spa-tmp-sig-mod} to determine local decision statistics for each grid point $\pxIdx\in[\numPx]$ based on the measurements at each sensor $\senIdx\in[\numSen]$. Those can then be used to decide on $\HPxEst = \HNul$ or $\HPxEst = \HAlt$}\old{ Eq.~\eqref{eq:spa-tmp-sig-mod} suggests to utilize the measurements $y_{\old{\pxIdx}\newMinor{\senIdx}}(\timeIdx),\,\forall\,\timeIdx\in[\numSam],$ to obtain the empirical models for the hypotheses $\HPxEst = \HNul$ or $\HPxEst = \HAlt$} and form the estimated regions $\nulRegEst$ and $\altRegEst$ associated with null hypothesis and alternative.
	 \old{The individual sensors condense their raw observations over observation period $\numSam$ into soft local decisions statistics $\tau_\pxIdx, \pxIdx\in[\numPx]$. These are more informative than local binary decisions, used, e.g., in \cite{Marano2019}. }%Proceeding with classic hypothesis testing, one would make a decision by thresholding each $\tau_\pxIdx$. However, one could not quantify the regions' accuracy. In contrast,
	 \newMinor{The decisions are made such that}\old{A procedure that controls} the \gls{fdr} \cite{Benjamini1995}, the expected ratio between the number of false discoveries and all discoveries 
	 \begin{equation}
	 	\label{eq:fdr-def}
	 	\fdr = \E{\frac{\sum_{\pxIdx\in\nulReg}\idc{\HPxEst = \HAlt}}{\sum_{\pxIdx = 1}^\numPx\idc{\HPxEst = \HAlt}}},
	 \end{equation}
 	\newMinor{is controlled }at a nominal level $\thr$\newMinor{. This} guarantees that on average, a proportion $(1-\alpha)$ of the locations in $\altRegEst$ are actual members of the true alternative region $\altReg$. %\gls{fdr} control demands comparability of the individual decision statistics with one another.
 	
 	\newMinor{Each sensor condenses its raw measurements over the observation period $\numSam$ into a local summary statistic $\tau_{\old{\pxIdx}\newMinor{\senIdx}}, \old{\pxIdx}\newMinor{\senIdx}\in[\old{\numPx}\newMinor{\numSen}]$.}
 	The type of deployed sensor and the distribution of the measurement noise may differ from sensor to sensor \cite{Rovatsos2020}. Thus, the $\tau_{\old{\pxIdx}\newMinor{\senIdx}}$ cannot be directly fused with each other. Instead, one defines local \old{summary}\newMinor{(soft) decision} statistics $\locStPxRV[\newMinor{\senIdx}]$ that are normalized such that they are i.i.d. $\forall\,\old{\pxIdx}\newMinor{\senIdx}\in\nulReg$, but not necessarily for $\old{\pxIdx}\newMinor{\senIdx}\in\altReg$. \newMajor{Conditioned on the field level $\spaSig$, however, the $\locStPxRV[\senIdx]$ from sensors where $\HAlt$ is in place are independently distributed.}%In other words, $\locStPxRV\sim\nulLocStPdf\,\forall\pxIdx\in\nulReg$.
%	 The individual sensors condense their raw observations to the scalar value of a sensor type dependent local function $\tau_\pxIdx = g_\pxIdx(y_\pxIdx(1), \dots, y_\pxIdx(\numSam))$. The $\tau_\pxIdx$ are not necessarily identically distributed for all $\pxIdx\in\nulReg$. However, identically distributed local summary statistics under $\HNul$ are a central prerequisite to \gls{fdr} control via \gls{lfdr}'s \cite{Efron2010}. Consequently, we define the local summary statistics $\locStPx\sim\nulLocStPdf\,\forall\,\pxIdx\in\nulReg$, where $\nulLocStPdf$ denotes their \gls{pdf} under $\HNul$. The distribution of $\locStPx$ is independent of the exact location if $\pxIdx\in\nulReg$. For locations where the alternative is in place, the distribution of $\locStPx$ can differ for different locations and we write $\locStPx\sim\altLocStPdfLoc{\locSt}\,\forall\,\pxIdx\in\altReg$.
	%The set of the local summary statistics from all locations is denoted by $\locStSetPx = \{\locStPx[1], \dots, \locStPx[\numPx]\}$. %To comply with this prerequisite, since different types of sensors might summarize their observations in summary statistics $\tau_{\pxIdx}$. To facilitate the comparability between sensors, we compute \pval s
	\old{The null and alternative regions for a given set of observations $\locStSetPx = \{\locStPx[1], \dots, \locStPx[\numPx]\}$ of the local summary statistics $\locStPxRV[1], \dots, \locStPxRV[\numPx]$ can be inferred by standard methods, such as the \gls{bh} \cite{Benjamini1995}.} $\domLocSt$ denotes the domain of $\locStPxRV[\newMinor{\senIdx}]$.
	 Common choices \old{for $\locStPxRV$} are \pval s
	 \begin{equation}
	 	\locStPxRV[\newMinor{\senIdx}] \coloneqq \int_{\old{-\infty}\newMinor{\uptau_\senIdx}}^{\infty}f_{\tau_{\old{\pxIdx}\newMinor{\senIdx}}|\HNul}(\tau)\mathrm{d}\tau = \pPxRV[\newMinor{\senIdx}],
	 \end{equation}
	 $\forall\,\old{\pxIdx}\newMinor{\senIdx}\in[\old{\numPx}\newMinor{\numSen}]$, where $f_{\tau_\pxIdx|\HNul}(\old{\pxIdx}\newMinor{\tau})$ is the \gls{pdf} of $\tau_\pxIdx$ under $\old{\HPx}\newMinor{\HSen} = \HNul$ \newMinor{and $\uptau_\senIdx$ is a realization of random variable $\tau_\senIdx$}, or \zval s $\locStPxRV[\newMinor{\senIdx}] \coloneqq \Phi^{-1}(\pPxRV[\newMinor{\senIdx}]) = \zPxRV[\newMinor{\senIdx}]$, where $\Phi(\cdot)$ is the standard normal \gls{cdf} \cite{Efron2010}. For \pval s, the domain is $\domLocSt = [0, 1]$ and for \zval s, $\domLocSt = \mathbb{R}$. $f_{\tau_{\old{\pxIdx}\newMinor{\senIdx}}|\HNul}(\tau)$ has to be known to compute \pval s or \zval s. If $f_{\tau_{\old{\pxIdx}\newMinor{\senIdx}}|\HNul}(\tau)$ is unknown, it can be estimated from the data using for example the bootstrap \cite{Zoubir2001, Goelz2017a}. Small \pval s indicate \old{very }little support for \old{the null hypothesis}\newMinor{$\HNul$}.
	\old{We propose to solve the spatial inference problem with the help of a suitable \gls{mht} soft decision statistic.}\newMinor{The soft decision statistics are transmitted from each sensor to the \gls{fc} via a wireless communication channel and have to be quantized in practice. Similar to \cite{Wang2018, Fonseca2019a, Fonseca2019b}, the proposed inference method is designed assuming the availability of infinitely precise local soft decision statistics at the \gls{fc}. However, our results in Sec.~\ref{sec:sim-res} underline, that close to optimal performance is achieved even when the local soft decision statistics are quantized using only few bits. }%Traditional methods for hypothesis testing in \glspl{wsn} only transmit hard $1$ bit local decisions \cite{Viswanathan1997} to keep the communication cost low. In contrast, recent works \cite{Ray2008, Ermis2005, Ermis2006, Vempaty2014, Nurellari2016, Panigrahi2019, Ciuonzo2020} assume that a few bits can be transmitted per sensors. In \cite{Wang2018, Fonseca2019a, Fonseca2019b}, the availability of unquantized local soft decision statistics at the \gls{fc} is assumed}. %Soft decision statistics quantify the believe in the true value of $\HPx$. They are more informative than hard decisions $\HPxEst$, because they can be used to make a hard decision, whereas the $\HPxEst$ cannot be transformed back into a soft decision statistic.

	Define the random variable $\locStRV$ that represents the mixture of all local \old{summary}\newMinor{decision} statistics from the sensors. The \gls{pdf} of $\locStRV$ is
	 \begin{equation}
	 	\label{eq:theo-locst-pdf}
	 	\mixLocStPdf = \nulFrc\nulLocStPdf + \newMinor{\frac{\textcolor{black}{(1-\nulFrc)}}{\newMinor{\sum_{\senIdx\in\altReg}}1}}\sum_{\old{\pxIdx}\newMinor{\senIdx} \in\altReg}\altLocStPdfLoc[S_{\newMinor{\senIdx}}]{\locSt},
	 \end{equation}
 	with $\nulFrc = \newMinor{\sum_{\senIdx\in\nulReg}1/\numSen}$ the \newMinor{fraction of sensors located in the null region}\old{relative size of the null region}, $\nulLocStPdf$ the \gls{pdf} for $\locStPxRV[\newMinor{\senIdx}]\,\forall\,\old{\pxIdx}\newMinor{\senIdx}\in\nulReg$ and $\altLocStPdfLoc[S_{\newMinor{\senIdx}}]{\locSt}$ the \gls{pdf} for $\locStPxRV[\newMinor{\senIdx}]$ if $\newMinor{\senIdx}\old{\pxIdx}\in\altReg$. The model in Eq.~\eqref{eq:theo-locst-pdf} exploits that the local \old{summary}\newMinor{decision} statistics are i.i.d. across locations $\old{\pxIdx}\newMinor{\senIdx}\in\nulReg$. Finally, the local false discovery rate is \cite{Efron2005, Efron2010}
	\begin{equation}
		\label{eq:lfdr-def}
		\lfdr[\locSt] = \frac{\nulFrc \nulLocStPdf}{\mixLocStPdf}.
	\end{equation}
	Appealingly, $\lfdrLoc[\newMinor{\senIdx}] = \lfdr[{\locStPx[\newMinor{\senIdx}]}]$ is the posterior empirical Bayes probability that $\old{\pxIdx}\newMinor{\senIdx}\in\nulReg$. %Thus, it allows us to assess the probability of making a false discovery.
	\newMinor{We interpolate the $\lfdrLoc[\senIdx]$ to obtain $\lfdrLoc$ for each grid point $\pxIdx\in[\numPx]$ in between sensors.}	To solve the spatial inference problem while controlling $\fdr\leq\thr$, we form the region associated with the alternative hypothesis
	\begin{equation}
	\label{eq:alt-reg-est}
%		\altRegEst = \underset{\mathcal{H}\subseteq[\numPx]}{\mathrm{argmax}}\bigg\{\sum_{\pxIdx\in\mathcal{H}}\lfdr[\locStPx]\Big|\sum_{\pxIdx\in\mathcal{H}}\lfdr[\locStPx]\leq\thr\bigg\}.
		\altRegEst = \underset{\mathcal{H}\subseteq[\numPx]}{\mathrm{argmax}}\bigg\{|\mathcal{H}|:\newNewMinor{\frac{1}{|\mathcal{H}|}}\sum_{\pxIdx\in\mathcal{H}}\lfdrLoc\leq\thr\bigg\}.
	\end{equation}
	This approach guarantees \gls{fdr} control at level $\thr$ while maximizing detection power, since the so-called \gls{bfdr} $\bfdr[\altRegEst] = \sum_{\pxIdx\in\altRegEst}\newNewMinor{\frac{\textcolor{black}{\lfdrLoc}}{|\altRegEst|}}$ is an upper bound of the Frequentist \gls{fdr} from Eq.~\eqref{eq:fdr-def} \cite{Efron2010}. Note that the \gls{bfdr} is the \text{average} false discovery probability \textit{across the alternative region} $\altRegEst$, while the \gls{lfdr} asserts each location $\pxIdx\in[\numPx]$ with the \textit{individual} risk of being a false discovery. 
	
	\newMinor{A key element of our proposed lfdr-based spatial inference method for \glspl{wsn} is the estimation of the \gls{lfdr}'s. The general concept of the \gls{lfdr} and \gls{lfdr} estimation are described in detail in Sec.~\ref{sec:lfdr-est-gen}. In Sec.~\ref{sec:prop-met}, we propose a novel method for \gls{lfdr} estimation at sensor locations. In Sec.~\ref{sec:lfdr-ipl}, we propose the interpolation of the sensor \gls{lfdr}'s in between sensor locations. \newMinor{The communication cost and power consumption of the proposed lfdr-based spatial inference approach are discussed in Sec.~\ref{sec:cc}}.}
	\old{In what follows, we develop a novel method to compute the \gls{lfdr}'s. The concept of the local false discovery rate is described in Sec.~\ref{sec:lfdr-est-gen}. The proposed method is based on a sophisticated probability model for the local \old{summary}\newMinor{decision} statistics. We introduce the novel model and a novel algorithm to compute \gls{lfdr}'s from observed data in Sec.~\ref{sec:prop-met}. To determine decision statistics in between the sparse sensor locations, we propose to interpolate the sensor \gls{lfdr}'s in Sec.~\ref{sec:lfdr-ipl}. \newMinor{The communication cost and power consumption of our proposed approach are discussed in Sec.~\ref{sec:cc}}.} We conclude by the simulation results in Sec.~\ref{sec:sim-res}.

	\vspace{-15pt}
\section{Local False Discovery Rate Estimation}
\label{sec:lfdr-est-gen}
	\old{Inference based on local false discovery rates contains estimation and hypothesis testing components, since the theoretical \gls{lfdr}'s defined in Eq.~\eqref{eq:lfdr-def} are not known in practice. Thus, the null and alternative regions are determined based on \gls{lfdr} estimates.}
	\newMinor{The theoretical \gls{lfdr}'s defined in Eq.~\eqref{eq:lfdr-def} are unavailable in practice. Hence, a central component of our proposed lfdr-based spatial inference approach (Fig.~\ref{fig:flowchart}) is the estimation of the \gls{lfdr}'s.} The accuracy of the deployed estimators has immediate consequences on \gls{fdr} control: underestimation of the \gls{lfdr}'s leads to violations of the nominal \gls{fdr} level, whereas overestimation reduces \newMinor{the power of testing}\old{power}. \newMinor{In this section, we develop an estimator for the lfdrs at sensor locations where local decision statistics are available. In Sec.~\ref{sec:lfdr-ipl}, we discuss how to obtain the \gls{lfdr} estimates in between sensor locations.}
	
	The general structure of \gls{lfdr} estimators follows from Eq.~\eqref{eq:lfdr-def}. The \gls{pdf} of $\locStRV$ for $\old{\pxIdx}\newMinor{\senIdx}\in\nulReg$ is most often assumed to be known \newMinor{or reliably estimated} \cite{Efron2010}\old{.}\newMinor{, whereas}%In particular, if \pval s are deployed as local summary statistics $\locStRV = \pRV$, then $\nulLocStPdf = \nulPPdf$ is the \gls{pdf} of the uniform distribution on $\domLocSt = [0, 1]$. For \zval s $\locStRV = \zRV$, $\nulLocStPdf = \nulZPdf$ is the \gls{pdf} of the standard normal distribution. 
	%In some cases, the theoretical $\nulLocStPdf$ might not hold due to violated assumptions during the process of computing the local summary statistics and one needs to resort to so-called empirical null distributions. In this work, we focus on \pval s and \zval s as local summary statistics which, as previously outlined, can be computed entirely based on the observed data by means of empirical methods. Consequently, the theoretical $\nulLocStPdf$ is not violated.
	\old{The relative size of the null region} $\nulFrc$ \newMinor{and $\altLocStPdf$} are unknown. \newMinor{Hence, the}\old{The} mixture $\mixLocStPdf$ from Eq.~\eqref{eq:theo-locst-pdf} is not available. The common approach to \gls{lfdr} estimation relies on \newMajor{the} separate estimation of $\nulFrc$ and $\mixLocStPdf$ by estimators $\nulFrcHat$ and $\mixLocStPdfEst$, which are then plugged into Eq.~\eqref{eq:lfdr-def} \cite{Efron2010}, \old{i.e.,}
	\begin{equation}
		\label{eq:lfdr-est}
		\lfdrHat = \frac{\nulFrcHat\nulLocStPdf}{\mixLocStPdfEst}.
	\end{equation}
	The \newMinor{unknown} underlying physical \old{effects}\newMinor{phenomenon} drive\newMinor{s} the statistical behavior of the \gls{lfdr} via $\mixLocStPdf$. Thus, a generally optimal estimator $\lfdrHat$ does not exist.
	
	The increasing interest in the incorporation of covariate information into \gls{mht}, e.g. \cite{Scott2015, Tansey2018, Chen2019a, Halme2019, Goelz2019}, has lead to a number of sophisticated \gls{lfdr} estimators that treat $\mixLocStPdf$ as a two component mixture, the \old{so-called }two-groups model \cite{Efron2001, Pounds2003}
	\begin{equation}
	\label{eq:two-grp-mod}
		\mixLocStPdfEst = \nulFrcHat\nulLocStPdf + (1-\nulFrcHat)\altLocStPdfHat.
	\end{equation}
	We also adopt the two-groups model.	In spatial inference, $\altLocStPdfHat$ is an estimator for the mixture \newMinor{\gls{pdf}} of \newMinor{the} local \old{summary}\newMinor{decision} statistics in the alternative region $\altLocStPdf = \sum_{\old{\pxIdx}\newMinor{\senIdx} \in\altReg}\altLocStPdfLoc[S_{\newMinor{\senIdx}}]{\locSt}$, see Eq.~\eqref{eq:theo-locst-pdf}.

	\begin{figure}
		\centering
		\includegraphics[scale=0.45]{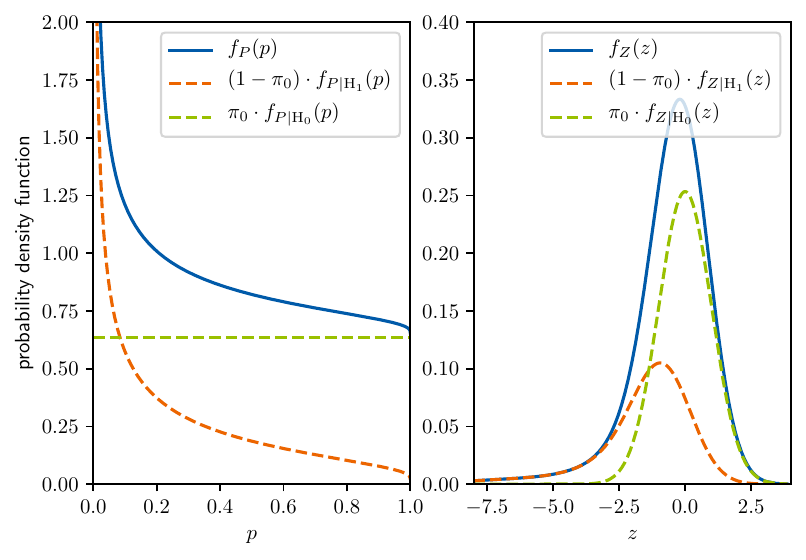}
		\caption{The true \gls{pdf}s $\mixPPdf$ and $\mixZPdf$ for \pval s and \zval s from \old{Sc.~B}\newMinor{ScB} described in Sec.~\ref{sec:sim-res}. The two-groups model from Eq.~\eqref{eq:two-grp-mod} always holds, i.e., $\mixLocStPdf$ is composed of a null and an alternative component. $\nulPPdf$ and $\nulZPdf$ are known analytically. $\nulFrc$ and $\altPPdf$ and $\altZPdf$ are unknown. In general, the alternative component for one-sided \pval s is a monotonically decreasing function. For one-sided \zval s, the alternative component exhibits a heavy left tail.}
		\label{fig:pdfs}
		\vspace{-15pt}
	\end{figure}
	\old{\subsection{The choice of \old{a}\newMinor{the} local \old{summary}\newMinor{decision} statistic}
	\label{sec:p-vs-z}}
		\old{We infer the alternative region $\altRegEst$ for a given set of realizations of the local summary statistics $\locStSetPx = \{\locStPx[1],\dots, \locStPx[\numPx]\}$ by Eq.~\eqref{eq:alt-reg-est} with $\lfdrLocHat[1], \dots, \lfdrLocHat[\numPx]$. We follow the traditional approach to \gls{lfdr} estimation. $\lfdrHat[\locSt] = \lfdrLocHat[q]$ is found from Eq.~\eqref{eq:lfdr-est} using plug-in estimates $\nulFrcHat$ and $\mixLocStPdfEst$. The latter obeys the two-groups model from Eq.~\eqref{eq:two-grp-mod}. $\nulLocStPdf$ is assumed to be known, while the behavior in the alternative region is captured by $\altLocStPdfHat$. If the \gls{lfdr}'s are underestimated at locations in $\altRegEst$, the nominal \gls{fdr} level $\thr$ is violated. Thus, on average more than a fraction $\alpha$ of the members in $\altRegEst$ are false positives.}
		\newMinor{The \gls{lfdr} can be computed on the basis of any local summary statistics which are i.i.d. under $\HNul$. Nevertheless, the local decision statistic plays an important role in \gls{lfdr} estimation since it determines the shape of $\mixPPdf$, as shown in Fig.~\ref{fig:pdfs}.}

		\newMinor{
			We use \pval s as local decision statistics. The domain of the \pval s $\domLocSt = [0, 1]$ is bounded and parametric models are used for \pval~density estimation. The exact analytical form of $\mixPPdf$ is unknown. However, if a sufficiently flexible parametric \pval~\gls{pdf} model is used, \gls{lfdr} estimation based on \pval s offers the following advantages. The \pval s from sensors where $\HNul$ is in place are uniformly distributed. The \pval s from sensors where $\HAlt$ is in place are highly concentrated towards $0$. Hence, $\altPPdf$ is a monotonically decreasing function and most of the mass of $\mixPPdf$ is located in a subregion of $\domLocSt$ that contains the statistically significant (small) \pval s.	Additionally, \pval-based \gls{lfdr} estimation allows for a simple way \cite{Pounds2003} to decompose estimate $\mixPPdfEst$ into the components of the two-groups model $\nulFrcHat\nulPPdf = \nulFrcHat$ and $\altPPdfEst$, $\p \in\domLocSt = [0, 1]$,}
		
		\newMinor{
			\begin{equation}
				\label{eq:extract-two-groups}
				\nulFrcHat = \min\mixPPdfEst, \qquad \altPPdfEst = \frac{\mixPPdfEst - 	\nulFrcHat}{1-\nulFrcHat}.
			\end{equation}
		}
		\newMinor{An alternate popular choice for local decision statistics are \zval s, for which a one-to-one mapping to $p$-values exists. The \zval s from sensor where $\HNul$ is in place follow a standard normal distribution. \zval-based \gls{lfdr} estimators have been designed assuming a finite Gaussian mixture model \cite{Le2003} or an exponential family model \cite{Muralidharan2010, Efron2010}. Non-parametric methods have been studied using kernel estimates \cite{Robin2007} and, more recently, predictive recursion \cite{Martin2011, Scott2015, Martin2018, Newton2002}. We prefer $p$-values for the following reasons. The \zval s from sensors where $\HAlt$ is in place lie in the tails of the \zval~mixture \gls{pdf}. Obtaining high tail-accuracy for an estimated \gls{pdf} is extremely difficult. In addition, the decomposition of $\mixZPdfEst$ into $\nulFrcHat \nulZPdf$ and $(1-\nulFrcHat)\altZPdfHat$ is non-trivial.}
		\old{The suitability of different density estimation techniques to determine $\mixLocStPdfEst$ depends on the local summary statistic.} %To this end, consider the histograms in Fig.~\ref{fig:p-vs-z}, which display histograms for the two most common choices \cite{Efron2010} of local summary statistics for the same data. The left histogram shows left-sided \pval s $\locStRV  = \pRV \sim\mixPPdf$ based on signal energies, the one on the right contains the corresponding left-sided \zval s with $\locStRV = \zRV\sim\mixZPdf$.
%		\begin{figure}
%			\caption{Histograms for \pval s and \zval s.}
%			\label{fig:p-vs-z}
%		\end{figure}
		\old{The two-groups model from Eq.~\eqref{eq:two-grp-mod} applies to \pval s and \zval s alike. %However, the properties of its components $\nulLocStPdf$ and $\altLocStPdfHat$ differ significantly for different $\locStRV$.
		As illustrated in Fig.~\ref{fig:pdfs}, the shapes of $\mixLocStPdf$, $\nulLocStPdf$ and $\altLocStPdf$ differ significantly. %for different choices of $\locStRV$.
		For \pval s, $\nulLocStPdf= \nulPPdf$ follows an $\mathcal{U}[0, 1]$. For \zval s, $\nulLocStPdf = \nulZPdf$ is the standard normal distribution, $\phi(\z)$.}
		
		\old{Much of the existing literature on \gls{lfdr} estimation focuses on \zval s. $\mixZPdfEst$ is often modeled as a parametric \gls{pdf} with infinite support $\zRV \in\domLocSt = \mathbb{R}$. $\mixZPdfEst$ is found by estimating the parameters of a finite Gaussian mixture model \cite{Le2003} or an exponential family model \cite{Muralidharan2010, Efron2010}. The decomposition of $\mixZPdfEst$ into $\nulFrcHat \nulZPdf$ and  $(1-\nulFrcHat)\altZPdfHat$ for Eq.~\eqref{eq:two-grp-mod} is to be performed in a subsequent, non-trivial step. Non-parametric methods, which directly decompose $\mixZPdfEst$ into the components of the two-groups model have been studied using kernel estimates \cite{Robin2007} and, more recently, predictive recursion \cite{Martin2011, Scott2015, Martin2018, Newton2002}.}
		
		\old{To maintain \gls{fdr} control, the estimator of $\mixZPdfEst$ is required to be highly accurate for values of $\z\in\domLocSt$ that are more likely under the alternative, that is, in the left tail for left-sided \zval s. High tail accuracy of the estimate $\mixZPdfEst$ is difficult to achieve, especially at small sample sizes. An overall well-fitting $\mixZPdfEst$ does not necessarily imply a good fit in the tails, due to their small probability mass. Tail behavior is difficult to capture. Finally, the domain of $\z$, $\domLocSt = \mathbb{R}$, covers the entire real axis. Hence, the support of $\mixZPdf$ and its tails is infinite.}

		%The most popular \gls{lfdr} estimators estimate $\nulFrc$ and $\altLocStPdf$ for \zval s. Due to $\nulLocStPdf = \nulZPdf$ being a standard normal distribution, a variety of estimation schemes based on semi-parametric approaches like Gaussian deconvolution, Kernel fitting, Poisson regression, predictive recursion and Gaussian mixture models exist. TODO: INSERT SEVERAL REFERENCES.
		
		\old{The \gls{lfdr} estimators based on \pval s are not subject to the previously discussed tail-accuracy-related issues of \zval-based \gls{lfdr} estimators. To this end, observe that the assumption of  uniformly distributed \pval s under $\HNul$, i.e., $\nulPPdf$ follows  $\mathcal{U}[0, 1]$, spreads statistically insignificant observations uniformly across the bounded domain $\domLocSt = [0, 1]$ instead of producing a bulk in the center of the domain. Consequently, most of the mass of $\mixPPdf$ is located in the subregion of $\domLocSt$ that contains statistically significant (small) \pval s, i.e., where $\altPPdf$ is large. %This is the left side of the interval $\domLocSt = [0, 1]$ for left-sided \pval s, see Fig.~\ref{fig:p-vs-z}.
		Additionally, \pval-based \gls{lfdr} estimation allows for %Due to the assumption of uniformly distributed \pval s under the null hypothesis,
		a simple way \cite{Pounds2003} to decompose estimate $\mixPPdfEst$ into the components of the two-groups model $\nulFrcHat\nulPPdf = \nulFrcHat$ and $\altPPdfEst$, $\p \in\domLocSt = [0, 1]$,
		\begin{equation}
		\label{eq:extract-two-groups}
			\nulFrcHat = \min\mixPPdfEst, \qquad \altPPdfEst = \frac{\mixPPdfEst - \nulFrcHat}{1-\nulFrcHat}.
		\end{equation}}

		\old{Since the domain of the \pval s $\domLocSt = [0, 1]$ is bounded, \pval~density estimators are commonly parametric. The true analytical form of $\mixPPdf$ is unknown. Consequently, one has to identify a parametric probability model $\mixPPdfPM$ and a parameter vector value $\hat{\boldsymbol{{\theta}}}$ such that $\mixPPdfPMid \approx \mixPPdf$. The quality of the \gls{pdf} estimate $\mixPPdfEst = \mixPPdfPMid$ can be evaluated by difference measures $d(\mixPPdf, \mixPPdfEst)$, such as the \gls{kld} divergence or the \gls{wsd} distance.}

		\newMinor{The presumeably most popular parametric \pval~\gls{pdf} model is the}
		\old{Since our proposed method relies on modeling \pval~mixture distributions, we consider the popular}\gls{bum} model \cite{Pounds2003},
		\begin{equation}
		\label{eq:bum-def}
			\mixPPdfBUM	= {w} + (1-{w})\,\mathrm{beta}(\newMinor{p;}\,{a}) = {w} + (1-{w}){a}\p^{{a}-1}.
		\end{equation}
		Superscript $\text{BUM}$ indicates the dependency on the \gls{bum} parameters ${w}\in[0, 1]$ and ${a}\in(0, 1)$. $\mixPPdfBUMEst$ denotes the \gls{bum} model with the respective \gls{mle} $\hat{w}$ for $w$ and $\hat{a}$ for $a$. %Note, that ${w}$ and $(1-{w})\textrm{beta}({a})$ are not the estimated mixture components $\nulFrcHat$ and $(1-\nulFrcHat)\altPPdfEst$ from Eq.~\eqref{eq:two-grp-mod}.
		%However, these can be determined from Eq.~\eqref{eq:extract-two-groups}.
		The \gls{bum} model exploits two known properties of $\mixPPdf$, which are also apparent in Fig.~\ref{fig:pdfs}. First, the uniform distribution under $\HNul$ is captured by the constant $w$. Second, $\altPPdf$ is known to be monotonically decreasing. A single-parameter beta distribution decreases monotonically $\forall\,a \in (0, 1)$. The \gls{bum} model is simple and has been applied successfully in a number of applications. However, it lacks flexibility due to its limited number of tuning parameters. Estimating $\mixPPdf$ by $\mixPPdfEst = \mixPPdfBUMEst$ leads to overly pessimistic \gls{lfdr} estimates, as the simulations in Sec.~\ref{sec:sim-res} underline. We introduce a more flexible model in Sec.~\ref{sec:prop-met}.

	%\subsection{To be moved elsewhere}
	%	One might resort to something one might call a multiple beta-uniform mixture model, where $\altPPdf = \sum_{\cmpIdx = 1}^{\numCmp}\altFrcLoc[\cmpIdx]\beta(a_\cmpIdx, 1)$ and $\sum_{\cmpIdx = 1}^{\numCmp}\altFrcLoc[\cmpIdx] = (1-\nulFrc)$.
	%	For each increment of $k$ by 1, two more parameters have to be estimated. This quickly renders the \gls{mle} approach of \cite{Pounds2003} computationally expensive, in particular, when working with large data sets. Finding \gls{mle}s for mixture models typically requires to deploy a variant of the \gls{em} algorithm TODO: INSERT REFERENCE, which is computationally demanding especially for large data sets and is at risk of finding only a locally optimal solution TODO: INSERT REFERENCE. Also, an ideal number of components $\numCmp$ has to be determined. Classic approaches for model order selection, like information theoretic criteria rely on the assumption of Gaussian distributed data and require the model to be estimated for a number of candidate models. Despite that specific \gls{itc} for beta distributed data can be derived \cite{Golz2019}, the computational burden of having to solve an already demanding, high-dimensional maximum likelihood estimation problem multiple times cannot be avoided.

	\vspace{-8pt}
\section{The Proposed \GLS{lfdr} Estimator}
\label{sec:prop-met}
	In this section, we introduce a novel \gls{lfdr} estimator. Our approach estimates \gls{lfdr}'s from \pval s. %, due to the previously highlighted advantages of \pval-based \gls{lfdr} estimators. %We focus on the model and estimation strategies for $\mixPPdf$, since $\altPPdfEst$ and $\nulFrcHat$ follow from Eq.~\eqref{eq:extract-two-groups}.
	%
%	The key problem is to find a systematic way of estimating the \pval~mixture $\mixPPdf$ accurately. The generating model of $\mixPPdf$ in Eq.~\eqref{eq:theo-locst-pdf} is unidentifiable, since we only have access to a single observation per alternative component. Instead, we deploy a parametric estimator by identifying a class of identifiable parametric models whose parameters can be estimated such that the true $\mixPPdf$ and estimated $\mixPPdfEst$ are as similar as possible. This similarity can be quantified by difference measures $d(\mixPPdf, \mixPPdfEst)$, such as the \gls{kld} or the \gls{wsd}.
	%
	%We suggest a more flexible model for the \pval~mixture than the \gls{bum} model from Eq.~\eqref{eq:bum-def}. In particular, be
	We propose the parametric probability model
	\begin{equation}
	\label{eq:pval-mix-def}
		\mixPPdfBM = \sum_{\cmpIdxAlt = 1}^\numCmpAlt\cmpWgt[\cmpIdxAlt]\,\mathrm{beta}\big(\newMinor{p;}\,\betaA[\cmpIdxAlt]\big),
	\end{equation}
	a finite single-parameter \gls{bm} with shape parameters $\betaA[\cmpIdxAlt]\in\mathbb{R}_{>0}$ and mixture weights $\cmpWgt[\cmpIdxAlt]\in[0, 1], \, \forall\,\cmpIdxAlt\in[\numCmpAlt]$ such that $\sum_{\cmpIdxAlt = 1}^\numCmpAlt\cmpWgt[\cmpIdxAlt] = 1$ and $\numCmpAlt<\infty$. For $\betaA[\cmpIdxAlt] < 1$, $\betaA[\cmpIdxAlt]=1$ and $\betaA[\cmpIdxAlt]>1$, the $\cmpIdxAlt$-th component $\mathrm{beta}(\newMinor{p;}\,\betaA[\cmpIdxAlt])$ is monotonically decreasing, constant and monotonically increasing in $\p\in[0, 1]$, respectively. Due to the increased number of components, $\mixPPdfBM$ is more flexible than $\mixPPdfBUM$. Estimating its parameters is more involved than for the \gls{bum} model% A more flexible model provides a better approximation to the true $\mixPPdf$. %Recovering all the individual components from Eq.~\eqref{eq:theo-locst-pdf} is neither possible nor required in large-scale inference \cite{Efron2010}. Yet, the inference results improve if an \gls{lfdr} estimator captures the underlying processes more accurately.
%	Eq.~\eqref{eq:pval-mix-def} allows us to discriminate with greater detail between groups of differently distributed \pval s than the \gls{bum} model. When observing radio waves, for example, we expect to obtain very small \pval s very close to the source, but also more intermediate \pval s at locations $\pxIdx\in\altReg$ further away from the source, due to the typically (at least) quadratic signal attenuation \cite{Molisch2012} over distance to the source. For this example, $\mixPPdfEst$ from Eq.~\eqref{eq:pval-mix-def} with $L=3$ and $\betaA[1]\approx 1$, $0<\betaA[2]<1$ and $0<\betaA[3] \ll 1$ could provide a good fit to $\mixPPdf$, where the values for $\cmpWgt[\cmpIdxAlt], \cmpIdxAlt\in[\numCmpAlt]$ would depend on the number of sources and their respective transmission power.
	, since $2\numCmpAlt-1$ model parameters are to be determined from the observations $\old{\old{\pSetPx}\newMinor{\pSetSen}}\newMinor{\pSetSen} = \{\pPx[1], \dots, \pPx[\old{\numPx}\newMinor{\numSen}]\}$. Closed-form \gls{mle}s for the parameters of mixture distributions are difficult to obtain. Instead, \gls{mle}s are commonly found iteratively by \gls{em} \cite{Dempster1977}, which is computationally expensive for larger model orders. Also, the parameter estimates are only locally optimal, which may result in poorly fitting models for non-convex likelihood functions.
	In this work, we target computationally light-weighted procedures suitable for large-scale sensor networks. Our approach bases upon the \gls{mom} that estimates model parameters by solving pre-defined equation systems.
	\vspace{-5pt}
	\subsection{The method of moments}
		The principle of moment-based parameter estimation \cite{Pearson1936, Iskander1999} is to match population and empirical moments. To this end, multivariate systems of moment equations are solved. The \gls{mom} is conceptionally simple, but also entails challenges. Its analytic complexity rapidly increases with the number of model parameters. In addition, empirical higher-order moments are prone to large variance \cite{Bowman2006}. Therefore, the sample size required to provide meaningful estimates grows exponentially in the number of model parameters \cite{Hsu2012}. As a consequence, the standard \gls{mom} is not well-suited to determining the parameters for Eq.~\eqref{eq:pval-mix-def}.
		
		The \gls{smom}, a recent approach \cite{Anandkumar2012, Hsu2012}, allows to determine the parameters of multivariate Gaussian mixtures from only the first three moments. Thus, \gls{smom} avoids higher-order moments. In contrast to other work on the field of low-order moment-based parameter estimation, the method in \cite{Hsu2012} does not require a minimum distance between the locations of the mixture components to guarantee identifiability. This suits particularly well to this work, since we are dealing with \pval s on the domain $\domLocSt = [0, 1]$ and need to discriminate between mixture components that are located closely to one another. Combining the \gls{bm} model and the \gls{smom} provides a base for a computationally efficient \pval~density estimator that we introduce in what follows. %In what follows, we introduce a computationally efficient \gls{mom} estimator for the parameters of $\mixPPdfBM$ by designing a suitable data model for \pval~mixture densities. %In what follows, we establish a method to exploit the computational efficiency of the \gls{mom} to estimate the model parameters of $\mixPPdfBM$ by designing a data model that enables its application to \pval~mixture density estimation.
	\vspace{-5pt}
	\subsection{\newMajor{Overview on the proposed method}}
		\newMajor{We follow the traditional approach to \gls{lfdr} estimation, which plugs estimates for $\mixPPdf$ and $\nulFrc$ into Eq.~\eqref{eq:lfdr-def}. Since we work with $p$-values, $\nulFrc$ can be estimated as the maximum value of $\mixPPdf$. Hence, we focus on the estimation of $\mixPPdf$. We start from the assumption that the model in Eq.~\ref{eq:pval-mix-def} holds for the \gls{pdf} $\mixPPdf$ of \gls{rv} $\pRV$ with realizations $\pSetSen = \{\pSen[1], \dots, \pSen[\numSen]\}$. We then subdivide $\pSetSen$ into $M$ subsets of equal size $\numVecEl$ and form \pval~vectors $\pvec_1,\dots,\pvec_{M}$. The $\pvec_m, m\in[M]$ are observations of a $\numVecEl$-dimensional random vector $\pvecRV=[\tilde{P}_1,\dots, \tilde{P}_d]^\top$, where each $\tilde{P}_i, i\in[d]$ represents the statistical behavior of a subset of all observed $p$-values. Due to the way that the $\pvec_m$ are obtained, there is a direct relation between the \glspl{pdf} of $\pRV$ and $\pvecRV$. Thus, the parametric multivariate \gls{pdf} model for $\pvecRV$ follows directly from Eq.~\ref{eq:pval-mix-def}. The details are provided in Sec.~\ref{sec:prop-met-mvmod}. Then, we estimate the parameters of this multivariate model from the first three moments of $\pvecRV$. The empirical moments are computed using the observations $\pvec_1,\dots,\pvec_{M}$. To this end, we exploit the one-to-one relations between the moments and the model parameters derived in Theorem~\ref{theo:estimators-mix-par} and Theorem~\ref{theo:relation-pop-mom} of Sec.~\ref{sec:prop-met-theo}. Finally, we again use the relation between the \glspl{pdf} of $\pRV$ and $\pvecRV$ to obtain the estimate of the univariate \gls{pdf} of $\pRV$ from the estimate for the multivariate \gls{pdf} of $\pvecRV$. The entire procedure is presented in detail in Sec.~\ref{sec:prop-met-alg}.}
	\vspace{-5pt}
	\subsection{The multivariate \pval~model}
	\label{sec:prop-met-mvmod}
		Traditional statistical techniques would treat the input data $\old{\pSetPx}\newMinor{\pSetSen} = \{\p[1], \dots, \p[\old{\numPx}\newMinor{\numSen}]\}$ as a single observation of a\newMinor{n} $\old{\numPx}\newMinor{\numSen}$-dimensional random vector with elements $\pRV[\old{\pxIdx}\newMinor{\senIdx}], \old{\pxIdx}\senIdx\in[\old{\numPx}\newMinor{\numSen}]$. The identification of the regions associated with null hypothesis and alternative $\nulReg$ and $\altReg$ based on a single observation of a high-dimensional random vector is fairly challenging. To perform spatial inference, we adopt the idea of \textit{learning from the experience of others} \cite{Efron2010} and treat $\old{\pSetPx}\newMinor{\pSetSen}$ as $\old{\numPx}\newMinor{\numSen}$ realizations of the same scalar random variable $\pRV\sim\mixPPdf$.  We first model the \pval s as a $\numVecEl$-dimensional random vector $\pvecRV$ instead of estimating $\mixPPdf$ directly from $\old{\pSetPx}\newMinor{\pSetSen}$. Then, we estimate its joint \gls{pdf} $\mulPPdfMix$. Finally, we average over the $\numVecEl$ marginals to obtain the univariate estimate $\mixPPdfEst$. Estimating a joint \gls{pdf} appears intuitively more challenging. However, our multivariate \pval~model enables fast and reliable estimation of $\mixPPdf$, since it facilitates the application of the \gls{smom}. To the best of our knowledge, this concept is entirely new to \gls{lfdr} estimation.
		
		In what follows, assume $\pRV\sim\mixPPdf = \mixPPdfBM$. We divide the set of observations $\old{\pSetPx}\newMinor{\pSetSen} = \{\p[1], \dots, \p[\old{\numPx}\newMinor{\numSen}]\}$ for random variable $\pRV\sim\mixPPdfBM$ into $\numTls$  distinct subsets $\pSubSetPx[1], \dots, \pSubSetPx[\numTls]$ of equal size $\numVecEl$, $\numTls = \lfloor\old{\numPx}\newMinor{\numSen}/\numVecEl\rfloor$. \newMinor{Any remaining \pval s are not used for the \gls{pdf} estimation}. Next, arrange the elements of each $\pSubSetPx$ in no particular order into $\numTls$ $\numVecEl$-dimensional \pval~vectors $\pvec_{\tlIdx}\,\forall\,\tlIdx\in[\numTls]$. $\pvec_{1}, \dots, \pvec_{\numTls}$ are observations of the random vector $\pvecRV\sim\mulPPdfMix$, whose $\vecElIdx$-th entry be random variable $\tilde{\pRV[\vecElIdx]}, \vecElIdx\in[\numVecEl]$. Since $\pRV\sim\mixPPdfBM$, the marginal distribution of each $\tilde{\pRV[\vecElIdx]}$ can be described without loss of generality by a $\numCmp$-component mixture
		\begin{equation}
		\label{eq:pvec-marginal}
			\tilde{\pRV[\vecElIdx]} \sim \sum_{\cmpIdx = 1}^{\numCmp}\cmpWgt\mathrm{beta}\big(\newMinor{p;}\, \betaMixA\big),
		\end{equation}
		with mixture proportion vector $\cmpWgtVec= [\cmpWgt[1],\dots, \cmpWgt[\numCmp]]^\top$ such that $\sum_{\cmpIdx = 1}^{\numCmp} \cmpWgt = 1$ and shape parameters $\betaMixA\in\mathbb{R}_{>0}$. The partitioning of $\old{\pSetPx}\newMinor{\pSetSen}$ into $\pSubSetPx[1], \dots, \pSubSetPx[\numTls]$ and the ordering of the entries within each $\pvec_{\tlIdx}$, $\tlIdx\in[\numTls]$\old{, is not to be based} \newMinor{must be found independently of the values} $\pPx[\newMinor{\senIdx}]$, $\old{\pxIdx}\newMinor{\senIdx}\in[\old{\numPx}\newMinor{\numSen}]$. Then, the $\tilde{\pRV[\vecElIdx]}\, \forall\,\vecElIdx\in[\numVecEl]$ are mutually independent random variables. Thus, $\mulPPdfMix$ is fully characterized by its marginals, which relate to $\mixPPdfBM$ through %The distributions of the random vector entries are hence mixtures of beta distributions with $\numCmp$ mixture components.
		\begin{equation}
		\label{eq:mul-p-to-uni}
			\mixPPdfBM = \numVecEl^{-1}\sum_{\cmpIdx = 1}^\numCmp\cmpWgt\sum_{\vecElIdx=1}^\numVecEl\mathrm{beta}\big(\newMinor{p;}\,\betaMixA\big).
		\end{equation}
		\newMinor{Note, that the total number of mixture components $\numCmpAlt$ from Eq.~\eqref{eq:pval-mix-def} is then $\numCmpAlt = \numCmp\cdot\numVecEl$.}
		
		\newMinor{Based on the result in} \cite[Chapter~24]{Johnson1995}, the first two cumulants, mean $\meanMvCmp{\vecElIdx}$ and variance $\varMvCmp{\vecElIdx}$, of the $\vecElIdx$-th marginal's $\cmpIdx$-th component are found $\forall\, \vecElIdx\in[\numVecEl], \,\forall\,\cmpIdx\in[\numCmp]$, as 
		\begin{equation}
		\label{eq:first-two-mom}
			\meanMvCmp{\vecElIdx} = \frac{\betaMixA}{\betaMixA+1}, \quad\varMvCmp{\vecElIdx} = \frac{\betaMixA}{\betaMixA + 2} - \Bigg(\frac{\betaMixA}{\betaMixA+1}\Bigg)^2.
		\end{equation}
		The third-order cumulant, is
		\begin{equation}
		\label{eq:third-mom}
			\thirdCenMoMMvCmp{\vecElIdx} \!=\! \frac{\betaMixA}{\betaMixA\! + 3}-\frac{3\betaMixA}{\big(\betaMixA\! + 2\big)}\frac{\betaMixA}{\big(\betaMixA\! + 1\big)} + 2\Bigg(\frac{\betaMixA}{\betaMixA\! + 1}\Bigg)^3\!.
		\end{equation}
		%The first three cumulants are mean, variance and third-order central moment. To simplify notation, we denote the mean, variance and third-order central moment of element $\vecElIdx\in[\numVecEl]$ for mixture component $\cmpIdx \in[\numCmp]$ by $\elCmpExp{\cmpIdx}$, $\elCmpVar{\cmpIdx}$ and $\elCmpThrdCenMom{\cmpIdx}$, respectively. From \cite[Chapter~24, Section~3]{Johnson1995}, we find
		%\begin{equation}
		%	\label{eq:beta-cum-def}
		%	\cum{\cumOrd_\vecElIdx}^{(\cmpIdx)}	= \prod_{\cumOrdIdx = 0}^{\cumOrd-1} \frac{\betaMixA+\cumOrdIdx}{\betaMixA+\betaMixB+\cumOrdIdx}
			%\elCmpVar{\cmpIdx}			&= \E[|\cmpIdx]{\p[\vecElIdx]^2} - \Big(\elCmpExp{\cmpIdx}\Big)^2 =\frac{\betaMixA\betaMixB}{\Big(\betaMixA+\betaMixB\Big)^2\Big(\betaMixA+\betaMixB+1\Big)},\label{eq:beta-mix-def-entry-var}\\
			%			\elCmpVar{\cmpIdx}			&\frac{\Big(\betaMixA+1\Big)}{den}\\
			%			\elCmpThrdCenMom{\cmpIdx} 	&= \E[|\cmpIdx]{\p[\vecElIdx]^3} - 3\elCmpVar{\cmpIdx}\elCmpExp{\cmpIdx}+2\Big(\elCmpExp{\cmpIdx}\Big)^3
		%\end{equation}
		%for $\cumOrd\in[3]$.
		Additionally, denote the $\cmpIdx$-th mixture component mean vector by $\vecCmpExp = \big[\elCmpExp[1]{\cmpIdx}, \dots, \elCmpExp[\numVecEl]{\cmpIdx}\big]^\top$ and the $\cmpIdx$-th component vector of third-order cumulants by $\vecCmpThrdOrdCenMom = \big[\elCmpThrdCenMom[1]{\cmpIdx}, \dots, \elCmpThrdCenMom[\numVecEl]{\cmpIdx}\big]^\top$ for all $\cmpIdx\in[\numCmp]$. We also define the averages across mixture components, $\vecExp = [\elExp[1], \dots, \elExp[\numVecEl]]^\top = \sum_{\cmpIdx = 1}^\numCmp\cmpWgt\vecCmpExp$, and $\avThrdOrdCenMom= [\avElCmpThrdCenMom[1], \dots, \avElCmpThrdCenMom[\numVecEl]]^\top=\sum_{\cmpIdx = 1}^\numCmp\cmpWgt\vecCmpThrdOrdCenMom $.
		Since the $\tilde{\pRV[\vecElIdx]}\,\forall\,\vecElIdx\in[\numVecEl]$ are independently distributed, the $\cmpIdx$-th component's $\numVecEl\times\numVecEl$ covariance matrix $\vecCmpCov$ is diagonal with the $\vecElIdx$-th entry $\varMvCmp{\vecElIdx}$. The $\numVecEl\times\numVecEl$ mixture covariance matrix is $\vecCov = \E{(\pvecRV-\vecExp)(\pvecRV-\vecExp)^\top}$.

		To conclude this section, we formulate the following assumption on the \pval s.
		\begin{ass}[Similar component variances]
		\label{ass:same-var}
			The marginal variances of the $\cmpIdx$-th mixture component are similar, such that they can be treated as approximately equivalent across the marginals, i.e., $\cmpVar \approx \varMvCmp{\vecElIdx}\,\forall\,\vecElIdx\in[\numVecEl]$, $\cmpIdx\in[\numCmp]$.
		\end{ass}
		%Note, that Assumption~\ref{ass:same-var} allows different variances for different mixture components, i.e., $\cmpVar[\cmpIdx]\neq\cmpVar[h]$, $\cmpIdx, h \in[\numCmp]$ in general. 
		In other words, for a certain $\tlIdx\in[\numTls]$, the entries of $\pvec_{\tlIdx}$ can be treated as observations of random variables with approximately equivalent variances.
		\old{\begin{prop}
		\label{prop::model}}

		\newMinor{Our simulation results in Sec.~\ref{sec:sim-res} confirm that Assumption~\ref{ass:same-var} is fairly mild. We conducted a large number of numerical experiments with very diverse underlying spatial signals. Our proposed lfdr estimator was always able to find a subset size $\numTls<\old{\numPx}\newMinor{\numSen}$ and a number of mixture components $\numCmp<\numVecEl$ for which the \pval~subsets $\pSubSetPx$ and vectors $\pvec_{\tlIdx}$, $\tlIdx\in[\numTls]$ were formed such that Assumption~\ref{ass:same-var} holds for each mixture component $\cmpIdx\in[\numCmp]$. Hence, }%
		\old{\end{prop}}%
		\old{Thus, }the $\pvec_{\tlIdx}, \tlIdx\in[\numTls]$, can be divided into $\numCmp$ groups such that joint \gls{pdf} of the \pval~vectors in each group is described by mixture component $\cmpIdx$.
		For illustration purposes, consider that the subsets $\pSubSetPx$ are formed based on spatial proximity, i.e., each $\pSubSetPx$ is composed of \pval s from $\numVecEl$ neighboring locations. For those subsets containing exclusively \pval s from locations $\old{\pxIdx}\newMinor{\senIdx}\in\nulReg$, the statistical properties of each element are similar by design. Due to the assumed spatial smoothness of \old{the}\newMinor{spatial}\old{physical} phenomena\old{ of interest}, also the \pval s obtained at close-by locations $\old{\pxIdx}\newMinor{\senIdx}\in\altReg$ have similar statistical properties.
		\old{Our simulation results in Sec.~\ref{sec:sim-res} confirm that Assumption~\ref{ass:same-var} is fairly mild.} %Its validity depends on the process of partitioning $\old{\pSetPx}\newMinor{\pSetSen}$ into subsets and the value of $\numVecEl$, since in particular for larger $\numVecEl$, one might end up with some subsets being composed of \pval s from locations in $\nulReg$ and $\altReg$. On the other hand, small $\numVecEl$ limits the number of mixture components and hence the flexibility of the model.
%		In practice, a violation of Assumption \ref{ass:same-var} can be identified by a reduced fit of $\mixPPdfEst$ from Eq.~\eqref{eq:pvec-marginal} estimated with the spectral \gls{mom} to the true $\mixPPdf$.
		
		Under Assumption~\ref{ass:same-var}, the $\cmpIdx$-th mixture component covariance matrix is $\vecCmpCov \approx \cmpVar\mathbf{I}_{\numVecEl\times\numVecEl}$, where $\mathbf{I}_{\numVecEl\times\numVecEl}$ is the $\numVecEl\times\numVecEl$ identity matrix. The average variance over mixture components is $\avVar \approx \sum_{\cmpIdx = 1}^\numCmp \cmpWgt\cmpVar$.
	\vspace{-5pt}
	\subsection{The spectral method of moments}
	\label{sec:prop-met-theo}
		The spectral method of moments was formulated for multivariate spherically Gaussian distributed data in \cite{Hsu2012}. Their approach builds on the relation between the population moments and the model parameters, namely, the mixture weights, means and variances. In this section, we formulate similar relations for the \pval~vectors, given that they follow the model from Eq.~\eqref{eq:pvec-marginal} and fulfill Assumption~\ref{ass:same-var}.
		We first extend \cite[Theorem 2]{Hsu2012} such that it fits to our proposed data model.
		%\subsubsection{Theory}
		\begin{theo}[Relation of mixture model parameters to spectral quantities]
		\label{theo:estimators-mix-par}
			For a $\numVecEl\times 1$ random vector $\pvecRV\newMinor{ = [\tilde{\pRV[1]}, \dots, \tilde{\pRV[\numVecEl]}]^\top}$ with joint \gls{pdf} $\mulPPdfMix$ \newMinor{and the marginals of $\tilde{\pRV[1]}, \dots, \tilde{\pRV[\numVecEl]}$ are distributed as} defined in Eq.~\eqref{eq:pvec-marginal}\newMinor{,} \old{and }under Assumption~\ref{ass:same-var}, the beta distribution shape parameters $\betaMixA\,\forall\,\vecElIdx\in[\numVecEl]$ and mixture component weights $\cmpWgt\,\forall\,\cmpIdx\in[\numCmp]$ can be expressed by $\betaMixA = \frac{\meanMvCmp{\vecElIdx}}{1-\meanMvCmp{\vecElIdx}}$ and $\cmpWgtVec =  \Big[\vecCmpExp[1],\dots, \vecCmpExp[\numCmp]\Big]^\dagger\vecExp = \Big[\cmpWgt[1], \dots, \cmpWgt[\numCmp]\Big]^\top$ with 
			\begin{equation}
			\label{eq:theo2-mean-vec}
%				\vecCmpExp[\cmpIdx] = \frac{\lambda^{(\cmpIdx)}}{\boldsymbol{\eta}^\top\secondPopSMoMTheo^{\frac{1}{2}} \mathbf{v}^{(\cmpIdx)}}\secondPopSMoMTheo^{\frac{1}{2}}\mathbf{v}^{(\cmpIdx)},
				\vecCmpExp[\cmpIdx] = \frac{\lambda^{(\cmpIdx)}}{\boldsymbol{\eta}^\top\B \mathbf{v}^{(\cmpIdx)}}\B\mathbf{v}^{(\cmpIdx)},
			\end{equation}
%			if $\numCmp<\numVecEl$, $\cmpWgt>0, \,\forall\,\cmpIdx\in[\numCmp]$ and the $\vecCmpExp\,\forall\,\cmpIdx\in[\numCmp]$ are linearly independent.
%			Here, $\boldsymbol{\eta} = [\eta_1, \dots, \eta_\numVecEl]^\top$ is a vector chosen uniformly at random from the unit sphere in $\mathbb{R}^\numVecEl$, $\big(\lambda^{(\cmpIdx)}, \mathbf{v}^{(\cmpIdx)}\big),\, \cmpIdx\in[\numCmp]$ are the non-zero eigenvalue, eigenvector pairs of a $\numVecEl\times\numVecEl$ matrix $\secondPopSMoMTheo^{\dagger^\frac{1}{2}}\thirdPopSMoMTheoVec{\boldsymbol{\eta}}\secondPopSMoMTheo^{\dagger^\frac{1}{2}}$ and $(\cdot)^\dagger$ and $(\cdot)^{\frac{1}{2}}$ denote the Moore-Penrose pseudoinverse and matrix square root, respectively. In addition, $\secondPopSMoMTheo, \thirdPopSMoMTheoVec{\boldsymbol{\eta}}\in\mathbb{R}^{\numVecEl\times\numVecEl}$ and $\thirdPopSMoMTheo\in\mathbb{R}^{\numVecEl\times\numVecEl\times\numVecEl}$ such that
			if $\numCmp<\numVecEl$, $\cmpWgt>0, \,\forall\,\cmpIdx\in[\numCmp]$ and the $\cmpIdx$-th mixture component mean vectors $\vecCmpExp\,\forall\,\cmpIdx\in[\numCmp]$ are linearly independent.
			Here, $\boldsymbol{\eta} = [\eta_1, \dots, \eta_\numVecEl]^\top$ is a vector chosen uniformly at random from the unit sphere in $\mathbb{R}^\numVecEl$, $\big(\lambda^{(\cmpIdx)}, \mathbf{v}^{(\cmpIdx)}\big),\, \cmpIdx\in[\numCmp]$ are the (eigenvalue, eigenvector) pairs of a $\numCmp\times\numCmp$ matrix $\W^\top\thirdPopSMoMTheoVec{\boldsymbol{\eta}}\W$. % and $(\cdot)^\dagger$ and $(\cdot)^{\frac{1}{2}}$ denote the Moore-Penrose pseudoinverse and matrix square root, respectively.
			The projection matrices $\B = \U\big(\U^\top\secondPopSMoMTheo\U\big)^\frac{1}{2}\in\mathbb{R}^{\numCmp\times\numVecEl}$ and $\W = \U\big(\U^\top\secondPopSMoMTheo\U\big)^{\dagger \frac{1}{2}}\in\mathbb{R}^{\numVecEl\times\numCmp}$ are based on the matrix of left singular vectors $\U\in\mathbb{R}^{\numVecEl\times\numCmp}$ of the thin \gls{svd} of $\secondPopSMoMTheo = \U\mathbf{S}\mathbf{R}^\top$. In addition, $\secondPopSMoMTheo, \thirdPopSMoMTheoVec{\boldsymbol{\eta}}\in\mathbb{R}^{\numVecEl\times\numVecEl}$ and $\thirdPopSMoMTheo\in\mathbb{R}^{\numVecEl\times\numVecEl\times\numVecEl}$ such that
			\begin{align}
				\secondPopSMoMTheo 			&= \sum_{\cmpIdx = 1}^\numCmp\cmpWgt\vecCmpExp\tensorprod\vecCmpExp,\\
				\thirdPopSMoMTheoVec{\boldsymbol{\eta}} 	&= \sum_{\vecElIdx_1 = 1}^\numVecEl\sum_{\vecElIdx_2 = 1}^\numVecEl\sum_{\vecElIdx_3 = 1}^\numVecEl\big[\thirdPopSMoMTheo\big]_{\vecElIdx_1, \vecElIdx_2, \vecElIdx_3}[\boldsymbol{\eta}]_{\vecElIdx_3}\baseVec[\vecElIdx_1]\!\tensorprod\baseVec[\vecElIdx_2],\label{eq:third-mom-op}\\
				\thirdPopSMoMTheo 			&= \sum_{\cmpIdx = 1}^\numCmp\big(\cmpWgt\vecCmpExp\tensorprod\vecCmpExp\tensorprod\vecCmpExp\big), \label{eq:theo1_M3}
			\end{align}
			with $\baseVec[\vecElIdx]$ a $\numVecEl\times 1$ vector of $\numVecEl-1$ zeros and $1$ as its $\vecElIdx$-th entry. $\tensorprod$ is the outer product.
		\end{theo}
		\begin{proof}
			$\betaMixA = \frac{\meanMvCmp{\vecElIdx}}{1-\meanMvCmp{\vecElIdx}}$ follows from Eq.~\eqref{eq:first-two-mom}. The relations for the mixture component means $\vecCmpExp$, $\cmpIdx\in[\numCmp]$ and the mixture proportion vector $\cmpWgtVec$ result directly from the proof of \cite[Theorem 2]{Hsu2012}. Note, that $\B\B^\top = \secondPopSMoMTheo$ and $\W\W^\top = \secondPopSMoMTheo^{-1}$, which implies that $\W^\top\thirdPopSMoMTheoVec{\boldsymbol{\eta}}\W = \mathbf{R}^\top\mathbf{D}(\boldsymbol{\eta})\mathbf{R}$ is diagonalizable along the lines of $\mathbf{M}_\text{GMM}(\boldsymbol{\eta})$ in \cite[Theorem~2]{Hsu2012} by a diagonal matrix $\mathbf{D}(\boldsymbol{\eta})\in\mathbb{R}^{\numCmp\times\numCmp}$ with diagonal entries $\lambda^{(\cmpIdx)}=\boldsymbol{\eta}^\top\vecCmpExp, \cmpIdx\in[\numCmp]$.
		\end{proof}

		The relations established in Theorem~\ref{theo:estimators-mix-par} allow to estimate the mixture model by means of its parameters $\vecCmpExp$ and $\cmpWgtVec$, since $\vecExp$, $\secondPopSMoMTheo$, and $\thirdPopSMoMTheo$ enable a one-to-one mapping between the model parameters and the observable population moments. In particular, $\vecExp = \E{\pvecRV}$ is the first moment of $\pvecRV$, whereas $\secondPopSMoMTheo$ and $\thirdPopSMoMTheo$ are related to the second and third moments $\E{\pvecRV\tensorprod\pvecRV}$ and $\E{\pvecRV\tensorprod\pvecRV\tensorprod\pvecRV}$ of $\pvecRV$. The exact relationships are derived in Theorem~\ref{theo:relation-pop-mom}.

		\begin{theo}[Relation of spectral to observable quantities]
			\label{theo:relation-pop-mom}
			Under the assumptions in Theorem~\ref{theo:estimators-mix-par}, the average variance over mixture components $\avVar = \sum_{\cmpIdx = 1}^\numCmp \cmpWgt\cmpVar$ is the smallest eigenvalue of the population covariance matrix $\vecCov$. With $\noiseEigVec \in \mathbb{R}^\numVecEl$ any unit-norm eigenvector of eigenvalue $\avVar$, we find
			\begin{align}
				%\firstPopSMoM 	&= \E{\pvec \Big(\noiseEigVec^\top\big(\pvec-\vecExp\big)\Big)^2},\\ %- \noiseEigVec\odot\noiseEigVec\odot\avThrdOrdCenMom \\
				\secondPopSMoMTheo 	&= \E{\pvecRV\tensorprod\pvecRV} - \avVar\mathbf{I}_{\numVecEl\times\numVecEl}, \label{eq:theo2_M2}\\
				\thirdPopSMoMTheo 	&= \thirdPopSMoM- \thirdPopSMoMErr,\label{eq:theo2_M3}
			\end{align}
			where the observable $\thirdPopSMoM$ and the difference $\thirdPopSMoMErr$ to the non-observable $\thirdPopSMoMTheo$ are
			
			\begin{align}
				\begin{split}
					\label{eq:third-mom-obs}
					\thirdPopSMoM 		&= \E{\pvecRV\tensorprod\pvecRV\tensorprod\pvecRV}-\sum_{\vecElIdx = 1}^{\numVecEl}\big(\firstPopSMoM\tensorprod \baseVec\tensorprod\baseVec\\ &\qquad+\baseVec\tensorprod\firstPopSMoM\tensorprod\baseVec+\baseVec\tensorprod\baseVec\tensorprod\firstPopSMoM\big),
				\end{split}\\
				\thirdPopSMoMErr	&= \thirdPopSMoMErrThird - \thirdPopSMoMErrFirst,
			\end{align}
			with third-order tensors $\thirdPopSMoM, \thirdPopSMoMErrThird, \thirdPopSMoMErrFirst \in\mathbb{R}^{\numVecEl\times\numVecEl\times\numVecEl}$, $\numVecEl\times1$ vector $\firstPopSMoM = \firstPopSMoMTheo + \firstPopSMoMErr \in\mathbb{R}^{\numVecEl}$, % $\secondPopSMoM \in \mathbb{R}^{\numVecEl\times\numVecEl}$, $\thirdPopSMoM\in\mathbb{R}^{\numVecEl\times\numVecEl\times\numVecEl}$ and the tensor product $\tensorprod$.
%			These observable quantities are related to the population parameters through
			\begin{align}
				\label{eq:first-mom-obs}
				\firstPopSMoM  	&= \E{\pvecRV \Big(\noiseEigVec^\top\big(\pvecRV-\vecExp\big)\Big)^2}, \quad  \firstPopSMoMErr 		= \vphantom{\sum_{\cmpIdx = 1}^\numCmp}\noiseEigVec\odot\noiseEigVec\odot\avThrdOrdCenMom.%= \firstPopSMoMTheo + \firstPopSMoMErr,\\
%				\secondPopSMoM	&= \sum_{\cmpIdx = 1}^\numCmp\cmpWgt\vecCmpExp\tensorprod\vecCmpExp \coloneqq \secondPopSMoMTheo,\\
%				\thirdPopSMoM	&= \thirdPopSMoMTheo + \thirdPopSMoMErr, \\
%								&= \sum_{\cmpIdx = 1}^\numCmp\Big(\cmpWgt\cmpVar\vecCmpExp\Big) + \vphantom{\sum_{\cmpIdx = 1}^\numCmp}\noiseEigVec\odot\noiseEigVec\odot\avThrdOrdCenMom
			\end{align}
%			where we distinguish between the terms depending on the mixture model defining parameters, i.e., the weights $\cmpWgt$, means $\vecCmpExp$ and variances $\cmpVar$,
%			\begin{align*}
%				\firstPopSMoMTheo 	&= \sum_{\cmpIdx = 1}^\numCmp\Big(\cmpWgt\cmpVar\vecCmpExp\Big) \\
%				\thirdPopSMoMTheo 	&= \sum_{\cmpIdx = 1}^\numCmp\big(\cmpWgt\vecCmpExp\tensorprod\vecCmpExp\tensorprod\vecCmpExp\big)
%			\end{align*}
%			and the terms depending on the average third cumulant $\avThrdOrdCenMom$,
%			\begin{align*}
%				&\firstPopSMoMErr 		= \vphantom{\sum_{\cmpIdx = 1}^\numCmp}\noiseEigVec\odot\noiseEigVec\odot\avThrdOrdCenMom\qquad\qquad\qquad\qquad\thirdPopSMoMErr	= \thirdPopSMoMErrThird - \thirdPopSMoMErrFirst,
%			\end{align*}
			$\odot$ denotes the Hadamard product and
			\begin{align}
				\thirdPopSMoMErrThird	= \sum_{\vecElIdx = 1}^\numVecEl&\Big(\avElCmpThrdCenMom\,\baseVec\tensorprod\baseVec\tensorprod\baseVec\Big)\label{eq:theo2_M3_delta3}\\
				\begin{split}
					\thirdPopSMoMErrFirst	= \sum_{\vecElIdx = 1}^\numVecEl&\Big(\firstPopSMoMErr\tensorprod\baseVec\tensorprod\baseVec \\
					&+ \baseVec\tensorprod\firstPopSMoMErr\tensorprod\baseVec + \baseVec\tensorprod\baseVec\tensorprod\firstPopSMoMErr\Big).
				\end{split}
			\end{align}
			$\avThrdOrdCenMom = \big[\avElCmpThrdCenMom[1],\dots, \avElCmpThrdCenMom[\numVecEl]\big]^\top$ is the vector composed of the $\numVecEl$ marginals' mean third cumulants, i.e., Eq.~\eqref{eq:third-mom} $\forall\,\vecElIdx\in[\numVecEl]$ averaged over the $\numCmp$ mixture components.
		\end{theo}
	\vspace{-5pt}
		\begin{proof}
			See Appendix~\ref{apx:smom-theo}.
		\end{proof}
	\vspace{-5pt}
		The first and second population moments $\vecExp$ and $\E{\pvecRV\tensorprod\pvecRV}$, the population covariance matrix $\vecCov$ and consequently also $\avVar$, $\noiseEigVec$ and $\secondPopSMoMTheo$ can be found using consistent sample estimates of the moments and covariance. %, using the observations $\pvec_{1},\dots, \pvec_{\numTls}$.
		$\thirdPopSMoMTheoVec{\boldsymbol{\eta}}$ cannot be estimated directly, since only $\thirdPopSMoM$ depends exclusively on sample moments, but not $\thirdPopSMoMErr$. However, we show in Appendix~\ref{apx:smom-conv} that $\thirdPopSMoMVec{\boldsymbol{\eta}} = \sum_{\vecElIdx_1 = 1}^\numVecEl\sum_{\vecElIdx_2 = 1}^\numVecEl\sum_{\vecElIdx_3 = 1}^\numVecEl\big[\thirdPopSMoM\big]_{\vecElIdx_1, \vecElIdx_2, \vecElIdx_3}[\boldsymbol{\eta}]_{\vecElIdx_3}\baseVec[\vecElIdx_1]\!\tensorprod\baseVec[\vecElIdx_2]$ is a sufficiently good approximation of $\thirdPopSMoMTheoVec{\boldsymbol{\eta}}$ for estimating the \pval~mixture density. $\thirdPopSMoMVec{\boldsymbol{\eta}}$ is observable. % With Corollary~\ref{cor:third-mom}, the sample estimates $\secondPopSMoMTheoEst$ and $\thirdPopSMoMEstVec{\boldsymbol{\eta}}$ can be used to estimate $\cmpWgtVec$ and $\betaMixA$, $\forall\,\vecElIdx\in[\numVecEl], \forall\,\cmpIdx\in[\numCmp]$, from Theorem~\ref{theo:estimators-mix-par}.
%		\begin{cor}
%		\label{cor:third-mom}
%			The sample-based estimator
%			$\thirdPopSMoMEstVec{\boldsymbol{\eta}} = \sum_{\vecElIdx_1 = 1}^\numVecEl\sum_{\vecElIdx_2 = 1}^\numVecEl\sum_{\vecElIdx_3 = 1}^\numVecEl\big[\thirdPopSMoMEst\big]_{\vecElIdx_1, \vecElIdx_2, \vecElIdx_3}[\boldsymbol{\eta}]_{\vecElIdx_3}\baseVec[\vecElIdx_1]\!\tensorprod\baseVec[\vecElIdx_2]$, with $\thirdPopSMoMEst$ found by plugging in $\pvec_{1},\dots, \pvec_{\numTls}$ into Eq.~\eqref{eq:third-mom-obs} and Eq.~\eqref{eq:first-mom-obs}, is a sufficiently good estimator of $\thirdPopSMoMTheoVec{\boldsymbol{\eta}}$. 
%		\end{cor}
%		\begin{proof}
%			See Appendix~\ref{apx:smom-conv}
%		\end{proof}
		In the following section, we exploit these relations to estimate the component first central moments $\vecCmpExp$ and the model parameters $\betaMixA$, $\cmpWgt$ $\,\forall\,\cmpIdx\in[\numCmp], \,\forall\, \vecElIdx\in[\numVecEl]$ for Eq.~\eqref{eq:mul-p-to-uni}.
	\vspace{-5pt}
	\subsection{The \gls{lfdr}-\gls{smom} estimator}
	\label{sec:prop-met-alg}
%		In this section, we present the complete \gls{lfdr} estimation algorithm. The simulation results in Sec.~\ref{sec:sim-res} indicate that this \gls{lfdr} estimator is conservative in the sense of Definition~\ref{def:con-lfdr-est}, thereby enabling the identification of alternative regions with strict control of the \gls{fdr}, while being more powerful than existing methods for a variety of spatial inference scenarios.
%		
		The proposed method is summarized in Alg.~\ref{alg:spa-lfdr-est}. It partitions the \pval s into subsets, estimates the parameters for the joint \gls{pdf} of the resulting \pval~vectors and determines $\mixPPdfEst=\mixPPdfBMEst$ from Eq.~\eqref{eq:mul-p-to-uni}. The density fit is repeated several times for \newMinor{$G$} different \pval~vectors, increasing subset sizes $\numVecEl = 2, 3, \dots$ and increasing model orders $\numCmp < \numVecEl$. \newMinor{$M \cdot\numVecEl = \lfloor\numSen/\numVecEl\rfloor\cdot\numVecEl$ $p$-values are used for the density fit if $\numSen/\numVecEl$ is not an integer.} If the distinct subsets $\pSubSetPx\subset\old{\pSetPx}\newMinor{\pSetSen}$ are \old{immovable}\newMinor{fixed}, because they are formed based on fixed covariate information like spatial proximity, we randomly rearrange the elements \newMinor{$G$ times} within each \pval~vector for different runs. \newMinor{Selecting $G=10$ created sufficient flexibility to find $p$-value vectors that yield sufficiently accurate density model estimates in all considered scenarios}. The \gls{lfdr}'s are estimated using $\mixPPdfEst$, Eq.~\eqref{eq:extract-two-groups} and Eq.~\eqref{eq:lfdr-est}. The goodness of fit is assessed by the value $\disMsr^\ast$ of a difference measure between $\mixPPdfEst$ and the data. $\disMsr^\ast$ is initialized with a large value, to ensure that the algorithm finds a solution.
		The best solution has been found, if additional degrees of freedom do not lead to a better fit.
		
		%Since the validity of Assumption~\ref{ass:same-var} and the impact of the neglected error terms $\thirdPopSMoMErr$ on the parameter estimates depends on the \pval~vectors, we also require also different runs with different $\numVecEl$ and $\pvec_{1},\dots, \pvec_{\numTls}$.
		
%		The multiple runs with differently sorted tile vectors are required since we aim to assign two tiles composed of similarly behaved \pval s to the same component $\cmpIdx$, regardless of the position of the corresponding pixes within the tile. Consider, for example, two $2\times2$ tiles. Say, both contain three pixels and one pixel at which $\HNul$ and $\HAlt$ are in place. Also, $\HAlt$ leads to statistically similarly behaving alternate \pval s in both tiles. However, for the first tile, $\HAlt$ is in place at the pixel located in the top left corner, whereas in the second, $\HAlt$ is in place at the pixel located in the bottom right corner. If the ordering of the entries of the tile vectors depended on the location of pixels within the tile, associating those two tiles the same mixture model component would not be justified.
		The parameters of the \pval~vector mixture density are estimated by Alg.~\ref{alg:learnSB}, which determines the right side of Eq.~\eqref{eq:theo2-mean-vec} from the sample moments. In Line 1, the data is split into two distinct sets of equal size. If $\numTls$ is odd, we drop one $\pvec_{\tlIdx}$. The sample moment-based estimates $\secondPopSMoMTheoEst$ for $\secondPopSMoMTheo$ and $\thirdPopSMoMEstVec{\boldsymbol{\eta}}$ for $\thirdPopSMoMVec{\boldsymbol{\eta}}$ are multiplied during the estimation process. Hence, they must be computed from different data to guarantee their independence. The sample covariance matrix estimates $\hat{\vecCov}_{\mathcal{S}}$ in Line 3 are full rank, but Line 5 reduces the rank to its assumed value $\numCmp$. $\thirdPopSMoMEst$ is determined in Lines 6 and 7. Lines 8 to 13 are dedicated to the estimation of the eigenvector, eigenvalue pairs $\big(\lambda^{(\cmpIdx)}, \mathbf{v}^{(\cmpIdx)}\big)$ for Eq.~\eqref{eq:theo2-mean-vec}. %Our proposed way to generate $\boldsymbol{\eta}_u$ ensures that $\thirdPopSMoMVec{\boldsymbol{\eta}_u}\approx\thirdPopSMoMTheoVec{\boldsymbol{\eta}_u}$ is justified. %To enable finite sample performance analysis, the authors of \cite{Hsu2012} estimate $\big(\lambda^{(\cmpIdx)}, \mathbf{v}^{(\cmpIdx)}\big)$ from estimates of $\W\thirdPopSMoMTheoVec{\boldsymbol{\theta}}\W^\top$ instead of $\secondPopSMoMTheo^{\dagger^\frac{1}{2}}\thirdPopSMoMTheoVec{\boldsymbol{\eta}}\secondPopSMoMTheo^{\dagger^\frac{1}{2}}$. The relation for $\vecCmpExpEst$ in Line 16 is the equivalent to Eq.~\eqref{eq:theo2-mean-vec}, since $\secondPopSMoMTheo^{\dagger^\frac{1}{2}} = \W\U^\top$.
		We generate $U$ different vectors $\boldsymbol{\eta}_u$ and select the best run in Alg.~\ref{alg:spa-lfdr-est} to fulfill Lemma~\ref{lem:apx} \newMinor{given in Appendix~\ref{apx:smom-conv}}. We found a value as low as $U = 10$ to provide satisfying results.
		\begin{algorithm}
			\caption{The proposed algorithm \gls{lfdr}-\gls{smom}}
			\label{alg:spa-lfdr-est}
			\hspace*{\algorithmicindent} \textbf{Input}: $\old{\pSetPx}\newMinor{\pSetSen} = \{\p[1], \dots, \p[\old{\numPx}\newMinor{\numSen}]\}$, \old{$\mathcal{K}$, $\mathcal{D}$, }$\numPer$, $\disFct{\cdot}{\cdot}$, $U$\\%$\numPxX$, $\numPxY$,\\
			\hspace*{\algorithmicindent} \textbf{Output}: $\lfdrLocHat[1], \dots, \lfdrLocHat[\old{\numPx}\newMinor{\numSen}]$
			\begin{algorithmic}[1]
				\vspace{1pt}
				\Statex\textbf{Step 1:} \textit{Estimation of $\mixPPdf$}
				\State Initialize $\disMsr^\ast$ as the largest possible number
				\State Compute the histogram $\pHist$ for the data in $\old{\pSetPx}\newMinor{\pSetSen}$
				\For{$\numVecEl\in\mathbb{N}_{\geq 2}$}
				\State Divide $\old{\pSetPx}\newMinor{\pSetSen}$ into subsets $\mathcal{P}^{\numVecEl}_\tlIdx$, $\tlIdx \in[\numTls]$
				\For{$\perIdx\in[\numPer]$}
				\State \Longunderstack[l]{Form vectors $\pvec_\tlIdx\in\mathbb{R}^{\numVecEl}, \forall\,\tlIdx\in[\numTls]$, by random\\ordering of the elements in $\mathcal{P}^{\numVecEl}_\tlIdx$}
				\State Define $\pvecSet^{\numVecEl\times\numTls} = \big\{\pvec_{\tlIdx}\in\mathbb{R}^{\numVecEl}\,|\,\tlIdx\in[\numTls] \big\}$, $\disMsr^{\numVecEl} = \infty$
				\For{$\numCmp\leq\numVecEl$}
				\State \Longunderstack[l]{Obtain $U$ sets of parameters $\cmpWgtEst_u$, $\betaMixAEst[\vecElIdx, u], $\\$ \forall\,\cmpIdx\in[\numCmp], \forall\,\vecElIdx\in[\numVecEl], \forall\, u\in[U]$ via Alg.~\ref{alg:learnSB}}
				\State Find $\mixPPdfEst[u]$ for $\cmpWgtEst_u,\betaMixAEst[\vecElIdx, u]$ and Eq.~\eqref{eq:mul-p-to-uni} $u\in[U]$
				\State Select $u^\ast\!=\!\mathrm{argmin}_u\big(\disFct{\pHist}{ \mixPPdfEst[u]}\big)$
				\If{$\disFct{\pHist}{\mixPPdfEst[u^\ast]}<\disMsr^{\numVecEl}$}
				\State $\disMsr^{\numVecEl} = \disFct{\pHist}{\mixPPdfEst[u^\ast]}$
				\If{$\disMsr^{\numVecEl}\leq\disMsr^\ast$}
				\State $\disMsr^{\ast} =\disMsr^{\numVecEl}, \mixPPdfEst = \mixPPdfEst[u^\ast]$, $\numVecEl^\ast = \numVecEl$
				\EndIf
				\Else{ break.}
				\EndIf 
				\EndFor
				\EndFor
				\If{$\numVecEl^\ast \neq \numVecEl$} break.
				\EndIf
				\EndFor
				\vspace{3pt}
				\Statex \textbf{Step 2:} \textit{Estimation of $\nulFrc$ and $\altPPdf$}
				\State Compute $\nulFrcHat = \min\big(\mixPPdfEst\big)$
				\State Compute $\altPPdfEst = \big(1-\nulFrcHat\big)^{-1}(\mixPPdfEst-\nulFrcHat\big)$
				\State Determine $\lfdrLocHat[\old{\pxIdx}\newMinor{\senIdx}], \forall\,\old{\pxIdx}\newMinor{\senIdx}\in[\old{\numPx}\newMinor{\numSen}]$ from Eq.~\eqref{eq:lfdr-est}% using $\nulFrcHat$, $\altPPdfEst$
			\end{algorithmic}
		\end{algorithm}
		\begin{algorithm}
			\caption{Mixture model parameter estimation}
			\label{alg:learnSB}
			\hspace*{\algorithmicindent} \textbf{Input}: $\pvecSet^{\numVecEl\times\numTls} = \big\{\pvec_{\tlIdx}\in\mathbb{R}^{\numVecEl}\,|\,\tlIdx = 1, \dots, \numTls\big\}$, $\numCmp$\\
			\hspace*{\algorithmicindent} \textbf{Output}: $\cmpWgtVecEst_u\in\mathbb{R}^\numCmp$, $\betaMixAEst[\vecElIdx, u],\forall\,\vecElIdx\in[\numVecEl],\forall\,\cmpIdx\in[\numCmp], \forall\, u\in[U]$
			\begin{algorithmic}[1]
				\State Randomly split $\pvecSet^{\numVecEl\times\numTls}$ into $\big\{\!\pvecSetOne,\pvecSetTwo\,\big|\,|\pvecSetOne| = |\pvecSetTwo|, \pvecSetOne\,\cap\,\pvecSetTwo\!=\! \emptyset\!\big\}$
				\State Compute $\vecExpEstSet{\mathcal{S}} = {|\mathcal{S}|}^{-1}\sum_{\pvec\in\mathcal{S}}\pvec, \mathcal{S} = \pvecSetOne, \pvecSetTwo$
				\State Find $\hat{\vecCov}_{\mathcal{S}} = {|\mathcal{S}|}^{-1}\big(\sum_{\pvec\in\mathcal{S}}\pvec\pvec^\top\big)-\vecExpEstSet{\mathcal{S}}\vecExpEstSet{\mathcal{S}}^\top\old{, \mathcal{S} = \pvecSetOne, \pvecSetTwo}$
				\State Determine $\avVarEstSet{\mathcal{S}}$, the smallest eigenvalue of $\hat{\vecCov}_{\mathcal{S}}$ and $\noiseEigVecEst{\mathcal{S}}$, its corresponding eigenvector\old{, $\mathcal{S} = \pvecSetOne, \pvecSetTwo$}
				\State Compute the best rank-$\numCmp$ estimate of $\secondPopSMoMTheo$,
				\begin{equation*}
					\secondPopSMoMTheoEst = \arg\min_{\mathbf{X}\in\mathbb{R}^{\numVecEl\times\numVecEl}|\mathrm{rank}(\mathbf{X})\leq\numCmp}\big|\big|\hat{\vecCov}_{\pvecSetOne}-\mathbf{X}\big|\big|_2.
				\end{equation*}
				\State Find $\firstPopSMoMEst{\mathcal{S}} = |\mathcal{S}|^{-1}$ $\sum_{\pvec\in\mathcal{S}}\pvec\Big(\noiseEigVecEst{\mathcal{S}}\big(\pvec-\vecExpEstSet{\mathcal{S}}\big)^2\Big)\old{, \mathcal{S} = \pvecSetOne, \pvecSetTwo}$
				\State Compute $\thirdPopSMoMEst = |\pvecSetTwo|^{-1}\Big(\sum_{\pvec\in\pvecSetTwo}\pvec\tensorprod\pvec\tensorprod\pvec\Big) - \sum_{\vecElIdx = 1}^{\numVecEl}\Big(\firstPopSMoMEst{\pvecSetTwo}\tensorprod \baseVec\tensorprod\baseVec+\baseVec\tensorprod\firstPopSMoMEst{\pvecSetTwo}\tensorprod\baseVec+\baseVec\tensorprod\baseVec\tensorprod\firstPopSMoMEst{\pvecSetTwo}\Big)$
				\State Find $\hat{\U}\in\mathbb{R}^{\numVecEl\times\numCmp}$, the left singular vectors of $\secondPopSMoMTheoEst$
				\State Compute $\hat{\W} = \hat{\U}\big(\hat{\U}^\top\secondPopSMoMTheoEst\hat{\U}\big)^{\dagger\frac{1}{2}}\in\mathbb{R}^{\numVecEl\times \numCmp}$ and $\hat{\B} = 	\hat{\U}\big(\hat{\U}^\top\secondPopSMoMTheoEst\hat{\U}\big)^{\frac{1}{2}}\in\mathbb{R}^{\numCmp\times\numVecEl}$
				\For{$u \in [U]$}
				\State Select $\boldsymbol{\eta}_u$ uniformly at random from the $\mathbb{R}^\numVecEl$ unit sphere % and define $\boldsymbol{\theta}_u = \hat{\W}\boldsymbol{\eta}_u\in\mathbb{R}^\numVecEl$
				\State Find $\hat{\W}^\top\!\thirdPopSMoMEstVec{\boldsymbol{\eta}_u}\hat{\W}\!\in\!\mathbb{R}^{\numCmp\!\times\!\numCmp}$ $\forall\,\cmpIdx, \cmpIdxAlt\!\in[\numCmp]$,
				\begin{equation*}
%					\Big[\hat{\W}^\top\thirdPopSMoMEstVec{\boldsymbol{\theta}_u}\hat{\W}\Big]_{\cmpIdx_1, \cmpIdx_2}=\sum_{\vecElIdx_1 = 1}^\numVecEl\sum_{\vecElIdx_2 = 1}^\numVecEl\sum_{\vecElIdx_3= 1}^\numVecEl[\hat{\W}]_{\vecElIdx_1, \cmpIdx_1}[\hat{\W}]_{\vecElIdx_2, \cmpIdx_2}[\boldsymbol{\theta}_u]_{\vecElIdx_3}
					\Big[\hat{\W}^\top\thirdPopSMoMEstVec{\boldsymbol{\eta}_u}\hat{\W}\Big]_{\cmpIdx, \cmpIdxAlt}\!=\!\sum_{h, i, j\in[\numVecEl] }\![\hat{\W}]_{h, \cmpIdx}[\hat{\W}]_{i, \cmpIdxAlt}[\boldsymbol{\eta}_u]_{j}\Big[\thirdPopSMoMEst\Big]_{h, i, j}
				\end{equation*}
				\State \Longunderstack[l]{Determine the (eigenvalue, eigenvector) pairs\\ $\Big(\hat{\lambda}^{(\cmpIdx)}_u, \hat{\mathbf{v}}^{(\cmpIdx)}_u\Big), \cmpIdx\in[\numCmp]$, of $\hat{\W}^\top\thirdPopSMoMEstVec{\boldsymbol{\eta}_u}\hat{\W}$}
				%\State Find $u^\ast = \arg_{u\in[U]}\max\Big(\min_{\cmpIdx_1, \cmpIdx_2\in[\numCmp]}\Big(\Big\{\big|\hat{\lambda}^{(\cmpIdx_1)}_u-\hat{\lambda}^{(\cmpIdx_2)}_u\big|:\cmpIdx_1\neq \cmpIdx_2\Big\}\cup\Big\{\big|\hat{\lambda}^{(\cmpIdx_1)}\big|\Big\}\Big)\Big)$
				\State Compute $\vecCmpExpEst = \Big(\boldsymbol{\eta}_{\old{t}\newMinor{u}}^\top\hat{\B}\hat{\mathbf{v}}^{(\cmpIdx)}_{\old{t}\newMinor{u}}\Big)^{-1}\hat{\lambda}^{(\cmpIdx)}_{\old{t}\newMinor{u}}\hat{\B}\hat{\mathbf{v}}^{(\cmpIdx)}_{\old{t}\newMinor{u}}$, $\cmpIdx \in[\numCmp]$
				\State Determine $\cmpWgtVecEst_u = \Big[\vecCmpExpEst[1],\dots, \vecCmpExpEst[\numCmp]\Big]^\dagger\vecExpEstSet{\pvecSetOne}$%, $\cmpIdx\in[\numCmp]$
				%\State Find $\cmpVarEst\!=\!\frac{1}{\cmpWgtEst}\baseVec[\cmpIdx]^\top\bigg(\Big[\vecCmpExpEst[1],\dots, \vecCmpExpEst[\numCmp]\Big]^\dagger\!\firstPopSMoMEst{\pvecSetOne}\!\bigg)$, $\cmpIdx \in\numCmp$
				\State Obtain $\betaMixAEst[\vecElIdx, u]=\frac{\vecCmpExpEst}{1-\vecCmpExpEst},\forall\,\vecElIdx\in[\numVecEl],\forall\,\cmpIdx\in[\numCmp]$
				\EndFor
			\end{algorithmic}

		\end{algorithm}			
		The best density fit is determined based on a difference measure $\disFct{\mixPPdfEst[u]}{\pHist}$ between model estimates $\mixPPdfEst[u]$ and the histogram of the observed data. We ran Alg.~\ref{alg:spa-lfdr-est} with some of the most popular measures for quantifying closeness based on probability densities and distributions. We found Alg.~\ref{alg:spa-lfdr-est} to be robust w.r.t. the selected difference measure in terms of \gls{fdr} control. We observed that using \gls{edf}-based distances such as the \gls{wsd} or \gls{ksd} distance resulted in slightly higher detection power for very small $\nulFrc$ than \gls{pdf}-based divergences, such as the \gls{kld} or the \gls{jsd} divergence. Thus, we stick to \gls{edf} distances.
%		
%		Model-overfitting, i.e., selecting always the largest candidate model order due to marginal improvements with every added model component, can be prevented by a couple of techniques. Particularly popular are the previously mentioned \gls{itc}, which account for increasing model complexity by a model-order-dependent penalty term. However, standard \gls{itc} like the \gls{bic} or the \gls{aic} rely on \gls{mle}s for the model parameters and are thus not applicable for \gls{mom}-based density fitting. This problem has been noted in the literature. The authors of \cite{Schroeder2017}, who deployed an iterative \gls{em}-like algorithm for two-parameter beta mixtures with moment parameter estimators in the M-step, proposed to threshold the \gls{ksd} between data and density fit.

		Our method inherently avoids overfitting. The model order is limited by the number of elements per subset, $\numCmp<\numVecEl$. Hence, increasing $\numVecEl$ adds flexibility and reduces the mismatch $\thirdPopSMoMErr$ between the observable and non-observable theoretical third-order moment dependent terms. Contrarily, increasing $\numVecEl$ decreases the validity of Assumption~\ref{ass:same-var}, i.e., equal variance among the multivariate mixture components. In addition, an increase in $\numVecEl$ reduces the accuracy of the model parameter estimates due to the decrease in sample size $\numTls = \lfloor\old{\numPx}\newMinor{\numSen}/\numVecEl\rfloor$ and increase in the number of parameters to be estimated. % \cite[Theorem~3]{Hsu2012}.

	\vspace{-4pt}
\section{Interpolation of local false discovery rates}
\vspace{-1pt}
\label{sec:lfdr-ipl}
	In this section, we discuss how to determine the decision statistics between the spatially sparse sensor locations. Like in Fig.~\ref{fig:ex-fd-lfdr}, sensors are typically positioned at a subset of locations $\{\crdSen\}_{\senIdx\in\numSen}, \numSen\leq\numPx$. These report local \pval s $\pSetSen = \{\pPx[1], \dots, \pPx[\numSen]\}$ to the \gls{fc}. The \gls{fc} computes sensor-level \gls{lfdr} estimates $\lfdrLocHat[\senIdx], \forall\,\senIdx\in[\numSen]$ using \gls{lfdr}-\gls{smom}. \old{Finally, }The \gls{lfdr}'s are interpolated \old{by \gls{rbf} interpolation with \gls{tps}} to estimate the \gls{lfdr}'s $\lfdrLocHat[\pxIdx]$ at locations $\{\crdPx\}$ between sensors, $\numSen<\pxIdx\leq\numPx$. \old{Finally, the region of interesting behavior is formed with Eq.~\eqref{eq:alt-reg-est}}\newMajor{The lfdr’s are unknown deterministic quantities. Thus, we deploy a deterministic interpolation method. \Gls{rbf} interpolation is an advanced mesh-free method to reconstruct an unknown deterministic function from observed data \cite{Schaback2007, Fasshauer2007}. The function value at a location of interest is calculated as a weighted sum of smooth basis functions whose value depends on the location's distance to the sampling locations. \Gls{rbf} is well suited to the problem at hand: The sensors can be located at arbitrary locations within the observation area. The interpolant is stable also for a large number of sensors \cite{Keller2019}. Finally, it produces a smooth interpolant, whose properties depend on the deployed radial basis function. The properties of \gls{tps} \cite{Duchon1977} fit particularly well to our central assumption of spatial smoothness, i.e., that the null and alternative regions are formed of locally continuous sub-regions with quick transitions in between. \gls{tps} finds values in between sampling points under a constraint on the energy of the	interpolant. Thus, it produces a very smooth interpolant with sharp transitions \cite{Eberly2018} in between sensor locations where the lfdr’s differ significantly. In addition, \gls{tps} \gls{rbf}s does not require ad-hoc tuning of additional parameters. This would be necessary, if other popular basis functions, such as the Gaussian or multiquadric were used.}
	
	The interpolated \gls{lfdr}'s are calculated as
	\begin{equation}
	\label{eq:lfdr-ipl}
		\lfdrLocHat[\pxIdx] = \sum_{\senIdx=1}^{\numSen}\alpha_\senIdx \varphi(\radius{\senIdx}{\pxIdx}), \qquad\pxIdx\in[\numPx],
	\end{equation}
	where $\radius{\senIdx}{\pxIdx}=||\crdSen-\crdPx||_2$ is the Euclidean distance between sensor and location, $\varphi(\radius{\senIdx}{\pxIdx}) = \radius{\senIdx}{\pxIdx}^2\ln(\radius{\senIdx}{\pxIdx})$ is the \gls{tps} \gls{rbf} and the weights $\alpha_\senIdx\in\mathbb{R}$ are determined based on the estimated sensor \gls{lfdr}'s and locations by solving \old{\cite[Chapter~7]{Fasshauer2007}}\newMinor{\cite[Chapter~8.3]{Fasshauer2007}}
	\begin{align*}
		&\begin{bmatrix}
			\varphi(\radius{1}{1}) & \cdots & \varphi(\radius{1}{\numSen})\\
			\vdots & \ddots & \vdots \\
			\varphi(\radius{\numSen}{1}) & \cdots & \varphi(\radius{\numSen}{1})
		\end{bmatrix}
		\begin{bmatrix}
			\alpha_1\\
			\vdots\\
			\alpha_\numSen
		\end{bmatrix}\\
		&\qquad\qquad\qquad+
		\old{
			\begin{bmatrix}
				\beta_1 \\
				\vdots \\
				\beta_\numSen
			\end{bmatrix}
		}
		\newMajor{
				\begin{bmatrix}
				1 & \xSen[1] & \ySen[1]\\
				\vdots & \vdots & \vdots \\
				1 & \xSen[\numSen] & \ySen[\numSen]
			\end{bmatrix}
			\begin{bmatrix}
				\beta_1 \\
				\beta_2 \\
				\beta_3
		\end{bmatrix}
		}
		=
		\begin{bmatrix}
			\lfdrLocHat[1]\\
			\vdots\\
			\lfdrLocHat[\numSen]
		\end{bmatrix},
	\end{align*}
	\old{such that $\sum_{\senIdx = 1}^\numSen\beta_\senIdx\alpha_\senIdx = 0 \,\forall\, \beta_\senIdx\in\mathbb{R}, \senIdx\in[\numSen]$.}\newMajor{with regularization parameters such that $\sum_{\senIdx=1}^\numSen\alpha_\senIdx = 0$, $\sum_{\senIdx=1}^\numSen\alpha_\senIdx\xSen = 0$, and $\sum_{\senIdx=1}^\numSen\alpha_\senIdx\ySen = 0$.} % $[\beta_1,\dots,\beta_\numSen][\alpha_1,\dots, \alpha_\numSen]^\top = 0$
	\newMajor{\section{Communication Cost and Power Consumption}
\label{sec:cc}
In practice, the $p$-values have to be quantized before the transmission over the wireless communication channel. Only few bits will suffice so that no significant performance loss is experienced. In this section, we quantitatively analyze the total network communication cost $\mathrm{CC}_N$ and communication-related power consumption $\mathrm{PC}_N$ for a \gls{wsn} of size $N$ w.r.t. the number of bits that each sensor is allowed to transmit. The $\mathrm{CC}_N$ and $\mathrm{PC}_N$ are expressed quantitatively as
\begin{align}
	\label{eq:cc}
	\mathrm{CC}_N &= \sum_{n = 1}^{N}\big(\mathrm{CF} + \mathrm{CB}(B)\big) \cdot \mathrm{NP}_n,\\
	\mathrm{PC}_N &= \sum_{n = 1}^N\big(\mathrm{NP}_n \cdot \mathrm{CT}(B)\big),
\end{align}

where $\mathrm{NT}_n$ is the number of packets that sensor $n\in[N]$ transmits, i.e., the number of transmissions per sensor. For simplicity, we assume that each transmission of sensor $n$ is equally costly. The total cost per transmission depends on $\mathrm{CF}$ which represents the payload-independent cost of the deployed transmission protocol and the cost $\mathrm{CB}(B)$ that depends on the number of payload bits $B$. $\mathrm{CT}(B)$ is the power consumption per transmitted packet with $B$ payload bits. For simplicity, we assume that all costs and the number of information-carrying bits are identical for all sensors. Transmitting data is the most power-intensive task for the sensors in our \gls{wsn}. Simple local computations to determine the soft decision statistics consume much less power.

$\mathrm{CC}_N$ and $\mathrm{PC}_N$ of our proposed method with $\mathrm{NP}_n = 1\, \forall\,n\in[N]$ is considerably different from a fully centralized approach where each local measurement is communicated to the \gls{fc}. For example, assume that each sensor records a total number of $256$ local measurements. If the $B$ is the same for the $p$-values and the raw measurements, a centralized approach consumes $256$ times more communication bandwidth and power than the proposed approach.

Consider the following practical example for the relevance of the number of payload bits. A simple and very important internet protocol (IP) is the user datagram protocol (UDP). An empty, i.e., no payload, UDP packet consumes $\mathrm{CF} = 28$ bytes. In addition, UDP transmits bytes instead of single bits, i.e., $\mathrm{CC}_N$ does not grow linear in the number of payload bits $B$, bit is constant for $B = 1$ to $B = 8$. Hence, if one sends $1$, $4$ or $8$ bits of payload data, the frame structure overhead is well above $90\%$ of the total communication cost and there is no difference if $1$, $4$ or $8$ bits of information is transmitted. 

The best way to reduce the communication cost and power consumption is to limit the total number of transmissions in the WSN. This could be achieved by imposing communication constraints and transmitting only informative data as is done in censoring \cite{Tay2007, Blum2007}. Our simulation results demonstrate that the proposed method also works if $p$-values above a certain threshold $\lambda\leq 1$ are censored, i.e., if a sensor $n\in[N]$ may only transmit its quantized $p$-value $p^{\text{q}}_n$ if $p^{\text{q}}_n>\lambda$. No censoring occurs if $\lambda = 1$. 
}
	\section{Simulation Results}
\label{sec:sim-res}
	\begin{figure}
		\centering
		\begin{subfigure}{.24\textwidth}
			\centering
			\includegraphics[scale=0.45]{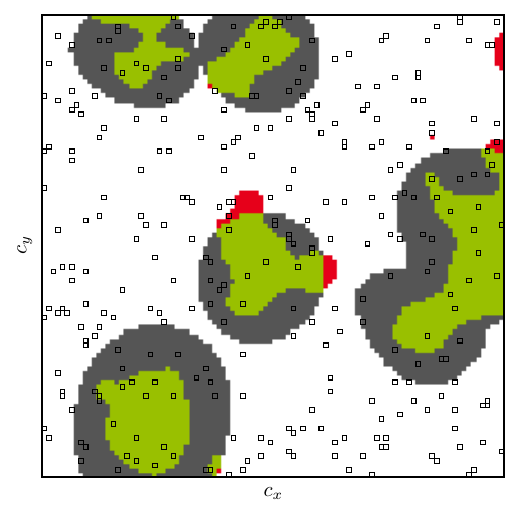}
			%\caption{\gls{fdr}}
			\caption{Detection pattern}
		\end{subfigure}
		\begin{subfigure}{.24\textwidth}
			\centering
			\includegraphics[scale=0.487]{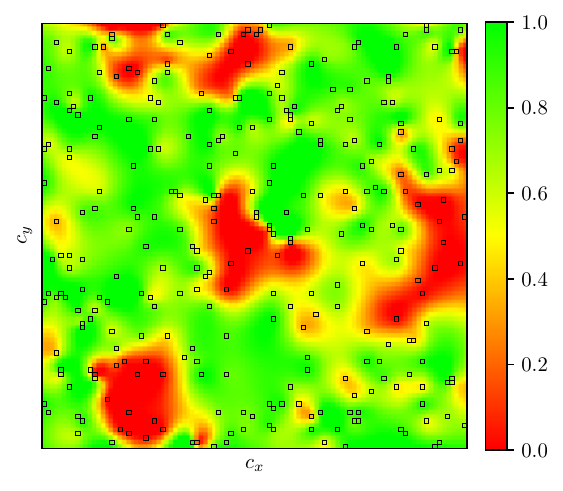}
			\caption{Interpolated \gls{lfdr}'s}
		\end{subfigure}
		\caption{Example for \old{Sc.~B}\newMinor{ScB}, Cnfg.~2. In a), the \gls{fdr} threshold is set to $\fdrThr=10\%$. Green, red and gray indicate true, false and missed discoveries. Boxes mark sensor locations. The interpolated \gls{lfdr}'s of b) are discussed in Sec.~\ref{sec:sim-res_spa-ipl}.}
		\label{fig:ex-fd-lfdr}
		\vspace{-15pt}
	\end{figure}
	We evaluate the performance of our proposed \newMinor{spatial inference} method on simulated radio frequency electromagnetic field data. The observed signals are simulated by the nonuniform sampling method from \cite{Cai2003}, which models the propagation of radio waves in a 2D spatial area with path loss and shadow fading \cite{Molisch2012}. The observations are subject to additive white Gaussian sensor noise. If the power of the received signal is below the noise floor, this sensor is classified as located in the true null region $\nulReg$. \newMinor{If no sensor is present at a grid point, this grid point is assigned to $\nulReg$ if the signal level of the field at this grid point is so low that it would be below the noise floor of a reference sensor placed at this location.} In our simulations, we alter the sensor noise power for different scenarios to obtain different sizes $|\altReg| = 1-\nulFrc$ of the region\newMinor{s of interesting, different or anomalous behavior}\old{ of interest}. %To obtain the ground true null and alternative regions $\HNul$ and $\HAlt$, we clip the signal below a certain The \gls{snr} is chosen such that the true desired alternative region To obtain the ground true null and alternative regions $\nulReg$ and $\altReg$, we clip the signal below a scenario-dependent threshold

	For simplicity, the sensor \pval s are computed from signal energies. This facilitates the analysis of the results, since its known distributions under null and alternative hypotheses allow for benchmarking the \gls{lfdr} estimation techniques against the true \gls{lfdr}'s. Also, it offers a simple way to simulate a heterogeneous sensor network by alternating the number of measurements $\numSam$ at different nodes. In practice, any type of sufficient test statistic could be deployed at the sensor level.

	We evaluated our method for a variety of simulated propagating radio wave fields in different environments and sensor network configurations. We discuss \newMinor{the} results for three scenarios typical for radio frequency sensing. %The spatial area of interest is a square of $100\times 100$ pixels, i.e., the total number of pixels is $\numPx=10\,000$.
	The monitored area is discretized by a grid of $100 \times 100$ spatial elements for all scenarios and the results are averaged over $200$ independent Monte Carlo runs, unless stated otherwise. The sources are placed at random locations, i.e., the fraction $\nulFrc$ of grid points in the null region varies slightly from run to run.
	%\textbf{\old{Sc.~A}\newMinor{ScA}}: A single source in an urban environment. The average relative size of $\altReg$ is $70\%$. This scenario is moderately challenging, since many of the pixels in the alternative region $\altReg$ are exposed to a strong signal. % are exposed for the \zval-based competitors, since many of the observed \zval s stem from the tail of the alternative \gls{pdf} $\altZPdf$ which facilitates its estimation.  %A large number of pixels is under the influence of a single source, i.e., the received signal power is large at many locations in the alternative region. Hence, many observations stem from the tails of the alternative \zval~distribution $\altZPdf$, which facilitates its estimation.

	\textbf{\old{Sc.~A}\newMinor{ScA}}: Five sources located in a suburban environment covering on average $1-\nulFrc \approx 55\%$ of grid points.
	
	\textbf{\old{Sc.~B}\newMinor{ScB}}: Eight sources located in a suburban environment covering on average $1-\nulFrc \approx 34\%$ of grid points. %More sources cover less pixels than in \old{Sc.~A}\newMinor{ScA}, i.e., the observation area in \old{Sc.~B}\newMinor{ScB} is larger and the spectral resolution of the grid is lower.
	
	 % This scenario is more challenging for the competitors than \old{Sc.~A}\newMinor{ScA}, since the pixels in the alternative region $\altReg$ are exposed to a weaker signal. The corresponding local summary statistics are hence less distinct from what would be expected in $\nulReg$. %.Many of the observed \zval s stem from $\altZPdf$, but the tails are less pronounced. %A large number of pixels is under the influence of the sources, but due to the propagation model,  since a significant amount of \pval s is generated under $\HAlt$. In contrast to Thus, the alternative a large the joint \gls{pdf} i.e., the 
	\textbf{\old{Sc.~C}\newMinor{ScC}}: Two sources in a suburban environment covering on average $1-\nulFrc\approx 10\%$ of grid points.
	
	{\old{Sc.~C}\newMinor{ScC}} is particularly challenging, regardless of the deployed method. Only a small proportion of test statistics provide information on the shape of the alternative component $\altPPdf$.%In addition, a small region of interesting/different/anomalous behavior $\altReg$ makes violations of the nominal \gls{fdr} level more likely, if individual \gls{lfdr}'s of locations in $\nulReg$ are underestimated. 
	
	We investigate \old{three}\newMajor{four} different sensor network configurations.
	
	\textbf{Cnfg.~1}: $\numSen = 10\,000$ identical sensors with $\numSam=256$. Decisions are only made at sensor locations, i.e., no interpolation of decision statistics. This is the classic \gls{mht} problem.
	%\textbf{Cnfg.~1}: Fully loaded. Each grid point hosts a sensor. This configuration of the sensor network is unrealistic in practice, but corresponds to the classic \gls{mht} setting where each hypothesis test is based on a local summary statistic.

	\textbf{Cnfg.~2}: $\numSen=300$ identical sensors with $\numSam = 256$, homogeneously distributed across the monitored area. The tests at the sensors are based on local summary statistics, but in between sensors, the \gls{lfdr}'s are interpolated.
	
	\textbf{Cnfg.~3}: $\numSen_1=170$ sensors with $\numSam_1= 256$, $\numSen=80$ sensors with $\numSam_2= 512$ and $\numSen_3=50$ sensors $\numSam_3 = 1024$, all types homogeneously distributed across the monitored area. Decisions in \old{areas }between sensors base upon interpolated \gls{lfdr}'s.
	
	\newMajor{\textbf{Cnfg.~4}: $\numSen = 2\,000$ identical sensors with $\numSam=256$, homogeneously distributed across the monitored area. The decisions at the sensors are based on local summary statistics, but in between sensors, the \gls{lfdr}'s are interpolated.}
	
	\textbf{Competitors}: \newMajor{To the best of our knowledge, there exist no other methods to identify the regions of anomaly with guarantees on error probabilities or the \gls{fdr} when sensors are located at a sparse but arbitrary subset of grid points. In \cite{Ermis2010}, the authors consider the classic \gls{mht} problem for \glspl{wsn}, hence, we can compare our proposed lfdr-based inference approach to their \gls{dbh} procedure \cite{Ermis2010} for a \gls{wsn} in Cnfg.~1. In addition, we compare the proposed lfdr-based inference approach when different lfdr-estimators are used.} We compare our \newMinor{proposed lfdr-sMoM}\old{method} to a variety of \old{widely used}\newMinor{popular} \gls{lfdr} estimators. We deploy the classic \gls{bum} model-based \gls{mle} approach from \cite{Pounds2003}. We also show results obtained with \gls{lm} as proposed by Efron \cite{Efron2010}. \Gls{lm} which approximates the \zval~\gls{pdf} by an exponential family model, fitted to the data using Poisson regression. In addition, \gls{pr} \cite{Newton2002, Martin2011} is applied, which computes the alternative \zval~\gls{pdf} by estimating the density of the mean shift of \zval s from locations in $\altReg$. Predictive recursion has recently \cite{Scott2015, Martin2018, Tansey2018} gained considerable attention in \gls{lfdr} estimation, due to its high accuracy and comparably low computation time. Our implementation of \gls{pr} follows \cite[Appendix~A]{Scott2015}. Finally, we include a standard \gls{gmm}. We do not consider methods whose computational complexity prevents scaling to large-scale sensor networks, as kernel density \gls{lfdr} estimators \cite{Robin2007}.
%	\vspace{-5pt}
	\subsection{The classic multiple hypothesis testing problem}
		\begin{figure}
			\centering
			% For single column!
			\begin{subfigure}{.49\textwidth}
				\begin{subfigure}{.48\textwidth}
					\centering
					\includegraphics[scale=0.42]{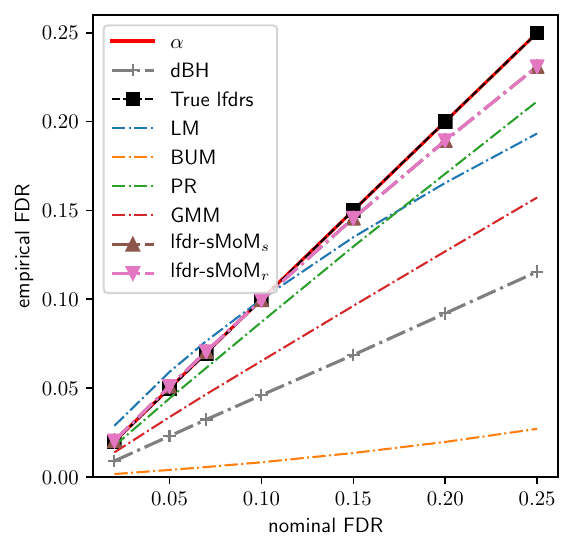}
			% For double column!
%			\begin{subfigure}{.5\textwidth}
%				\begin{subfigure}{.48\textwidth}
%					\centering
%					\includegraphics[scale=0.45]{figures/sc21_fdr.pdf}
				\end{subfigure}
				\begin{subfigure}{.48\textwidth}
					\centering
					% For single column!
					\includegraphics[scale=0.42]{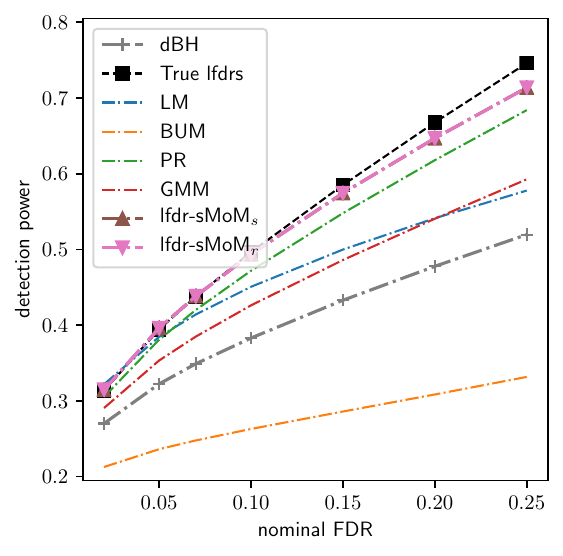}
					% For double column!
%					\includegraphics[scale=0.45]{figures/sc21_pow.pdf}
				\end{subfigure}
			\caption{\old{Sc.~A}\newMinor{ScA}}
			\label{fig:sca-cnfg1}
			\end{subfigure}
			% For double column!
%			\begin{subfigure}{.5\textwidth}
%				\begin{subfigure}{.48\textwidth}
%					\centering
%					\includegraphics[scale=0.45]{figures/sc19_fdr.pdf}
			% For single column!
			\begin{subfigure}{.49\textwidth}
				\begin{subfigure}{.48\textwidth}
					\centering
					\includegraphics[scale=0.42]{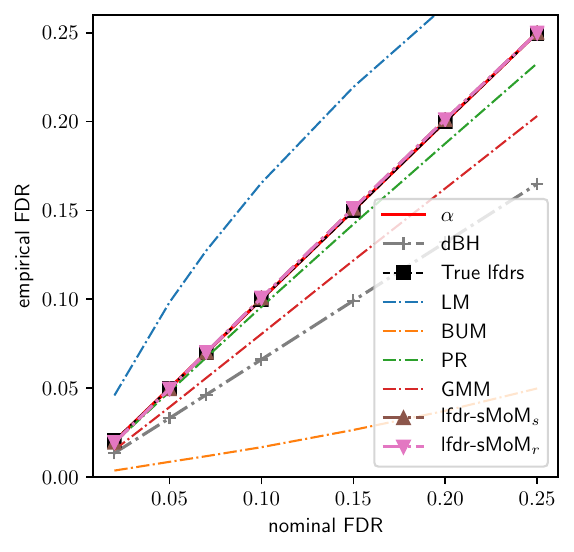}
				\end{subfigure}
				\begin{subfigure}{.48\textwidth}
					\centering
					% For single column!
					\includegraphics[scale=0.42]{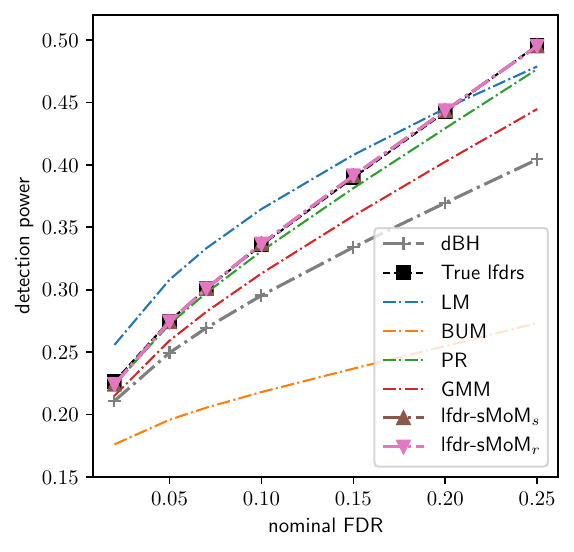}
					% For double column!
%					\includegraphics[scale=0.45]{figures/sc19_pow.pdf}
				\end{subfigure}
				\caption{\old{Sc.~B}\newMinor{ScB}}	
				\label{fig:scb-cnfg1}		
			\end{subfigure}
			% For double column!
%			\begin{subfigure}{.5\textwidth}
%				\begin{subfigure}{.48\textwidth}
%					\centering
%					\includegraphics[scale=0.45]{figures/sc18_fdr.pdf}
			% For single column!
			\begin{subfigure}{.49\textwidth}
				\begin{subfigure}{.48\textwidth}
					\centering
					\includegraphics[scale=0.42]{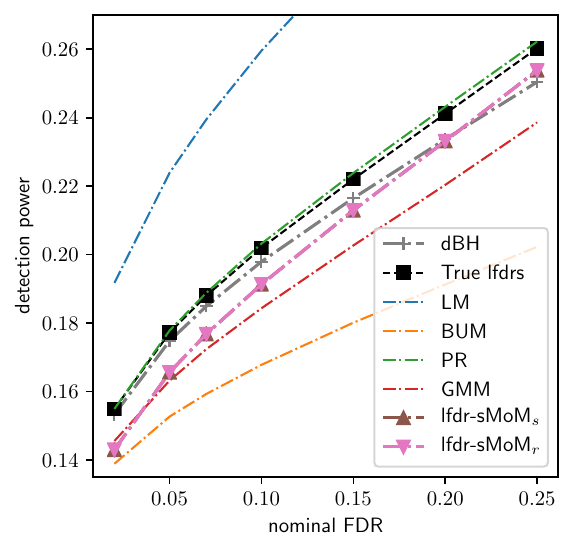}
				\end{subfigure}
				\begin{subfigure}{.48\textwidth}
					\centering
					% For double column!
%					\includegraphics[scale=0.45]{figures/sc18_pow.pdf}
					% For single column
					\includegraphics[scale=0.42]{figures/sc18_pow.pdf}
					%\caption{Detection Power}
				\end{subfigure}
				\caption{\old{Sc.~C}\newMinor{ScC}}	
				\label{fig:scc-cnfg1}		
			\end{subfigure}
			\caption{\Gls{fdr} and detection power for all scenarios and a sensor network composed of $N=10\,000$ nodes. \old{Inference is only performed}\newMinor{Decisions are only made} at the sensor locations (Cnfg.~1). The \gls{fdr} is strictly controlled for all scenarios \old{by all methods}\newMinor{with all estimators} except \gls{lm}. Among \old{the \gls{fdr} controlling methods,}\newMinor{those for which the \gls{fdr} is controlled, the proposed lfdr-sMoM}\old{ ours provides } \newMinor{yields} the highest power in \old{Sc.~A}\newMinor{ScA, ScB}. In \old{Sc.~C}\newMinor{ScC}, where $\altReg$ contains only a small fraction of all nodes, \old{\gls{pr}'s}\newMinor{the} power \newMinor{with PR} is slightly higher.}
			\label{fig:cngf1}
			\vspace{-13pt}
		\end{figure}
		The results for Cnfg.~1 are shown in Fig.~\ref{fig:cngf1}. \newMajor{\gls{dbh} works well in ScC where $\nulFrc\approx10\%$ is small. When $\nulFrc$ gets larger, \gls{dbh} lacks detection power. The results for the different lfdr estimators with the proposed lfdr-based spatial inference method are ambiguous. We benchmark by the detection results obtained when using the (in practice unknown) true \gls{lfdr}'s. A higher detection power than with the true \gls{lfdr}'s can only be achieved if the respective lfdr estimator leads to a violation of the nominal \gls{fdr} level. }\Gls{lm} faces stability issues due to the heavy one-sided tail of $\mixZPdf$ when the relative size of the true null region $\nulFrc = |\nulReg|/\numSen$ increases. \old{All}\newMinor{With all} other \newMinor{lfdr estimators, the considered nominal \gls{fdr} levels are met. }\old{applied methods meet the nominal \gls{fdr} level for all scenarios. }\old{We benchmark by the detection results obtained when using the (in practice unknown) true \gls{lfdr}'s. A higher detection than with the true \gls{lfdr}'s can only be achieved, if the respective method violates the nominal \gls{fdr} level. }For \old{our proposed method}\newMinor{the proposed lfdr estimator}, we show results for two variants. For \gls{lfdr}-\gls{smom}${}_s$, the \pval~vectors are formed by subdividing the grid into square tiles of spatially close grid points. For \gls{lfdr}-\gls{smom}${}_r$, the $\{\pPx[\senIdx]\}_{\senIdx\in[\numSen]}$ are randomly partitioned into subsets.

		The results underline that our proposed multivariate \pval~vector probability density model is very flexible. Even if the \pval~subsets are formed at random, the method finds a parametrization such that the individual components have equal variance. This confirms that Assumption~\ref{ass:same-var} can be relaxed in practice.
		For the remainder of this section, we stick to the conceptually simpler random partitioning and label results obtained with randomly formed \pval~vectors by \textit{\gls{lfdr}-\gls{smom}}.
		
		\newMinor{For the}\old{The} traditional \gls{bum} \newMinor{estimator}\old{method controls the \gls{fdr}, but exhibits a low detection power}\newMinor{, the lfdr is controlled but the detection power is low} in all scenarios, as \old{does}\newMinor{for} \gls{gmm} for larger interesting regions. \gls{pr} performs best by a small margin when the relative size of the alternative region is very small (\old{Sc.~C}\newMinor{ScC}), but \old{our proposed method achieves}\newMinor{with our proposed estimator,} the largest detection power \newMinor{is achieved in ScB and ScA}\old{Sc.~A and Sc.~B}. Thus, \gls{lfdr}-\gls{smom} provides the best or very close to the best results \old{in }\newMinor{for }all \old{of the }considered scenarios.
		
		\begin{figure}
			\centering
			% For double column!
			\begin{subfigure}{.28\textwidth}
			% For single column!				
%			\begin{subfigure}{.48\textwidth}
				\centering
				\includegraphics[scale=0.45]{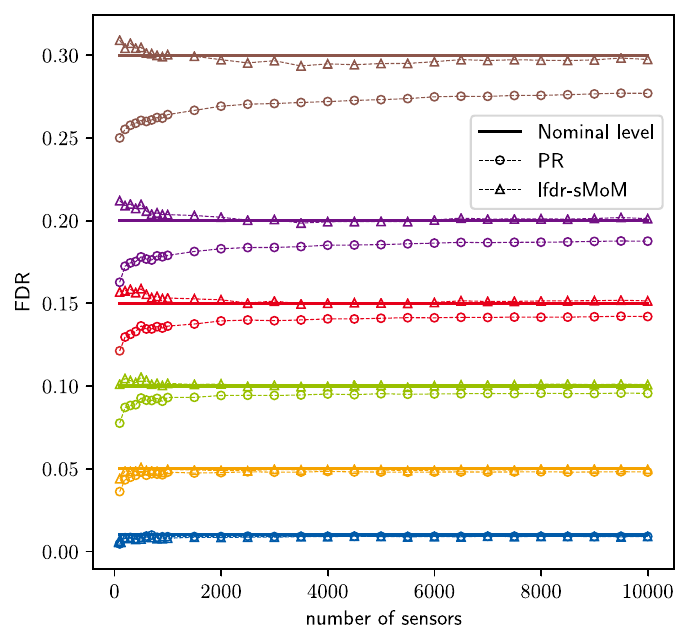}
				\caption{\gls{fdr}}
			\end{subfigure}
			% For double column!
			\begin{subfigure}{.18\textwidth}
			% For single column!			
%			\begin{subfigure}{.36\textwidth}	
				\centering
				\begin{subfigure}{\textwidth}
					\centering
					\includegraphics[scale=0.45]{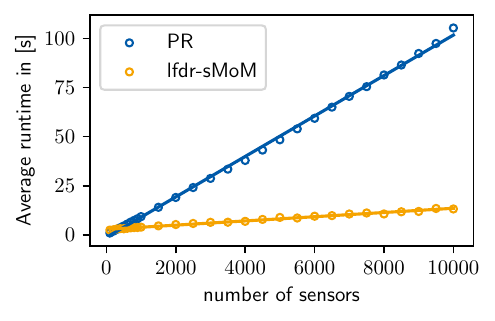}
				\end{subfigure}
				\begin{subfigure}{\textwidth}
					\centering
					\includegraphics[scale=0.45]{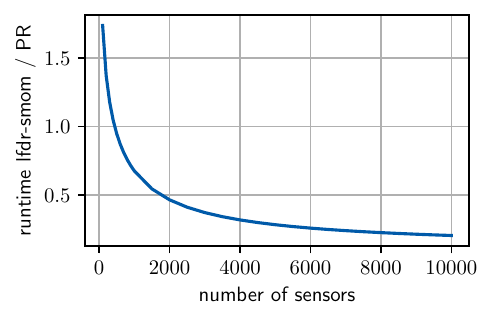}
				\end{subfigure}
				\caption{Execution time}
				\label{fig:rt-b}
			\end{subfigure}	
			\caption{\Gls{fdr}s and run times with increasing sensor network size for \gls{lfdr}-\gls{smom} vs. the best competitor (\gls{pr}) \newMinor{for ScB}.}
			\label{fig:rt}
			\vspace{-10pt}
		\end{figure}
	\vspace{-7pt}
	\subsection{Computational demands as the network size increases}
		In Fig.~\ref{fig:rt}, we compare the best competitor, \gls{pr}, to our proposed method for different sizes of the sensor network $N$ and $1\,000$ \gls{mc} runs. We obtained almost identical results for all scenarios and hence show only those for \old{Sc.~B}\newMinor{ScB}. Again, inference is only performed at the sensor locations. Both, \newMinor{the proposed lfdr-sMoM} and \gls{pr}, are subject to transient effects when the number of nodes is small. This was to be expected, as fitting a \old{heavy-tailed}\newMinor{complicated} distribution \newMinor{model} with such a low number of data points is challenging. Yet, the obtained \gls{fdr}s are close to the nominal level also for the smallest considered network sizes $\numSen$. As $\numSen$ grows, the nominal and empirical \gls{fdr} for \gls{lfdr}-\gls{smom} coincide almost perfectly. Hence, \gls{lfdr}-\gls{smom} is more efficient in exploiting all permitted false positives than \gls{pr}, which results in a larger detection power (plot not shown \newMinor{due to space limitations}). The comparison of the average execution time per Monte Carlo run in Fig.~\ref{fig:rt-b} illustrates that \old{our method}\newMinor{the proposed estimator} is considerably faster than \gls{pr} as $N$ increases. The upper plot shows that the execution time for both methods grows approximately linearly. For a small number of sensors, the runtimes of \gls{lfdr}-\gls{smom} and \gls{pr} are almost identical, but \old{our method}\newMinor{the former} scales significantly better for larger sensor networks. At $N=10\,000$, \gls{lfdr}-\gls{smom} is more than five times faster. Both methods were run in Python 3.8.3 on an AMD Ryzen~9 3900X 12-Core CPU.
		\newMinor{We conclude that while}\old{While} \gls{lfdr}-\gls{smom} often provides the highest detection power, it also outperforms its strongest competitor significantly in terms of computation time for growing $\numSen$. \old{\Gls{lfdr}-\gls{smom} is the best fit for large-scale sensor networks.}

	\subsection{Spatial interpolation of \gls{lfdr}'s}
		\label{sec:sim-res_spa-ipl}
		It is of high interest to find the boundaries of the regions associated with interesting, anomalous or different behavior. \newMinor{Our proposed lfdr-based spatial inference approach determines the values in between the sensor locations by spatial interpolation}\old{The \gls{lfdr} values between the sensor locations can be determined through spatial interpolation}. This allows for segmenting the observation area into $\nulReg$ and $\altReg$ \newMinor{with \gls{fdr} control}. %Finally, we examine Cnfg.~2 and Cnfg.~3, which involve the interpolation of \gls{lfdr}'s at locations between sensor to divide the entire observation area into null and alternative regions.
		
		An exemplary map of interpolated \gls{lfdr}'s with $\numSen = 300$ identical sensors (Confg.~2). is shown in Fig.~\ref{fig:ex-fd-lfdr}, along the detection pattern for the most commonly used \cite{Efron2010} nominal \gls{fdr} level $\fdrThr = 0.1$. In Tab.~\ref{tab:cnfg-2}, we present the numerical \gls{fdr}s and detection powers  obtained by interpolating the estimated sensor \gls{lfdr}'s. Note, that \newMinor{the interpolation step of the proposed inference approach is independent of the selected lfdr estimator, see also Fig.~\ref{fig:flowchart}.}\old{interpolating the \gls{lfdr}'s is a novel approach and could also be combined with any existing \gls{lfdr} estimator, such as \gls{pr}. }%Again, we compare our proposed estimator \gls{lfdr}-\gls{smom} to \gls{pr}.
		
		The \gls{fdr} is controlled at the nominal level, except for the very small $\fdrThr=0.01$. This nominal level is so small, that even the slightest interpolation error can lead to its violation. \newMinor{For higher, more realistic nominal \gls{fdr} levels, the FDR is strictly controlled}. In this example, \old{\gls{lfdr}'s are computed}\newMinor{sensors are located} at only $3\%$ of all grid points. \newMinor{The \gls{lfdr}'s are} interpolated at the remaining $97\%$. These results strongly indicate that the interpolation of \gls{lfdr}'s is a powerful tool to decide between $\HNul$ and $\HAlt$ at locations where no sensor is present and local summary statistics are not available.  
		
		In Cnfg.~3, we considered a heterogeneous sensor network composed of multiple types of sensors with different individual detection capabilities due to varying sensor noise levels. The results in Tab.~\ref{tab:cnfg-3} verify the applicability of our method to heterogeneous sensor networks, i.e., to accommodate various types of sensors into the inference process.
		\begin{table}
			\setlength{\tabcolsep}{5.5pt}
			\centering
			\caption{\old{Sc.~B}\newMinor{ScB} in Cnfg.~2. The columns indicate different nominal \gls{fdr} levels. The similarity in the empirical values of \gls{fdr} and detection power obtained when using \gls{lfdr}-\gls{smom} or in practice the unavailable true \gls{lfdr}'s as base for the interpolation underline the effectiveness of our method.}
			\begin{tabular}{*{9}{c}}
				& $\fdrThr$& $.01$ & $.05$	& $.1$ 	& $.15$ & $.2$ 	& $.25$ & $.3$\\\toprule
				\multirow{2}{*}{\gls{fdr}}		& \textit{True} & $\mathit{.020}$ & $\mathit{.030}$ & $\mathit{.05}$ & $\mathit{.08}$ & $\mathit{.12}$ 	& $\mathit{.17}$ & $\mathit{.22}$\\
				& \gls{lfdr}-\gls{smom} & $.029$ & $.043$	& $.07$ & $.10$ & $.15$ & $.19$ & $.24$	\\\midrule
				%& \gls{pr} & $.018$ & $.028$ & $.05$ & $.07$ & $.10$ 	& $.14$ & $.18$\\

				\multirow{2}{*}{Power}			& \textit{True} & $\mathit{.10}$ & $\mathit{.18}$ & $\mathit{.25}$ & $\mathit{.33}$ & $\mathit{.40}$& $\mathit{.47}$ & $\mathit{.55}$\\
				& \gls{lfdr}-\gls{smom} & $.11$ & $.19$	& $.26$	& $.34$ & $.41$	& $.48$ & $.55$	\\\bottomrule
			%	& \gls{pr} & $.10$ & $.17$	& $.24$	& $.30$ & $.37$	& $.43$ & $.50$\\

			\end{tabular}
			\label{tab:cnfg-2}
			\vspace{-7pt}
		\end{table}
		\begin{table}
			\centering
			\setlength{\tabcolsep}{5.5pt}
			\caption{\old{Sc.~B}\newMinor{ScB} in Cnfg.~3. The columns indicate different nominal \gls{fdr} levels. Also for a heterogeneous sensor network, the \gls{fdr} is controlled except for the very small $\fdrThr=.01$.}
			\begin{tabular}{*{9}{c}}
				& $\fdrThr$ & $.01$ & $.05$	& $.1$ 	& $.15$ & $.2$ 	& $.25$ & $.3$\\\toprule
				\multirow{1}{*}{\gls{fdr}}		& \gls{lfdr}-\gls{smom} & $.023$ & $.039$	& $.07$ & $.11$ & $.15$ & $.20$ & $.25$\\\midrule
				%& \gls{pr} & $.015$ & $.026$ & $.05$ & $.07$ & $.11$& $.15$ & $.19$\\
				
				\multirow{1}{*}{Power}			& \gls{lfdr}-\gls{smom} & $.15$ & $.24$	& $.33$	& $.42$ & $.49$	& $.56$ & $.63$\\\bottomrule
%				& \gls{pr} & $.14$ & $.22$	& $.30$	& $.38$ & $.45$	& $.52$ & $.59$\\
			\end{tabular}
			\label{tab:cnfg-3}
			\vspace{-10pt}
		\end{table}
	\newMajor{
		\subsection{Quantized $p$-values}
		\label{sec:sim-res_quan}
		In practice, the $p$-values have to be quantized prior to the transmission over the wireless communication channel. In this section, we demonstrate that the performance of our proposed spatial inference approach with lfdr-sMoM is close to the optimum when $p$-values are quantized with few bits. Since little support for the null hypothesis is indicated by small $p$-values, quantizers that provide	a higher resolution for smaller $p$-values and lower resolution for larger $p$-values are expected to be more efficient
		than for example uniform quantizers. For the purpose of illustration, we use the following ad-hoc $p$-value quantizer. Divide the $p$-value domain $[0, 1]$ with $B$ bits into $2^B$ intervals of width
		\begin{equation*}
			w_i = \frac{i}{\sum_{j=1}^{2^{B}} j}, \qquad i = 1, \dots, 2^{B}.
		\end{equation*}
		The left and right edges of the $i$-th quantization interval are
		\begin{equation*}
			\mathrm{le}_i	=\frac{\sum_{j=1}^{i-1} j}{\sum_{j=1}^{2^{B}} j},\quad
			\mathrm{re}_i	=\frac{\sum_{j=1}^i j}{\sum_{j=1}^{2^{B}} j},
		\end{equation*} respectively. 
		Fig.~\ref{fig:quan} shows the empirical FDR and detection power in Cnfg.~1, i.e., with a sensor located at each of the $10\,000$ grid points. We use $B \in[3, 5, 8]$ bits. Only the results with the true (unknown) lfdr's, the proposed lfdr estimator lfdr-sMoM and PR are shown.}
		
		\newMajor{The results underline that our proposed inference method works well with quantized $p$-values with only a negligible performance loss compared to unquantized case. We first analyze the performance with the true lfdr's, i.e., independent of the influence of an lfdr estimator. As $B$ increases, the detection power rises. Saturation is reached quickly, $B = 5$ bits provide nearly identical performance as in Fig.~\ref{fig:cngf1} with unquantized $p$-values. Only for the very challenging ScC and nominal FDR levels $\alpha < .15$, using more than $B = 5$ bits yields a noticeable increase in detection power. With the proposed lfdr-sMoM, the results are very similar: As few as $B = 5$ bits are enough for close to nominal performance at the commonly used nominal FDR level $\alpha= .1$ in ScA and ScB. In the very challenging ScC, a few bits are needed, but $B=8$ yields very good results. Again, a very small $B$ may lead to reduced detection power, but FDR control is strictly maintained with lfdr-sMoM. This is significantly different with PR. With PR, the nominal FDR level $\alpha$ is violated if $B$ is too small. Even for $B=8$, the nominal FDR level is violated in ScB and ScC. We obtained similar results for other $N$ and interpolated lfdr's.
		While these results confirm that our proposed method is applicable if the sensors report $p$-values that are quantized using few bits, these results also make a strong case in favor of the proposed lfdr estimator lfdr-sMoM for spatial inference with \glspl{wsn}.}
		\begin{figure}
			\centering
			% For double column
%			\begin{subfigure}{.5\textwidth}
%				\begin{subfigure}{.48\textwidth}
%					\centering
%					\includegraphics[scale=0.45]{figures/sc21_quan_fdr.pdf}
			% For single column
			\begin{subfigure}{.49\textwidth}
				\begin{subfigure}{.48\textwidth}
					\includegraphics[scale=0.42]{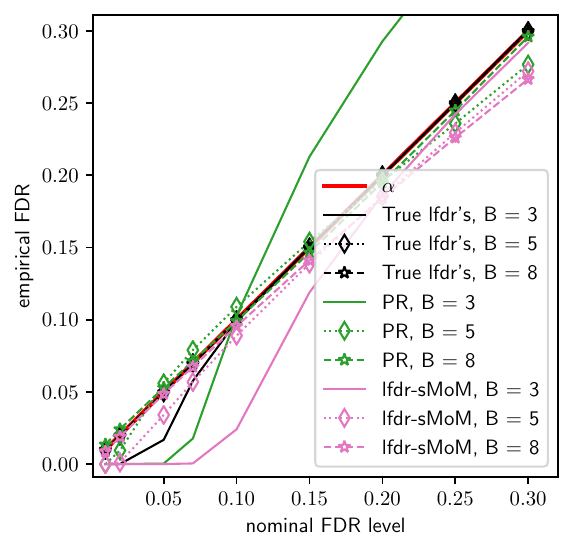}
				\end{subfigure}
				\begin{subfigure}{.48\textwidth}
					\centering
					% For doublee column
%					\includegraphics[scale=0.45]{figures/sc21_quan_pow.pdf}
					% For single column
					\includegraphics[scale=0.42]{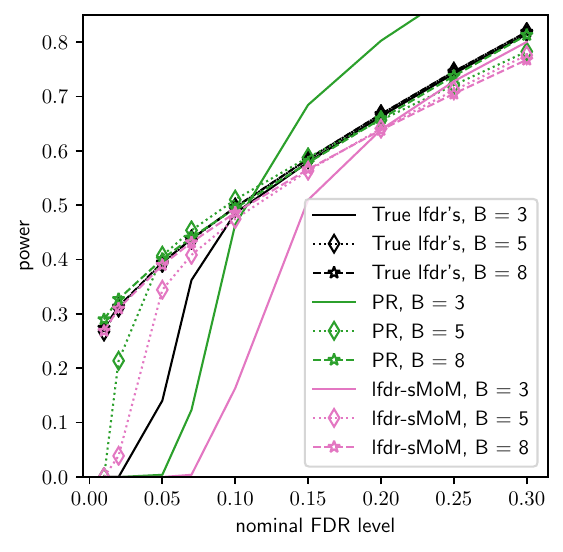}
				\end{subfigure}
				\caption{ScA}
				\label{fig:sca-cnfg4}
			\end{subfigure}
			% For double column
%			\begin{subfigure}{.5\textwidth}
%				\begin{subfigure}{.48\textwidth}
%					\centering
%					\includegraphics[scale=0.45]{figures/sc19_quan_fdr.pdf}
			% For single column
			\begin{subfigure}{.49\textwidth}
				\begin{subfigure}{.48\textwidth}
					\centering
					\includegraphics[scale=0.42]{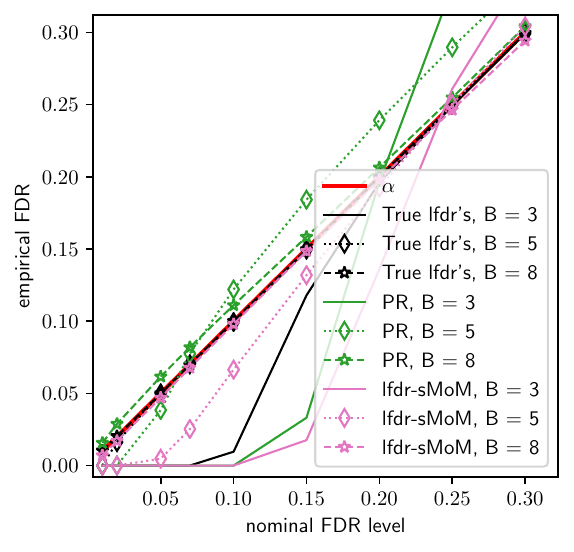}
				\end{subfigure}
				\begin{subfigure}{.48\textwidth}
					\centering
					% For double column
%					\includegraphics[scale=0.45]{figures/sc19_quan_pow.pdf}
					% For single column
					\includegraphics[scale=0.42]{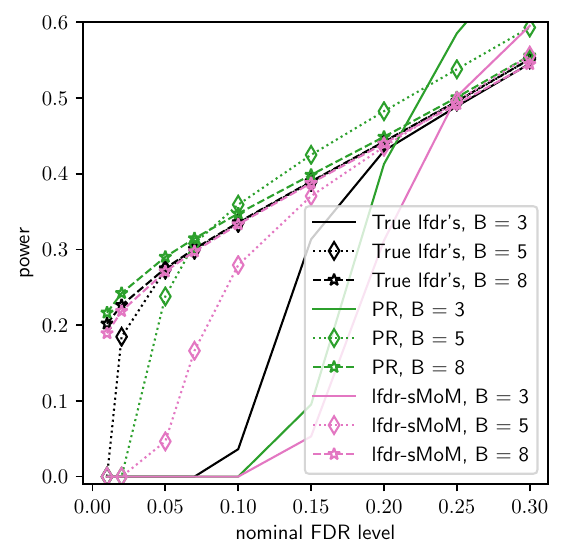}
				\end{subfigure}
				\caption{ScB}	
				\label{fig:scb-cnfg4}		
			\end{subfigure}
			% For double column
%			\begin{subfigure}{.5\textwidth}
%				\begin{subfigure}{.48\textwidth}
%					\centering
%					\includegraphics[scale=0.45]{figures/sc18_quan_fdr.pdf}
			% For single column
			\begin{subfigure}{.49\textwidth}
				\begin{subfigure}{.48\textwidth}
					\centering
					\includegraphics[scale=0.42]{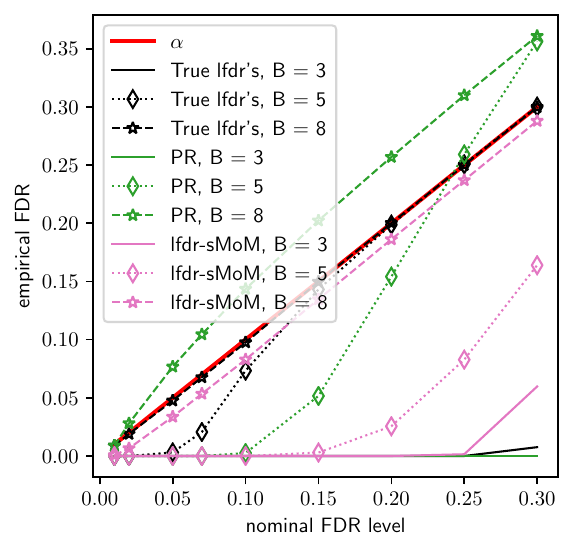}
				\end{subfigure}
				\begin{subfigure}{.48\textwidth}
					\centering
					% For double column
%					\includegraphics[scale=0.45]{figures/sc18_quan_pow.pdf}
					% For single column
					\includegraphics[scale=0.42]{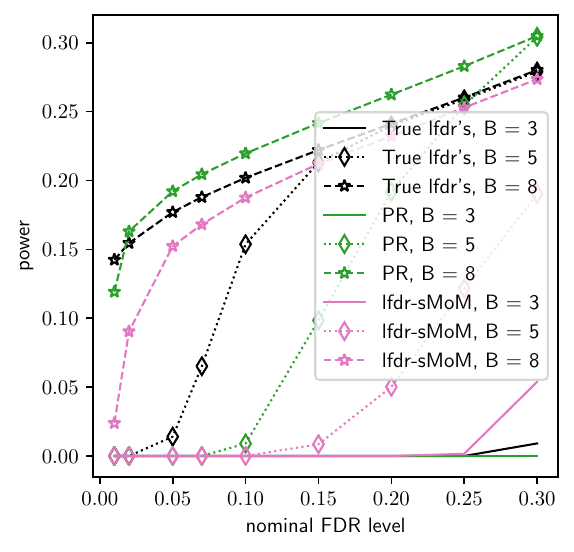}
				\end{subfigure}
				\caption{ScC}	
				\label{fig:scc-cnfg4}		
			\end{subfigure}
			\caption{\newMajor{\Gls{fdr} and detection power for all scenarios and a sensor network composed of $N=10\,000$ nodes and quantized $p$-values. Decisions are only made at the sensor locations (Cnfg.~1). lfdr-based spatial inference works also with $p$-values that are quantized using few bits. The FDR is strictly controlled with proposed lfdr estimator lfdr-sMoM.}}
			\label{fig:quan}
			\vspace{-13pt}
		\end{figure}
	
		\newMajor{We conclude the simulation results with a brief outlook on the effect that censoring \cite{Tay2007, Blum2007} has on the proposed approach. As discussed in Sec.~\ref{sec:cc}, allowing only sensors with quantized $p$-values $\leq \lambda$ to transmit can reduce the communication cost and power consumption considerably. In Fig.~\ref{fig:cen}, the empirical FDR and detection power for ScA and Cnfg.~4, i.e., $N=2\,000$ sensors are shown. The same quantizer as in the previous paragraph is used with $B = 5$ bits. We compare the lfdr-based spatial inference approach with true lfdr's and lfdr's estimated by our proposed lfdr-sMoM. If $\lambda = 1$, no censoring occurs. There is no performance difference for the considered values of $\lambda$ with the true lfdr's. However, this may be different for even smaller $\lambda$ or other scenarios, where $\lambda$ would censor $p$-values that would lead to a discovery at the \gls{fc}. The performance with lfdr-sMoM is close to optimum at standard nominal FDR levels for all $\lambda$. This is remarkable, as only about $1/3$ of all sensors transmit their $p$-values to the \gls{fc} for $\lambda=.12$. This offers great savings in communication cost and power consumption. The analytical derivation of the censoring regions remains a future research topic.}
		\begin{figure}
			\centering
			% For double column!
			\begin{subfigure}{.24\textwidth}
			% For single column!
%			\begin{subfigure}{.3\textwidth}
				\centering
				\includegraphics[scale=0.45]{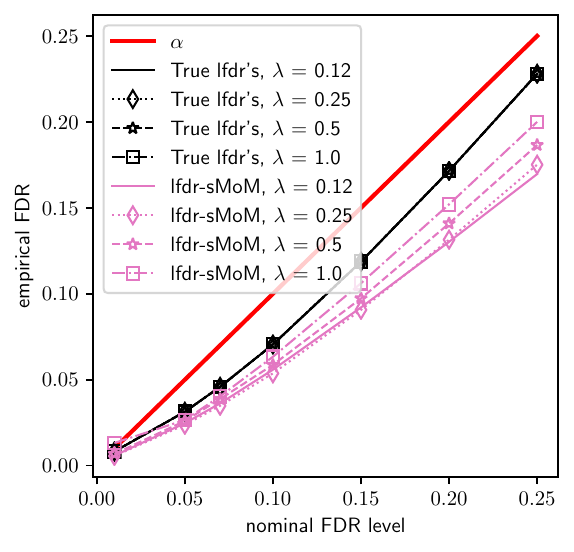}
				%\caption{\gls{fdr}}
			\end{subfigure}
			% For double column!
			\begin{subfigure}{.24\textwidth}
			% For single column!
%			\begin{subfigure}{.3\textwidth}
				\centering
				\includegraphics[scale=0.45]{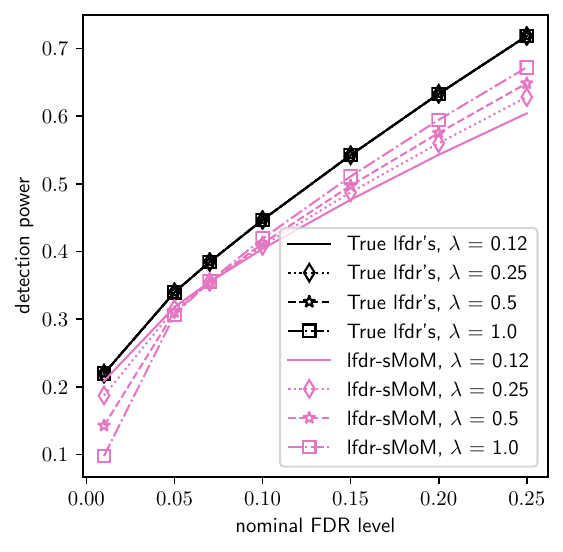}
			\end{subfigure}
			\caption{\newMajor{FDR and detection power for ScA, Cnfg.~4 and different censoring thresholds $\lambda$. Censoring has little impact on the detection power while offering great potential to saving on communication cost and power consumption.}}
			\label{fig:cen}
			\vspace{-13pt}
		\end{figure}
%	\subsection{Accuracy and computation time with increasing network size}

	\vspace{-5pt}
\section{The inclusion of domain-specific knowledge}
	While we kept this work general to maintain its applicability to a wide area of practical problems, extending the procedure to explicitly account for the particular nature of the observed physical phenomenon, such as electromagnetic spectrum, air quality or agricultural fields appears promising. This could further increase the detection power. \newMinor{Due to the modular nature of the proposed inference method, improving the individual blocks in Fig.~\ref{fig:flowchart} can lead to improved results without the necessity to come up with an entirely new approach. For example, one can exploit the spatial smoothness assumption not only in the estimation of the probability models, but also to increase the detection power. One possible approach is to replace the \gls{lfdr} with the \gls{clfdr}. The difference is that $\nulFrc$ depends on the sensor location in the \gls{clfdr}. Its location-dependent value can be estimated from the data.} One could plug in an novel or existing method \cite{Halme2019, Tansey2018} \old{to replace the constant $\nulFrc$ by a suitable spatially varying prior}in\newMinor{to} Eq.~\eqref{eq:lfdr-def}\old{ when calculating the sensor \gls{lfdr}s}. \old{Also, }\newMinor{Alternatively, one could include a penalty term into }\gls{rbf} interpolation\old{ allows to include a penalty} to spatially \old{smooth}\newMinor{smoothen} the interpolant \cite{Schaback2007}. Choosing the value of this parameter is non-trivial and should be done in an application-dependent manner.

\vspace{-5pt}
\section{Summary}
	We %investigated the general spatial inference problem of identifying the region of anomaly of a spatial phenomenon in an area of observation. We embedded the problem into the framework of large-scale multiple hypothesis testing by means of local false discovery rates.
	proposed \newMinor{a novel lfdr-based inference method}\old{ a general spatial inference method} for detecting interesting, different or anomalous regions \old{for}\newMinor{of} an observed physical phenomenon. Our approach provides statistical performance guarantees in terms of false positives. %Its computational efficiency makes it particularly suitable to solving real-world problems with the help of heterogeneous large-scale sensor networks, such as they are used in the \gls{iot}, for example. 
	The proposed method facilitates solving real-world spatial inference problems in which distributed heterogeneous large-scale sensor networks observe spatial phenomena. \old{A prime example is \newMinor{the} \gls{iot}. }The sensors communicate with a fusion center or cloud only in a limited fashion. This saves battery and ensures a long sensor life-span. \newMinor{Moreover, the method estimates the local false discovery rates in between the actual sensor locations. Consequently, the proposed method allows for identifying spatial regions where alternative hypotheses are in place while strictly controlling the FDR. The decision making takes place at the fusion center or in the cloud. In addition, we proposed a novel approach to estimating the local false discovery based on the spectral method of moments}. It outperforms existing methods in terms of detection power and runtime in a variety of scenarios.  It is considerably faster than state-of-the-art methods \newMajor{and yields close-to-optimum results for $p$-values that are quantized using few bits. This is important for inference with wireless sensor networks, since the sensors communicate their information over the wireless communication channel.}
	\old{Our individual contributions are as follows. We formulated the problem using a large-scale multiple hypothesis framework where the false discovery rate is controlled.	}%The observation area was monitored by a large-scale sensor network. % with limited sensor-to-sensor and sensor-to-cloud communication.
	\old{We developed a novel method to \old{compute}\newMinor{estimate the} local false discovery rates in a data-driven fashion based on \old{quantities estimated with} the spectral method of moments. It outperforms existing methods in terms of detection power and runtime in a variety of scenarios. We proposed a method for interpolating local false discovery rate values at locations where no sensor is present. This allows for identifying regions associated with null hypothesis and alternative, as well as estimating their boundaries.}
	The performance was evaluated by an application to spatially varying radio frequency waves. %In addition, the simulation results demonstrate the superior performance of the proposed method in large-scale sensor networks, such as used for example in the \gls{iot}.
	 % suggested to interpolate the \gls{lfdr}s to perform inference everywhere in the observation area.
	The code to reproduce the results is available on \url{https://github.com/mgoelz95/lfdr-sMoM}.
	\appendices
\vspace{-6pt}
	\section{Proof of Theorem~\ref{theo:relation-pop-mom}}
	\label{apx:smom-theo}
		We establish the one-to-one relations between the first three population moments $\E{\pvecRV}, \E{\pvecRV\tensorprod\pvecRV}, \E{\pvecRV\tensorprod\pvecRV\tensorprod\pvecRV}$ and the model parameter-dependent operators $\secondPopSMoMTheo\in\mathbb{R}^{\numVecEl\times\numVecEl}$, $\thirdPopSMoMTheo\in\mathbb{R}^{\numVecEl\times\numVecEl\times\numVecEl}$, which are given in Eq.~\eqref{eq:theo2_M2} and Eq.~\eqref{eq:theo2_M3}. Since the sample moments are consistent estimates of the population moments, Theorem~\ref{theo:estimators-mix-par} enables model parameter estimation based on the sample moments.
		\vspace{-9pt}
		\subsection{Proof of Eq.~\eqref{eq:theo2_M2}}
			Under Assumption~\ref{ass:same-var} with $\numVecEl > \numCmp$, the \old{$\numVecEl - (\numCmp-1)$ smallest eigenvalues of covariance matrix $\vecCov\in\mathbb{R}^{\numVecEl\times\numVecEl}$ with $\mathrm{rank}(\vecCov)= \min(\numCmp, \numVecEl)-1$,}\newMinor{covariance matrix $\vecCov\in\mathbb{R}^{\numVecEl\times\numVecEl}$ can be written as}
			\begin{align}
				\vecCov &= \E{(\pvecRV-\vecExp)\tensorprod(\pvecRV - \vecExp)}\nonumber\\
						%&= \E{\pvec\pvec^\top - \vecExp\pvec^\top - \pvec\vecExp^\top + \vecExp\vecExp^\top}\\
						%&= \E{\pvec\pvec^\top} - \vecExp\E{\pvec^\top} - \E{\pvec}\vecExp^\top + \vecExp\vecExp^\top\\
						%&= \sum_{\cmpIdx = 1}^\numCmp\Big(\cmpWgt\E[|\cmpIdx]{\pvec\pvec^\top}\Big) - \vecExp\vecExp^\top\\
						%&= \sum_{\cmpIdx = 1}^\numCmp\cmpWgt \bigg(\cmpVar\mathbf{I} + \vecCmpExp\Big(\vecCmpExp\Big)^\top\bigg) - \vecExp\vecExp^\top\\
						%&= \sum_{\cmpIdx = 1}^\numCmp\cmpWgt\bigg(\Big(\vecCmpExp - \vecExp\Big)\tensorprod\Big(\vecCmpExp-\vecExp\Big)+\cmpVar\mathbf{I}\bigg)\\
						&= \avVar\mathbf{I}_{\numVecEl\times\numVecEl} + \sum_{\cmpIdx = 1}^\numCmp\cmpWgt\Big(\vecCmpExp-\vecExp\Big)\tensorprod\Big(\vecCmpExp-\vecExp\Big),\label{eq:apx_cov}
			\end{align}
			\newMinor{where $\sum_{\cmpIdx = 1}^\numCmp\cmpWgt\Big(\vecCmpExp-\vecExp\Big)\tensorprod\Big(\vecCmpExp-\vecExp\Big)$ is a $\numVecEl\times\numVecEl$ matrix of rank $\numCmp-1$. Consequently, the $\numVecEl-(\numCmp-1)$ smallest eigenvalues of $\vecCov$} are all equal to $\avVar$. $\noiseEigVec\in\mathbb{R}^\numVecEl$ is any of the unit norm eigenvectors corresponding to $\avVar$. The derivation for Eq.~\eqref{eq:apx_cov} in the proof of \cite[Theorem~1]{Hsu2012} applies also to our data model. Thus, also
			
			\begin{align*}
				\secondPopSMoMTheo 	&= \E{\pvecRV\tensorprod\pvecRV} - \avVar\mathbf{I}_{\numVecEl\times\numVecEl}
%				&= \sum_{\cmpIdx = 1}^\numCmp \cmpWgt\Big(\E[|\cmpIdx]{\pvec\pvec^\top}- \cmpVar\mathbf{I}\Big)\\
%				&= \sum_{\cmpIdx = 1}^\numCmp \cmpWgt\Big(\cmpVar\mathbf{I} + \vecCmpExp\big(\vecCmpExp\big)^\top- \cmpVar\mathbf{I}\Big)\\
				= \sum_{\cmpIdx = 1}^\numCmp\cmpWgt\vecCmpExp\tensorprod\vecCmpExp
			\end{align*}
			holds and relates the observable second population moment $\E{\pvecRV\tensorprod\pvecRV}$ to the model parameters.
		\vspace{-9pt}
		\subsection{Proof of Eq.~\eqref{eq:theo2_M3}}
			We now proof the relation $\thirdPopSMoMTheo = \thirdPopSMoM- \thirdPopSMoMErr$ between the observable $\thirdPopSMoM$ and non-observable $\thirdPopSMoMTheo$. Random vector $\pvecRV\in\mathbb{R}^\numVecEl$ is assumed to follow a $\numVecEl$-variate, $\numCmp$-component mixture model. It can hence be described by $\pvecRV = \pvecRV^{(\mathsf{\cmpIdx})}$, where $\mathsf{k}\in[\numCmp]$ is a discrete random variable for the mixture component index taking value $\mathsf{\cmpIdx} = \cmpIdx$ with probability $\cmpWgt \in[0, 1]$ and $\sum_{\cmpIdx=1}^\numCmp\cmpWgt = 1$. Let $\qvecRV^{(\mathsf{k})} = \pvecRV^{(\mathsf{\cmpIdx})} - \vecCmpExp[\mathsf{\cmpIdx}]$ denote the centered data random vector and conditioned on $\mathsf{\cmpIdx} = \cmpIdx$, $\qvecRV^{(\cmpIdx)} = \pvecRV^{(\cmpIdx)} - \vecCmpExp$. This implies $\E{\qvecRV^{(\cmpIdx)}} = \mathbf{0}_\numVecEl$. Since $\qvecRV\in\mathbb{R}^\numVecEl$, its $\vecElIdx$-th marginal is $\q^{(\cmpIdx)}$, $\vecElIdx\in[\numVecEl]$, and the marginal variances are $\E{{\q^{(\cmpIdx)}}^2} = \cmpVar\,\forall\,\cmpIdx\in[\numCmp]$ under Assumption~\ref{ass:same-var}. The third order central moments are $\E{{\q^{(\cmpIdx)}}^3} = \elCmpThrdCenMom{\cmpIdx}$.  %conditioned on the component random variable $\mathsf{\cmpIdx}$for component $\cmpIdx$, i.e., for $\pvec$ conditioned on originating in component $\cmpIdx$, $\qvec= \pvec - \vecCmpExp$ with $\E[|\cmpIdx]{\qvec} = \mathbf{0}$. %such that $\qvec = \pvec - \vecExp$ with $\E{\qvec} = \mathbf{0}$ and such that for $\pvec$ conditioned on originating in component $\cmpIdx$, $\qvec= \pvec - \vecCmpExp$ to ensure $\E[|\cmpIdx]{\qvec} = \mathbf{0}$.
			
			%We prove Eq.~\eqref{eq:theo2_M3} by relating  deriving all involved quantities $\thirdPopSMoM, \thirdPopSMoMTheo, \thirdPopSMoMErr, \thirdPopSMoMErrFirst, \thirdPopSMoMErrThird\in\mathbb{R}^{\numVecEl\times\numVecEl\times\numVecEl}, \firstPopSMoM, \firstPopSMoMTheo\in\mathbb{R}^\numVecEl$ and the third order population moment $\E{\pvecRV\tensorprod\pvecRV\tensorprod\pvecRV}\in\mathbb{R}^{\numVecEl\times\numVecEl\times\numVecEl}$,
			The third order population moment tensor $\E{\pvecRV\tensorprod\pvecRV\tensorprod\pvecRV}$ is % \in\mathbb{R}^{\numVecEl\times\numVecEl\times\numVecEl}$ is
			\vspace{-3pt}
			\begin{align}
				&\E{\pvecRV\tensorprod\pvecRV\tensorprod\pvecRV} \nonumber\\
				&= \sum_{\cmpIdx = 1}^\numCmp\cmpWgt\bigg(\!\E{\Big(\!\vecCmpExp\!+\!\qvecRV^{(\cmpIdx)}\!\Big)\!\tensorprod\!\Big(\!\vecCmpExp\!+\!\qvecRV^{(\cmpIdx)}\!\Big)\!\tensorprod\!\Big(\!\vecCmpExp\!+ \!\qvecRV^{(\cmpIdx)}\!\Big)}\!\bigg)\nonumber\\
				&= \sum_{\cmpIdx = 1}^\numCmp\!\cmpWgt\!\bigg(\!\E{\vecCmpExp\!\tensorprod\vecCmpExp\!\tensorprod\vecCmpExp} + \E{\vecCmpExp\!\tensorprod\qvecRV^{(\cmpIdx)}\tensorprod\qvecRV^{(\cmpIdx)}}\nonumber\\
				&\qquad\qquad\quad +\!\E{\qvecRV^{(\cmpIdx)}\tensorprod\vecCmpExp\!\tensorprod\qvecRV^{(\cmpIdx)}}\!+\!\E{\qvecRV^{(\cmpIdx)}\tensorprod\qvecRV^{(\cmpIdx)}\tensorprod\vecCmpExp}\nonumber\\
				&\qquad\qquad\quad + \!\E{\qvecRV^{(\cmpIdx)}\tensorprod\qvecRV^{(\cmpIdx)}\tensorprod\qvecRV^{(\cmpIdx)}}\!\bigg)\nonumber\\
				&=\sum_{\cmpIdx = 1}^\numCmp\!\cmpWgt\!\bigg(\!\vecCmpExp\!\tensorprod\vecCmpExp\!\tensorprod\vecCmpExp\!+\! \sum_{\vecElIdx=1}^\numVecEl\Big(\vecCmpExp\cmpVar\tensorprod \baseVec\tensorprod\baseVec\nonumber\\
				&\qquad\qquad\quad +\baseVec\tensorprod\vecCmpExp\cmpVar\tensorprod\baseVec+\baseVec\tensorprod\baseVec\tensorprod\vecCmpExp\cmpVar\nonumber\\
				&\qquad\qquad\quad+\elCmpThrdCenMom{\cmpIdx}\baseVec \!\tensorprod\!\baseVec\!\tensorprod\!\baseVec\Big)\bigg)\nonumber\\
				&=\thirdPopSMoMTheo + \thirdPopSMoMErrThird + \sum_{\vecElIdx = 1}^\numVecEl\bigg(\sum_{\cmpIdx = 1}^\numCmp\cmpWgt\Big(\vecCmpExp\cmpVar\tensorprod \baseVec\tensorprod\baseVec\nonumber\\
				&\qquad\quad+\baseVec\tensorprod\vecCmpExp\cmpVar\tensorprod\baseVec+\baseVec\tensorprod\baseVec\tensorprod\vecCmpExp\cmpVar\Big)\bigg)\label{eq:apx_M3}.
			\end{align}
			The definitions of $\thirdPopSMoM, \thirdPopSMoMErrThird \in\mathbb{R}^{\numVecEl\times\numVecEl\times\numVecEl}$ are provided in Eq.~\eqref{eq:theo1_M3} and Eq.~\eqref{eq:theo2_M3_delta3}. We proceed with $\firstPopSMoM\in\mathbb{R}^\numVecEl$,
			\begin{align}
				\firstPopSMoM 	&= \E{\pvecRV\Big(\noiseEigVec^\top\big(\pvecRV-\vecExp\big)\Big)^2}\nonumber\nonumber\\
								&= \sum_{\cmpIdx = 1}^\numCmp \cmpWgt\E{\pvecRV^{(\cmpIdx)}\Big(\noiseEigVec^\top\big(\pvecRV^{(\cmpIdx)}-\vecExp\big)^2\Big)}\nonumber\\
								&= \sum_{\cmpIdx = 1}^\numCmp \cmpWgt\E{\bigg(\vecCmpExp + \qvecRV^{(\cmpIdx)}\bigg)\bigg(\noiseEigVec^\top\Big(\vecCmpExp-\vecExp+\qvecRV^{(\cmpIdx)}\Big)\bigg)^2}\nonumber\\
								&= \sum_{\cmpIdx = 1}^\numCmp \cmpWgt\E{\Big(\vecCmpExp + \qvecRV^{(\cmpIdx)}\Big)\Big(\noiseEigVec^\top\qvecRV^{(\cmpIdx)}\Big)^2}\nonumber\\
								&= \sum_{\cmpIdx = 1}^\numCmp \cmpWgt\bigg(\E{\vecCmpExp\Big(\noiseEigVec^\top\qvecRV^{(\cmpIdx)}\Big)^2} +\E{\qvecRV^{(\cmpIdx)}\Big(\noiseEigVec^\top\qvecRV^{(\cmpIdx)}\Big)^2}\bigg)\nonumber\\
								&= \sum_{\cmpIdx = 1}^\numCmp\cmpWgt\Big(\vecCmpExp\cmpVar + \noiseEigVec\odot\noiseEigVec\odot\vecCmpThrdOrdCenMom\Big).\label{eq:apx_m1}
			\end{align}
			The transition from line 3 to 4 in Eq.~\eqref{eq:apx_m1} uses the property that $\noiseEigVec$ lies in the null space of $\cmpWgt\Big(\vecCmpExp-\vecExp\Big)\tensorprod\Big(\vecCmpExp-\vecExp\Big)$. With  $\E{\q[\vecElIdx_1]^{(\cmpIdx)}\q[\vecElIdx_2]^{(\cmpIdx)}} = 0$ if $\vecElIdx_1 \neq \vecElIdx_2$ due to the diagonal covariance matrix, the first addend in line 6 follows from line 5 by
			
			\begin{align*}
				&\E{\vecCmpExp\Big(\noiseEigVec^\top\qvecRV^{(\cmpIdx)}\Big)^2} = \vecCmpExp\Bigg(\!\E{\sum_{\vecElIdx = 1}^\numVecEl\bigg(\noiseEigVecEl^2{\q^{(\cmpIdx)}}^2\bigg)}\!\\
				&\qquad+\sum_{\vecElIdx_1 = 1}^\numVecEl\sum_{\vecElIdx_2 = \vecElIdx_1 + 1}^{\numVecEl}\!2\E{\q[\vecElIdx_1]^{(\cmpIdx)}\q[\vecElIdx_2]^{(\cmpIdx)}}\noiseEigVecEl[\vecElIdx_1]\noiseEigVecEl[\vecElIdx_2]\!\Bigg)\nonumber.
			\end{align*}
				For the second addend in Eq.~\eqref{eq:apx_m1},
			\begin{align*}
				&~= \vecCmpExp\sum_{\vecElIdx = 1}^\numVecEl\noiseEigVecEl^2\E{{\q^{(\cmpIdx)}}^2}\nonumber= \vecCmpExp\sum_{\vecElIdx = 1}^\numVecEl\noiseEigVecEl^2\cmpVar\nonumber
																			%= \vecCmpExp\cmpVar||\noiseEigVec||^2\\
				%&~
				= \vecCmpExp\cmpVar\nonumber,\\
				&\E{\qvecRV^{(\cmpIdx)}\Big(\noiseEigVec^\top \qvecRV^{(\cmpIdx)}\Big)^2} = \E{\qvecRV^{(\cmpIdx)}\sum_{\vecElIdx = 1}^\numVecEl\Big(\noiseEigVecEl^2{\q^{(\cmpIdx)}}^2\Big)}\\
				&\qquad+ \E{\qvecRV^{(\cmpIdx)}\sum_{\vecElIdx_1 = 1}^\numVecEl\sum_{\vecElIdx_2 =\vecElIdx_1+1 }^{\numVecEl}\big(2\q[\vecElIdx_1]^{(\cmpIdx)}\q[\vecElIdx_2]^{(\cmpIdx)}\noiseEigVecEl[\vecElIdx_1]\noiseEigVecEl[\vecElIdx_2]\big)}.
																		%\E{\qvecRV^{(\cmpIdx)}\bigg(\sum_{\cmpIdx = 1}^\numCmp\big(\noiseEigVecEl^2\q^2\big)\\+\sum_{\vecElIdx_1 = 1}^\numVecEl\sum_{\vecElIdx_2 = \vecElIdx_1+1}^\numVecEl\big(2\q[\vecElIdx_1]\q[\vecElIdx_2]\noiseEigVecEl[\vecElIdx_1]\noiseEigVecEl[\vecElIdx_2]\big)\bigg)}\\
			\end{align*}
			For the $j$-th entry, $j\in[\numVecEl]$, of vector $\E{\qvecRV^{(\cmpIdx)}\Big(\noiseEigVec^\top \qvecRV^{(\cmpIdx)}\Big)^2}$, %$\Bigg[\E{\qvecRV^{(\cmpIdx)}\Big(\noiseEigVec^\top \qvecRV^{(\cmpIdx)}\Big)^2}\Bigg]_{\vecElIdx^\ast}, \vecElIdx^\ast\in[\numVecEl]$, 
			
			\begin{align*}
				&\E{\q[j]^{(\cmpIdx)}\sum_{\vecElIdx=1}^\numVecEl\Big(\noiseEigVecEl^2\q^{(\cmpIdx)}\Big)} = \E{\q[j]^{(\cmpIdx)}\Big(\noiseEigVecEl[j]^2{\q[j]^{(\cmpIdx)}}^2\Big)} = \noiseEigVecEl[j]^2\E{{\q[j]^{(\cmpIdx)}}^3},\\
				&\E{\q[j]^{(\cmpIdx)}\sum_{\vecElIdx_1 = 1}^\numVecEl\sum_{\vecElIdx_2 = \numVecEl_1+1}^\numVecEl\big(2\q[\vecElIdx_1]^{(\cmpIdx)}\q[\vecElIdx_2]^{(\cmpIdx)}\noiseEigVecEl[\vecElIdx_1]\noiseEigVecEl[\vecElIdx_2]\big)} = 0,
			\end{align*}											
			and since $\E{{\q[j]^{(\cmpIdx)}}^3} = \elCmpThrdCenMom[j]{\cmpIdx}$,
			\begin{align*}
				\E{\qvec\Big(\noiseEigVec^\top \qvec\Big)^2} &= \noiseEigVec\odot\noiseEigVec\odot\vecCmpThrdOrdCenMom.
			\end{align*}
			The proof is concluded by subtracting $\firstPopSMoM$ as in Eq.~\eqref{eq:apx_m1} from $\E{\pvecRV\tensorprod\pvecRV\tensorprod\pvecRV}$ in Eq.~\eqref{eq:apx_M3}, which results in Eq.~\eqref{eq:theo2_M3} when inserting $\thirdPopSMoMErrFirst\in\mathbb{R}^{\numVecEl\times\numVecEl\times\numVecEl}$ for $\sum_{\vecElIdx = 1}^\numVecEl\firstPopSMoMErr\tensorprod\baseVec\tensorprod\baseVec + \baseVec\tensorprod\firstPopSMoMErr\tensorprod\baseVec + \baseVec\tensorprod\baseVec\tensorprod\firstPopSMoMErr$ with $\firstPopSMoMErr\coloneqq\sum_{\cmpIdx = 1}^\numCmp \noiseEigVec\odot\noiseEigVec\odot\vecCmpThrdOrdCenMom$.
\vspace{-9pt}
	\section{Approximation of the third order moment}
		\label{apx:smom-conv}
		In this section, we demonstrate that approximating the non-observable $\thirdPopSMoMTheoVec{\boldsymbol{\eta}}$ by the observable $\thirdPopSMoMVec{\boldsymbol{\eta}} \in\mathbb{R}^{\numVecEl\times\numVecEl}$ is accurate up to an approximation error $\thirdPopSMoMVec{\boldsymbol{\eta}} - \thirdPopSMoMTheoVec{\boldsymbol{\eta}}$, which is negligible under certain conditions. % with $\boldsymbol{\eta} = [\eta_1, \dots, \eta_\numVecEl]^\top\in\mathbb{R}^\numVecEl$.
		By definition,
		
		\begin{align*}
			\thirdPopSMoMVec{\boldsymbol{\eta}}
			&= \sum_{\vecElIdx = 1}^\numVecEl\sum_{\vecElIdxAltI = 1}^\numVecEl\sum_{\vecElIdxAltII = 1}^\numVecEl\big[\thirdPopSMoM\big]_{\vecElIdx, \vecElIdxAltI, \vecElIdxAltII}\eta_\vecElIdxAltII\baseVec[\vecElIdx]\!\tensorprod\baseVec[\vecElIdxAltI]\\
			&= \sum_{\vecElIdx = 1}^\numVecEl\sum_{\vecElIdxAltI = 1}^\numVecEl\sum_{\vecElIdxAltII = 1}^\numVecEl\Big(\big[\thirdPopSMoMTheo\big]_{\vecElIdx, \vecElIdxAltI, \vecElIdxAltII} + \big[\thirdPopSMoMErr\big]_{\vecElIdx, \vecElIdxAltI, \vecElIdxAltII}\Big)\eta_\vecElIdxAltI\baseVec[\vecElIdx]\!\tensorprod\baseVec[\vecElIdxAltI].
		\end{align*}
		The entries are thus $\big[\thirdPopSMoMVec{\boldsymbol{\eta}}\big]_{\vecElIdx,\vecElIdxAltI} = \sum_{\vecElIdxAltII= 1}^{\numVecEl}\big[\thirdPopSMoMTheo\big]_{\vecElIdx, \vecElIdxAltI, \vecElIdxAltII}{\eta}_\vecElIdxAltII + \sum_{\vecElIdxAltII= 1}^{\numVecEl}\big[\thirdPopSMoMErrThird\big]_{\vecElIdx, \vecElIdxAltI, \vecElIdxAltII}{\eta}_\vecElIdxAltII\,\forall\,\vecElIdx, \vecElIdxAltI \in[\numVecEl]$. If the value of the second sum is small $\forall\,\vecElIdx, \vecElIdxAltI \in[\numVecEl]$, we can conclude $\thirdPopSMoMVec{\boldsymbol{\eta}}\approx\thirdPopSMoMTheoVec{\boldsymbol{\eta}}$. With Theorem~\ref{theo:relation-pop-mom}, the entries of $\thirdPopSMoMErr$ are found to be
		\begin{align*}
			\big[\thirdPopSMoMErr\big]_{\vecElIdx, \vecElIdxAltI, \vecElIdxAltII}
			= \begin{cases}
				\big(-3\noiseEigVecEl[\vecElIdx]^2+1\big)\avElCmpThrdCenMom[\vecElIdx]	&  \vecElIdx = \vecElIdxAltI = \vecElIdxAltII,\\
				-\noiseEigVecEl[{\vecElIdx}]^2\avElCmpThrdCenMom[\vecElIdx] 			&  \vecElIdx\neq\vecElIdxAltI = \vecElIdxAltII,\\
				-\noiseEigVecEl[{\vecElIdxAltI}]^2\avElCmpThrdCenMom[\vecElIdxAltI] 			&  \vecElIdx=\vecElIdxAltII \neq \vecElIdxAltI,\\
				-\noiseEigVecEl[{\vecElIdxAltII}]^2\avElCmpThrdCenMom[\vecElIdxAltII] 			&  \vecElIdx=\vecElIdxAltI \neq \vecElIdxAltII,\\
				0 & \vecElIdx\neq\vecElIdxAltI \neq \vecElIdxAltII.
			\end{cases}
		\end{align*}
		Consequently, the element-wise error between $\thirdPopSMoM$ and $\thirdPopSMoMTheo$ is 
		\begin{align*}
			\big[\thirdPopSMoMVec{\eta} - \thirdPopSMoMTheoVec{\eta}\big]_{\vecElIdx, \vecElIdxAltI} &= \eta_{\vecElIdx}\avElCmpThrdCenMom[\vecElIdx] -2\eta_{\vecElIdx}\noiseEigVecEl[\vecElIdx]^2\avElCmpThrdCenMom[\vecElIdx]-\sum_{\vecElIdxAltII = 1}^\numVecEl\eta_{\vecElIdxAltII}\noiseEigVecEl[\vecElIdxAltII]^2\avElCmpThrdCenMom[\vecElIdxAltII]\\
			&=\eta_{\vecElIdx}\Big(1-2\noiseEigVecEl[\vecElIdx]^2\Big)\avElCmpThrdCenMom[\vecElIdx] - \sum_{\vecElIdxAltII = 1}^\numVecEl\eta_{\vecElIdxAltII}\noiseEigVecEl[\vecElIdxAltII]^2\avElCmpThrdCenMom[\vecElIdxAltII]
		\end{align*}
		if row index $\vecElIdx$ and column index $\vecElIdxAltI$ are equivalent, $\vecElIdx = \vecElIdxAltI$, and
		\begin{equation*}
			\big[\thirdPopSMoMVec{\eta} - \thirdPopSMoMTheoVec{\eta}\big]_{\vecElIdx, \vecElIdxAltI} = -\eta_{\vecElIdxAltI}\noiseEigVecEl[\vecElIdx]^2\avElCmpThrdCenMom[\vecElIdx]-\eta_{\vecElIdx}\noiseEigVecEl[\vecElIdxAltI]^2\avElCmpThrdCenMom[\vecElIdxAltI]
		\end{equation*}
		if $\vecElIdx\neq\vecElIdxAltI$, $\forall\,\vecElIdx,\vecElIdxAltI\in[\numVecEl]$.
		
		Recall the average third cumulant $\avElCmpThrdCenMom = \sum_{\cmpIdx = 1}^\numCmp\cmpWgt\thirdCenMoMMvCmp{\vecElIdx}$, with $\thirdCenMoMMvCmp{\vecElIdx}$ from Eq.~\eqref{eq:third-mom}.
		%$\avElCmpThrdCenMom$ takes its maximum value $\max(\avElCmpThrdCenMom) \approx 0.02644$, if $\betaMixA \approx 0.22219\,\forall\,\cmpIdx\in[\numCmp]$ and its minimum value $\min(\avElCmpThrdCenMom) \approx -0.00742$, if $\betaMixA \approx 2.0664\,\forall\,\cmpIdx\in[\numCmp]$.
		$\old{-0.0075}\newMinor{-0.00742}<\avElCmpThrdCenMom<\old{0.2222}\newMinor{0.02641}$, i.e., $\avElCmpThrdCenMom$ is bounded from below and above. \newMinor{These bounds are obtained by a worst-case assessment: Since for the weights $\sum_{k=1}^K w^{(k)}= 1$ holds, the worst case is that for all $k\in[K]$, $\kappa_{3_i}^{(k)}$ takes its minimal (maximal) value to find the lower (upper) bound on $\overline{\kappa}_3$. Differentiating the expression for $\kappa_{3_i}^{(k)}$ in Eq.~\eqref{eq:third-mom} w.r.t $a_i^{(k)}$ yields
			\begin{equation*}
				\frac{\mathrm{d}\kappa_{3_i}^{(k)}}{\mathrm{d}a_{i}^{(i)}} = \frac{3}{\Big(a_i^{(k)}+3\Big)^2} - \frac{3a_i^{(k)}\Big(3a_i^{(k)} + 4\Big)}{\Big({a_i^{(k)}}^{2}+3a_i^{(k)} + 2\Big)^2} + \frac{{6a_i^{(k)}}^2}{\Big(a_i^{(k)}+1\Big)^4}.
			\end{equation*}
			The two positive (since for the beta distribution parameters $a_i^{(k)}> 0$ must hold) roots of $\frac{\mathrm{d}\kappa_{3_i}^{(k)}}{\mathrm{d}a_{i}^{(i)}}$ are $a_i^{(k)} = 1/2(-1\pm\sqrt{2} + \sqrt{11-2\sqrt{2}})$. The corresponding bounds for $\kappa_{3_i}^{(k)}$ then follow by inserting into Eq.~\eqref{eq:third-mom}.}
		
		By definition, $|\noiseEigVec|^2= \sum_{\vecElIdx=1}^\numVecEl\noiseEigVecEl^2 = 1$, which implies $\noiseEigVecEl^2 \leq 1\,\forall\,\vecElIdx\in[\numVecEl]$.
		
		\begin{lemma}
			\label{lem:apx}
			Let $\boldsymbol{\eta}\in\mathbb{R}^\numVecEl$ be distributed uniformly on the unit sphere in $\mathbb{R}^\numVecEl$, which holds if its $\vecElIdx$-th element is $\eta_\vecElIdx = \frac{x_\vecElIdx}{\sqrt{x_1^2 + \dots + x_ \numVecEl^2}}$ with $x_\vecElIdx\sim\mathcal{N}(0, 1)\,\forall\,\vecElIdx\in[\numVecEl]$. For large $\numVecEl$, the $\eta_\vecElIdx$ become approximately i.i.d. Gaussian distributed, $\eta_\vecElIdx\sim\mathcal{N}(0, \numVecEl^{-1})$. If we generate $U$ samples of $\boldsymbol{\eta}$, we will find an observation $\boldsymbol{\eta}^\ast$ such that the errors $\big[\thirdPopSMoMVec{\eta} - \thirdPopSMoMTheoVec{\eta}\big]_{\vecElIdx, \vecElIdxAltI}$ become arbitrarily small $\forall\,\vecElIdx, \vecElIdxAltI\in[\numVecEl]$ as $U\rightarrow\infty$.
		\end{lemma}
		\begin{proof}
			The assumption of i.i.d. Gaussianity of the $\eta_i\,\forall\,\vecElIdx\in[\numVecEl]$ for large $\numVecEl$ is justified by the law of large numbers, since $\numVecEl^{-1}(x_1^2 + \dots + x_ \numVecEl^2) \rightarrow1$ for $\numVecEl\rightarrow\infty$. For the second part of the Lemma, we first assume that all $\numVecEl$ entries $\noiseEigVecEl$ are non-zero. Consequently, the fraction of the total mass $\sum_{\vecElIdx=1}^\numVecEl\noiseEigVecEl^2=1$ allocated to an individual entry $\noiseEigVecEl^2$ decreases with increasing $\numVecEl$ and the Lemma follows since $\E{\eta_i} = 0$ and $\E{\eta_i^2} = \numVecEl^{-1}$ $\forall\,\vecElIdx\in[\numVecEl]$. Assume now the opposite case, i.e., the mass $\sum_{\vecElIdx=1}^\numVecEl\noiseEigVecEl^2=1$ is concentrated in a few $\vecElIdx\in[\tilde{\numVecEl}], \tilde{\numVecEl}\ll\numVecEl$. Then, $\sum_{\vecElIdxAltII = 1}^\numVecEl\eta_{\vecElIdxAltII}\noiseEigVecEl[\vecElIdxAltII]^2\avElCmpThrdCenMom[\vecElIdxAltII] = \sum_{\vecElIdxAltII = 1}^{\tilde{\numVecEl}}\eta_{\vecElIdxAltII}\noiseEigVecEl[\vecElIdxAltII]^2\avElCmpThrdCenMom[\vecElIdxAltII]\sim\mathcal{N}\Big(0, \numVecEl^{-1}\sum_{\vecElIdxAltII= 1}^{\tilde{\numVecEl}}\big(\noiseEigVecEl[\vecElIdxAltII]^2\avElCmpThrdCenMom[\vecElIdxAltII]\big)^2\Big)$ due to the summation of normally distributed random variables and since $\sum_{\vecElIdxAltII= 1}^{\tilde{\numVecEl}}\noiseEigVecEl[\vecElIdxAltII]^2\avElCmpThrdCenMom[\vecElIdxAltII]\ll\tilde{\numVecEl}$, the claim of the Lemma follows.% $\forall\,\vecElIdx, \vecElIdxAltI\in[\numVecEl]$.
		\end{proof}
		In practice, small but non-zero $\big[\thirdPopSMoMVec{\eta} - \thirdPopSMoMTheoVec{\eta}\big]_{\vecElIdx, \vecElIdxAltI}\,\forall\,\vecElIdx, \vecElIdxAltI\in[\numVecEl]$ are sufficient. To this end, observe that the impact of a non-zero approximation error also depends on the true value of the mixture model parameters, since they determine the model error's \textit{relative} magnitude. In addition, the theoretical $\thirdPopSMoM$ has to be estimated from the data based on sample moments, which introduces an additional non-zero empirical error. This error often has a larger influence on the results than the mismatch between $\thirdPopSMoMEstVec{\boldsymbol{\eta}}$ and $\thirdPopSMoMTheoVec{\boldsymbol{\eta}}$.
		
		In our proposed method, we account for Lemma~\ref{lem:apx} by making $\numVecEl$ a tuning parameter whose value is chosen such that the resulting \pval~mixture \gls{pdf} $\mixPPdfEst$ is close to the data. Our simulation results underline that even for small $U$ and $\numVecEl$, using $\thirdPopSMoMVec{\boldsymbol{\eta}}$ instead of $\thirdPopSMoMTheoVec{\boldsymbol{\eta}}$ is reasonable to enable \gls{lfdr} estimation by the spectral method of moments.
	\printbibliography[title = {References}]

\end{document}